\documentclass[11pt]{article}
\usepackage{graphicx} 
\usepackage[T1]{fontenc}
\usepackage{mathpazo,bbm}
\DeclareSymbolFont{calletters}{OMS}{cmsy}{m}{n}
\DeclareSymbolFontAlphabet{\mathcal}{calletters}

\usepackage{todonotes}

\usepackage{pdfpages}

\usepackage{amssymb,amsmath,amsthm, bbm}
\usepackage{thm-restate}
\usepackage[dvipsnames]{xcolor}
\usepackage{graphicx}
\usepackage{setspace}
\usepackage{xspace}
\usepackage{multirow}
\usepackage{array}
\usepackage{complexity}
\usepackage[linesnumbered,ruled,vlined]{algorithm2e}
\SetKwInput{KwInput}{Input}                
\SetKwInput{KwOutput}{Output}
\SetKwFor{RepTimes}{repeat}{times}{end}
\SetKwComment{Comment}{/* }{ */}

\usepackage{dirtytalk}
\usepackage{tikz}
\usepackage[hypertexnames=false]{hyperref}
\usepackage[capitalize,nameinlink]{cleveref}

\usepackage[shortlabels]{enumitem}

\usepackage[margin=1in]{geometry}
\usepackage{braket}
\usepackage{ wasysym }
\usepackage{mathtools}

\hypersetup{colorlinks=true,urlcolor={MidnightBlue},linkcolor={MidnightBlue},
citecolor={MidnightBlue}}

\theoremstyle{plain}
\newtheorem{theorem}{Theorem}[section]
\newtheorem{corollary}[theorem]{Corollary}
\newtheorem{proposition}[theorem]{Proposition}
\newtheorem{lemma}[theorem]{Lemma}

\newtheorem{definition}[theorem]{Definition}

\theoremstyle{remark}
\newtheorem{remark}[theorem]{Remark}
\theoremstyle{plain}
\newclass{\DNF}{DNF}
\newclass{\DNFs}{DNFs}
\newclass{\ACzero}{AC^0}
\newclass{\TCzero}{TC^0}
\newclass{\Logspace}{L}

\newcommand{\F}{\mathbb{F}} 
\renewcommand{\C}{\mathbb{C}} 
\newcommand{\Z}{\mathbb{Z}} 

\renewcommand{\Pr}{\mathop{\bf Pr\/}}

\newcommand{\abs}[1]{\left|#1\right|}

\newcommand{\bigparen}[1]{\left ( #1 \right )}
\newcommand{\ceil}[1]{\lceil #1 \rceil}

\newcommand{\floor}[1]{\lfloor #1 \rfloor}

\newcommand{\norm}[1]{\| #1 \| }
\newcommand{\bignorm}[1]{\Bigl \| #1 \Bigr \| }

\newcommand{\curly}[1]{\{#1\}}

\newfunc{\MAJ}{MAJ}
\newfunc{\pmaj}{promise\text{-}MAJ}
\newfunc{\MUX}{MUX}
\newfunc{\NAE}{NAE}
\newfunc{\OR}{OR} 
\newfunc{\AND}{AND}
\newfunc{\XOR}{XOR}
\newfunc{\Tribes}{Tribes}
\newfunc{\LocalCorrect}{LocalCorrect}

\newfunc{\sgn}{sgn} 

\newfunc{\spar}{sparsity}
\newfunc{\rank}{rank}
\newfunc{\spn}{span}
\newfunc{\basis}{B}
\newfunc{\quasipoly}{quasipoly}
\newfunc{\Bias}{Bias}

\newfunc{\DT}{DTdepth} 
\newfunc{\DTs}{DTsize} 

\newcommand{\tildeO}{{\widetilde O}}

\renewcommand{\tilde}{\widetilde}

\newcommand{\bits}{\{0,1\}}

\newcommand{\calA}{\mathcal{A}}

\newcommand{\calC}{\mathcal{C}}
\newcommand{\calD}{\mathcal{D}}

\newcommand{\calF}{\mathcal{F}}

\newcommand{\calI}{\mathcal{I}}

\newcommand{\calO}{\mathcal{O}}

\newcommand{\calW}{\mathcal{W}}

\newfunc{\Parity}{PARITY}

\newfunc{\Dict}{Dict}
\newfunc{\Corr}{Corr}
\newfunc{\avg}{avg}
\newfunc{\smooth}{smooth}
\newfunc{\dist}{\calD}

\newclass{\ETH}{ETH}

\renewcommand{\epsilon}{\varepsilon}

\newcommand{\indicator}{\mathbbm{1}}

\newcommand{\tr}{\mathrm{Tr}} 

\newfunc{\DISJ}{DISJ}
\newfunc{\search}{SEARCH}

\newfunc{\Th}{Th}
\newfunc{\coll}{Coll}

\newcommand{\pisucc}{\Pi_{\mathrm{succ}}}
\newcommand{\maxmix}{\pi_{\textrm{mix}}}

\newcommand{\easierdp}[3]{#1^{(#2, #3)}}

\newcommand{\range}{\mathsf{Y}}
\newcommand{\domain}{\mathsf X}

\newcommand{\oracle}{\calO}
\newcommand{\alg}{\calA}

\newcommand{\madv}{\textsc{Madv}}

\newcommand{\relation}{\mathsf R}
\newcommand{\inputspace}{\mathsf D}

\newcommand{\outputspace}{Y}
\newcommand{\inputreg}{\calI}
\newcommand{\workreg}{\calW}

\newcommand{\domainsize}{\abs{\inputspace}}
\newcommand{\identity}{\,\mathbb{I}}

\newcommand{\op}[2]{\ket{#1}\!\bra{#2}}
\newcommand{\ip}[2]{\langle #1 | #2 \rangle}
\newcommand{\pibad}{\Pi_\text{bad}}

\newcommand{\addAdvMat}{\Gamma^{\pm}}
\newcommand{\multAdvMat}{\Gamma}

\newcommand{\progress}{W}

\newcommand{\addProgress}{\delta^\pm}
\newcommand{\addStop}{\tau^\pm}

\newcommand{\fixedtargets}{\mathcal{S}}

\newcommand{\Shaltiel}{\textrm{Sh}}

\newcommand{\crefdefpart}[2]{%
  \hyperref[#2]{\namecref{#1}~\labelcref*{#1},~\ref*{#2}}%
}

\newcommand{\blank}[1]{}

\newcommand{\negAdvBound}{\textsc{Adv}^{\pm}}

\newcommand{\easierlistdp}[3]{#1^{L_{(#2, #3)}}}

\newcommand{\listofout}{\mathcal{L}}

\newcommand{\listbadspace}{V'_{\text{bad}}}

\newcommand{\lambdapmtemp}{\lambda'}
\newcommand{\vsub}[1]{v_{#1}}

\title{Strong Selective and List-Decoding Direct Product Theorems for Quantum Query Complexity 
}
\author{
    Paul Beame\thanks{Research supported by NSF grant  CCF-2422205.}\\Computer Science \& Engineering\\University of Washington \
    \and 
    Niels Kornerup\thanks{Research supported by the LDRD Program at Sandia National Laboratories. Sandia is managed and operated by
NTESS under DOE NNSA contract DE-NA0003525.}\\ Sandia National Laboratories
    \and
    Michael Whitmeyer$^*$
    \\Computer Science \& Engineering\\University of Washington
}
\date{September 2026}

\allowdisplaybreaks

\begin{document}

\maketitle
\thispagestyle{empty}
\begin{abstract}
    Strong direct-product theorems for the quantum query complexity of specific functions have been known for nearly two decades.
    These were extended, using various versions of the quantum multiplicative adversary method, to broader classes of functions and relations, culminating in a strong direct-product theorem for the quantum query complexity of computing arbitrary functions and, more generally, for the quantum query complexity of state generation.
    The proofs of these culminating results use a version
    of the multiplicative adversary method that does not seem applicable to the computation of general relations. 
    Standard strong direct-product theorems apply when the quantum algorithm is required to produce answers to every question it is given and its correctness is measured with respect to all of them, though an equivalent generalization, threshold direct-product
    theorems, only requires that algorithms be correct on a substantially larger fraction of these answers than could be obtained by random guessing.

    We focus on two substantially stronger generalizations
    of strong direct-product theorems. 
    The first strengthening, strong \emph{selective} direct-product theorems, apply to algorithms that can adaptively choose, based on what they learn from their queries, which questions out of a
    potentially much larger list they want to answer.
    Even if the original problem involves a function, this selective direct-product problem is fundamentally relational and cannot be reduced to the computation of a direct-product function.
    This generalization is a natural one that arises
    in the context of proving time-space tradeoff lower bounds for quantum query algorithms.
    \begin{itemize}
        \item We prove a strong selective direct-product theorem for the quantum query complexity of computing arbitrary functions.
        \item We obtain this by showing that a relational formulation of the
        multiplicative adversary method satisfies a strong selective direct property and that an equivalent formulation is 
        strong enough to be able to capture lower bounds
        proved using the negative-weights adversary method for function computation.  This was not previously known even without the requirement of
        selectivity.
        \item We also provide direct proofs of such bounds for specific problems such as search, or, and parity that feature elegant arguments for the single copy lower bounds.
    \end{itemize}
    The second further generalization of strong selective
    direct-product theorems (which also generalizes threshold direct product problems) is 
    \emph{list-decoding} direct product problems recently introduced by Ben-David and Blais (independent of our consideration of selective direct product problems) in the context
    of classical randomized query complexity.  
    This allows an algorithm to produce a large list of possible output vectors for each input and counts the algorithm as correct if the true answer is in the
    list.  
    Ben-David and Blais proved that classical
    randomized complexity for arbitrary Boolean functions satisfies
    a strong list-decoding direct-product property.
    \begin{itemize}
        \item We prove a full quantum strong list-decoding direct-product analogue of this theorem for all partial Boolean-valued functions, and show that such a theorem cannot hold for relations.
        \item We obtain this by proving that versions of the multiplicative adversary argument with no bias below the eigenvalue threshold give us strong list-decoding direct-product theorems and then
        giving a reduction from negative-weights adversaries for Boolean-valued functions to multiplicative adversaries with exactly that property.
        \item We also give direct proofs of strong list-decoding direct product theorems for the quantum query complexity
        of search, or, and parity.
    \end{itemize}
\end{abstract}
\setcounter{page}{1}

\section{Introduction}

Direct sum and direct product theorems have had major importance in the
study of complexity in many models of computation.
The fundamental question they focus on is the way that the computational complexity of a task of computing a function $f$ (or relation $\relation$) changes as one requires that outputs of $f$ (resp. $\relation$) be produced for many independent inputs; we use $f^{\otimes k}$ (resp. $\relation^{\otimes k}$) to
denote the problem for $k$ independent inputs.
\emph{Direct sum} theorems focus on how the resources required to solve
$f^{\otimes k}$ grow as a
function of $k$ -- ideally showing that this requires an $\Omega(k)$ factor more resources than solving a single instance.\footnote{A direct sum theorem for communication complexity of relations (circuit depth) would yield a proof that $\P \not\subseteq \NC^1$~\cite{DBLP:journals/cc/KarchmerRW95}.
Not all models of computation have ideal direct sum theorems for all functions: The direct sum of $n$ copies of matrix-vector product 
 $f_A(x)=Ax$ over $\F^n_2$ is just matrix multiplication $AX$ and therefore has circuits of size at most $n^\omega \leq n^{2.38}$ but by a counting argument are there are matrices $A$ such that computing $Ax$ requires circuit size $n^2/\log n$~\cite{DBLP:journals/siamcomp/BarakBCR13}.}
\emph{Direct product} theorems focus on how the answer quality for solving
$f^{\otimes k}$ (resp. $\relation^{\otimes k}$) changes -- ideally showing that probability
of being correct declines by an $\Omega(k)$ power from the probability
of solving a single instance.
\emph{Strong direct product} theorems give the best of both, yielding ideal direct sum and direct product properties simultaneously.\footnote{Even direct product theorems that are not strong can have many important applications; for example, Raz's parallel repetition theorem~\cite{DBLP:journals/siamcomp/Raz98} is a direct product theorem for two-player games that has led to many improved hardness of approximation results \cite{raz2p1rsurvey}.
In the case that the function $f$ is Boolean, \emph{XOR lemmas}~(e.g.,~\cite{DBLP:conf/focs/Yao82a}) which require computation of $f^{\oplus k}$ which is a single bit that is the parity of the $k$ outputs of the function $f$ have also been very important.  These are closely related to direct product theorems
and also have many applications, but are not our focus.}

We particularly focus on quantum query complexity.

\paragraph{Strong direct-product theorems in quantum query complexity and  multiplicative adversary methods}
Direct products in quantum computation were first considered in earnest by
Aaronson~\cite{Aar05} in the context of search.
Klauck, \v{S}palek, and de Wolf~\cite{KSdW07} strengthened Aaronson's bound for search and proved a strong direct product theorem for $\OR_n$.
By reduction, this also yielded direct product theorems for quantum query complexity of functions based on their
\emph{block sensitivity}\footnote{If the number of queries for an algorithm for the problem $f^{\otimes k}$ is at most
some $\alpha k$ times the block sensitivity of $f$, which is a complexity measure that is at least the cube root of the quantum query complexity of $f$, then the probability of its
success is exponentially small in $k$. (See, e.g. \cite{DBLP:phd/ndltd/David17} for a survey.)}~\cite{DBLP:journals/cc/NisanS94,DBLP:journals/cc/NisanS94}.
Their work also leveraged these strong direct product theorems to prove the first non-trivial quantum time-space tradeoff lower bounds.

Klauck, \v{S}palek, and de Wolf proved their quantum strong direct product theorems by combining the polynomial method with extremal properties of Chebyshev polynomials.
Alternative ideas for achieving strong direct-product properties
were given by Ambainis, \v{S}palek, and de Wolf~\cite{ASdW09}, whose
basic ideas were generalized by
\v{S}palek~\cite{Spa08} with his \emph{multiplicative adversary method}. 
\v{S}palek showed that this new method could be used to derive strong direct-product lower bounds for relations; indeed, the method itself seems to have been 
developed expressly with this purpose in mind.
\cite{DBLP:conf/coco/AmbainisMRR11} extended \v{S}palek's ideas and
derived a somewhat sharper form of strong direct product theorem for function computation and quantum state generation (but not relational computation) based on a modified form of the multiplicative adversary method. 
Around the same time,
Sherstov~\cite{sherstovSQDPT} showed that the polynomial method satisfies a strong direct product theorem in general, thus proving a tight result for any function $f$ where the polynomial method is tight, subsuming the quantum direct product theorems in \cite{KSdW07}).

Finally, Lee and Roland~\cite{LeeRolandSQDPT} leveraged (1) a proof~\cite{DBLP:conf/focs/LeeMRSS11} that an extension, known as the \emph{negative-weights adversary method}, by H\/{o}yer, Lee, and \v{S}palek~\cite{HLS07} of Ambainis' \emph{additive adversary method}~\cite{Amb02} yields asymptotically optimal quantum query lower bounds for function computation and (2) a refinement of a simulation of
the extended additive adversary method by the multiplicative adversary method introduced in~\cite{DBLP:conf/coco/AmbainisMRR11}, in order to derive 
a general quantum strong direct product theorem for function computation.
Lee and Roland also gave a direct method based on properties of semi-definite programs that gave a general strong direct product theorem for quantum state generation, improving on a similar result in \cite{DBLP:conf/coco/AmbainisMRR11}, from which they derived their bounds for function (but not relational) computation.
Recently, in the context of relating these approaches to recording
query methods, Jeffery and Zur~\cite{DBLP:conf/icalp/JefferyZ26} state a relational version of the multiplicative adversary method that retains the initial conditions of \cite{Spa08} but
replaces the progress bound and target by the stronger ones of~\cite{DBLP:conf/coco/AmbainisMRR11}.\footnote{They also define a related method that they call the ladder adversary method and claim to prove a strong direct product theorem for this method. However, their proof of this strong direct product theorem (contained in the arXiv version~\cite{JZ25-arxiv}) contains a bug. The final line in the proof of their Theorem 6.1 assumes that the most multiplicative progress can always be made in the first $10T/k$ steps.} However, they do not show that their method captures
the universal function lower bounds derived by the negative-weights adversary method or that it satisfies a strong direct product theorem; we establish both of these facts, in part by giving a new equivalent formulation of the progress bound.
Together, these involve a variety of subtly different versions of what
is commonly referred to by a single
phrase: ``the multiplicative adversary method". 

\paragraph{Strong direct-product theorems in classical randomized query complexity}

The story of classical direct-product theorems is also
relevant.
Direct product lower bounds for classical computation date back at least to work of Impagliazzo, Raz, and Wigderson~\cite{DBLP:conf/coco/ImpagliazzoRW94}, who showed nontrivial direct product theorems for both randomized decision tree (i.e. query) complexity and communication complexity.
With classical computation, Yao's minimax lemma~\cite{DBLP:conf/focs/Yao77} implies that complexity for randomized computation and is equally captured by the success of deterministic algorithms under arbitrary distributions, termed the \emph{distributional complexity} with respect
to those distributions.
Since then, the use of distributional complexity has been the 
cornerstone of analysis of classical randomized algorithms.
However, 
Shaltiel~\cite{DBLP:journals/cc/Shaltiel03} gave an example of a Boolean function
and an associated distribution to show that 
distributional query complexity does not in general satisfy a strong direct-product property and therefore randomized query complexity cannot satisfy an optimal strong direct product property with probability decay
that matches the obvious upper bound.

Though this was a bad example, Klauck, \v{S}palek, and De Wolf, in the same paper discussed earlier~\cite{KSdW07}, were able to prove specific strong direct product theorems for classical randomized query complexity for $\OR_n$ and $\search_n$ with improved parameters relative to their strong direct product theorems for quantum computation. 
Their result, like their version for quantum query complexity, is limited by the block sensitivity $bs(f)$ and is only a strong-direct product result
if the randomized query complexity is $\Omega(bs(f))$.

Around the same time as some of the quantum strong direct product results we have discussed, Drucker~\cite{DBLP:journals/cc/Drucker12} was finally able to prove a fully general strong direct product theorem for classical query complexity for all functions, even for the distributional version of Shaltiel's function
(and input distribution) that takes into account
some necessary losses, and consequently for all of randomized query complexity.
The parameters necessarily are weaker relative to the
obvious algorithms, largely because of examples like Shaltiel's.
Drucker's method relies on analyzing deterministic algorithms via
Yao's lemma by appropriately scaling the bounds based on the tradeoffs between error and complexity in the single copy case.
Blanc, Koch, Strassle, and Tan~\cite{DBLP:conf/coco/BlancKST24} extended the range of strong direct product theorems for classical query complexity to the case of distributional computing using the expected numbers of queries and error rather than the worst-case numbers of queries and error.
Very recently, Ben-David and Blais~\cite{DBLP:conf/focs/BenDavidB25} produced a stronger form of classical
strong direct-product theorem, but we will delay discussing their work until after our work on selective
direct product problems, which
was completed independently of theirs.

\paragraph{Selective direct products}
We focus on an especially strong form of direct product theorem, which we call a strong \emph{selective} direct product theorem.  
For this problem, instead of being given $k$ inputs to a relation $\relation$ (or function $f$), there are $m$ inputs given, and the task is to produce answers for any $k$ of these inputs of the solver's choice, which can depend on the vector of inputs and randomness.
This is formally defined as follows:

\begin{definition}[Selective Direct Product]
    \label{def:easierdp}
    Given $m$ independently chosen inputs of a relation $\relation\subseteq \domain^n \times \range$, we use $\easierdp{\relation}{k}{m}$ to denote the \emph{selective direct product} relation on $(\domain^n)^m \times([m] \times \range)^k$, where
    $$\big((x_1,\ldots, x_m), \left ((i_1,y_1),\ldots, (i_k, y_k)\right )\big) \in \easierdp{\relation}{k}{m}$$
    iff $(x_{i_j},y_j) \in \relation$ for all $j\in[k]$ and $i_1\neq i_2 \neq \cdots \neq i_k$.
\end{definition}

In principle, for an algorithm, solving this selective version is an easier task than simply computing $k$ independent copies of the relation, since the algorithm can now pick and choose which copies are easiest, or which copies it can correlate best.
It is, however, a \emph{harder} task than correctly computing at least $k$ out of $m$ copies, as we require the algorithm to \emph{identify} which subproblems it has correctly computed. 
If we only required at least $k$ copies of the relation to be computed correctly, this would correspond to what is known in the literature as a \emph{threshold direct product} problem.
Indeed, it is known from a work of Unger~\cite{DBLP:conf/focs/Unger09} and related work (see also \cite{DBLP:journals/joc/ImpagliazzoJK09,DBLP:conf/approx/ImpagliazzoK10,DBLP:conf/coco/CleveSUU07}) that an ``optimal'' strong direct product theorem with $p' = p^k$ implies a threshold direct product theorem.
(Drucker~\cite{DBLP:journals/cc/Drucker12} gives a very general 
framework for similar kinds of problems that are equivalent to strong direct-product theorems.)
For large values of $m$, correctly computing at least $k$ copies of the function can become trivial via random guessing, whereas selectively computing $k$ copies may still be difficult.  
See~\cref{rem:selective-vs-threshold} for an example
of this contrast shows up in our results.

We can also see an example of the difference between ordinary strong direct-product theorems and the selective version if we consider a
function closely related to Shaltiel's
bad example, but with a more benign distribution:
Consider the Boolean function $f_\Shaltiel$ given by 
$f_\Shaltiel(x_1, \ldots, x_n) =  x_1 (x_2 \oplus \ldots \oplus x_n)$, which is essentially the same as Shaltiel's function, but with the uniform
distribution on inputs rather than Shaltiel's distribution that was biased on the first bit.
Under the uniform distribution, a single copy requires a linear number of queries to beat success probability 3/4 and has expected
complexity $\Omega(n)$.
This yields a strong direct-product result for
$f_\Shaltiel$ in distributional complexity with minimal loss.

However, the selective direct-product problem
$f_\Shaltiel^{(k,m)}$ has a trivial $O(k)$
query algorithm under the uniform distribution when $m$ is at least some constant factor larger than $k$, say at least $3k$:   
The algorithm queries the first bits of the copies of $f_\Shaltiel$ in order, stopping when it has found $k$ of them to be 0 and outputting the value 0 where 0's have been found.
This means that there is no strong selective direct-product theorem for $f_\Shaltiel$ under the uniform distribution, despite the fact that it does have a strong direct-product theorem under the uniform distribution.

Quantum adversary methods avoid such bad
examples for a simple reason:
Informally, these methods track the maximum progress that can be made per query towards solving the problem on the given input distribution.
For $f_\Shaltiel$, 
the uniform distribution would be a bad choice to use with quantum adversary methods, since a query on the first index would provide too much progress; the single-copy lower bound under the uniform distribution would only
be $1$ since a single query could determine the answer.

\paragraph{Strong quantum selective direct-product theorems}

We first prove strong selective direct product theorems for a few specific functions and relations entirely from first principles. These include the promise problem $\search_n$ function which takes as input
$x\in \bits^n$ with $|x|=1$ and outputs the unique index $i$ such that $x_i = 1$, the $\OR_n$ problem, and the $\Parity_n$ problem.

\begin{restatable}{theorem}{selectivesearch}
    \label{thm:search-sqdpt}
    Let $0< \beta\leq 1/2$ be such that $2^{2(\beta+H_2(\beta))/(1-\beta)}\le n$.
    Every quantum query algorithm solving $\easierdp{\search_n}{k}{m}$ with $T \leq \frac{\beta}{24} k \sqrt{n}$ queries has success probability at most $9\cdot 2^{-2\beta k}$.
\end{restatable}

\begin{remark} If $n\ge 64$ then the second constraint on $\beta$ is automatically satisfied by any $\beta\le 1/2$.
\end{remark}

\begin{restatable}{theorem}{selectiveOR}
    \label{thm:selectiveOR}
    Let $0 < \beta$ be such that $3\beta+2H_2(\beta) \leq 1$. 
    Every quantum query algorithm solving $\easierdp{\OR_n}{k}{m}$ with $T \leq \frac{\beta}{16} k \sqrt{n}$ queries has success probability at most $9 \cdot 2^{-2\beta k}$.
\end{restatable}

\begin{remark}
    The condition on $\beta$ in \cref{thm:selectiveOR} holds for any $0<\beta\le 0.07562$.
\end{remark}

\cref{thm:search-sqdpt,thm:selectiveOR} are qualitatively tight: Using Grover's search we can find $k$ 1's in our input with $O(k\sqrt n)$ queries and success probability at least exponentially large in $k$.

\begin{remark}
    \label{rem:selective-vs-threshold}
    We highlight the difference between threshold direct product statements and selective direct product statements with the following setting of parameters: 
    Let $m = n$ and $k = n/100$ in \cref{thm:selectiveOR}.
     Without any queries at all, via simple random guessing, with probability $1-o(1)$ we could correctly compute $\approx n/2$ copies of $\OR_n$ which is much larger than $k$. 
    It is then impossible to prove any sort of nontrivial threshold direct product theorem for this regime of parameters.
    On the other hand, \cref{thm:selectiveOR} shows that 
    any algorithm \emph{selectively} computing $n/100$ copies of $\OR_n$ with  $o(n^{1.5})$ queries can only succeed with probability that is exponentially small in $n$.
\end{remark}

The following selective direct product theorem for $\Parity_n$ is also tight up to constants: With $kn/2$ queries, an algorithm can solve the first $k$ copies with no error at all.

\begin{restatable}{theorem}{selectiveparity}
    \label{thm:selectiveparity} 
    Every quantum query algorithm solving $\easierdp{\Parity_n}{k}{m}$ with $T \leq kn/4$ queries has success probability at most $2^{k ( 2H_2(p) - 1 + p)}$, where $H_2(\cdot)$ is the binary entropy function and $p= 2T/(kn)$.
    In particular, if $T \leq kn/32$, then our success probability is at most $2^{-k/4}$.
\end{restatable}

These two theorems rely on simple and direct single-copy lower bound proofs for $\search_n$, $\OR_n$, and $\Parity_n$, that are inspired by the recording query method~\cite{Zha19}, seem to be new, and may be of independent interest.   

Going well beyond strong selective direct product theorems for these specific functions, we show two very general theorems:

\begin{theorem}[Informal]
\label{thm:selective-relations}
   The relational multiplicative adversary method applied to any relation $\relation$ and distribution $\dist$
   satisfies a strong selective direct product property.
\end{theorem}

\begin{theorem}
 \label{thm:selectivedp-from-additive}
    Every function satisfies a quantum strong selective direct product property; namely, 
    there are constants $\gamma >0$ and $k_0 \geq 1$ such that for all
    sufficiently small constant $\varepsilon>0$, 
    for any function $f$ and any integers $k_0\le k<m$,  
  $Q_{1-(1-\varepsilon)^{\gamma k}}(\easierdp{f}{k}{m})$
    is $\Omega( k\cdot Q_{\varepsilon}(f))$.
\end{theorem}

Our proof of \cref{thm:selectivedp-from-additive} follows by showing that any quantum query lower bound for functions proven with the negative-weights adversary method in \cite{HLS07}, which is known to be tight for quantum query complexity \cite{DBLP:conf/focs/Reichardt09, DBLP:conf/soda/Reichardt11a, DBLP:conf/focs/LeeMRSS11},
can be converted into an asymptotically matching lower bound using the Jeffery and Zur's relational formulation of the multiplicative adversary method~\cite{DBLP:conf/icalp/JefferyZ26}.
While a similar conversion is outlined in \cite{LeeRolandSQDPT}, they only perform the reduction for zero error computation algorithms and derived a form of the multiplicative adversary method that seems fundamentally incompatible with what we need.  
Indeed, while both \cite{DBLP:conf/coco/AmbainisMRR11} and \cite{LeeRolandSQDPT} give strong direct product theorems for function computation, it is unclear whether
they can yield nice characterizations for the general computation of relations, which is essential for us since selective direct products inherently are relations with many correct answers.

We also show that lower bounds for selective direct product relations automatically imply bounds for a
slightly relaxed version, which we term the \emph{$(1-\delta)$-selective direct product
relation} associated with a relation $f$: it relaxes the condition on the output list $(i_1,y_1),\ldots,(i_k,y_k)$  produced on each input and only requires that $(x_j,y_j)\in f^{(i_j)}$ for at least $(1-\delta)\,k$ values of $j\in [k]$.
These are essentially threshold versions of our strong selective direct product properties. 

\paragraph{List-decoding direct products}

In work noted earlier, Ben-David and Blais~\cite{DBLP:conf/focs/BenDavidB25} very recently
introduced an even stronger notion than selective
direct products, \emph{list-decoding direct products}, primarily in the context of Boolean functions, but
the notion has a natural extension to more general computation.
Like selective direct products, the number of input 
coordinates $m$ can be arbitrarily large compared
with their parameter $k$ that is the analog of the number of answers that the algorithm must produce in the selective
direct product case:
Consider the single-copy function $f:
\domain\rightarrow\range$. 
In the list-decoding
direct-product problem with parameter $k$,
for each input $x\in \domain^m$, the algorithm
is allowed to produce a list of possible answers $\listofout(x)\subset \range^m$ of size at most
$\abs{\range}^{m-k}$; the algorithm is correct on input $x$ if $f^{\otimes m}(x)\in \listofout(x)$. 
Clearly, the fact that this list-decoding version is
more general than the selective one is immediate since there are at most $\abs{\range}^{m-k}$ answers that are consistent with any $k$ answers to specific questions.

Using a careful charging mechanism and appropriate
distributions, Ben-David and Blais prove a classical strong list-decoding 
direct product theorem for all Boolean functions $f$:
There is a constant factor $c>0$, such that any randomized decision-tree solving the list-decoding direct-product problem for $f$ with parameter $k$ with a number of queries that is at most $ck$ times the number
of queries required to compute $f$ with
probability at least $2/3$ can succeed with probability that is at most exponentially small in $k$.  
Of course,
this proves a classical strong selective direct-product theorem, a fact that they note using different terminology from ours. 

\paragraph{Quantum strong list-decoding direct-product theorems}
We first prove quantum strong list-decoding direct product theorems for a few specific functions from first principles.

\begin{theorem}
\label{thm:list-decoding-specific}
    There are constants $c, c''>0$ such that for every positive integer $k \leq m$,
    any quantum algorithm solving the
    list-decoding direct-product problem with parameter $k$ for either (i)
    $\search_n$ or (ii) $\OR_n$ making at most $c k\sqrt{n}$ queries, or
    (iii) $\Parity_n$ making at most $ck n$ queries has
    success probability at most $2^{-c''k}$.
\end{theorem}

The proofs of these theorems use second-moment method inequalities for the 
probabilities of individual events that could contribute to the overall
algorithmic success.
This involves bounding 4-norm properties of the underlying amplitudes with an argument similar to Bonami's hypercontractive inequality for low degree functions~\cite{Bonami1970}.
The use of hypercontractivity in analyzing quantum complexity is not new, going back at least to to~\cite{DBLP:journals/siamcomp/GavinskyKKRW08,DBLP:conf/focs/Ben-AroyaRW08}; see~\cite{Mon12} for a decade-old survey.

A critical property of these lower bounds is that the bounds for the single copy case apply under a natural input distribution under which,
\emph{a priori}, a random guess of the output value will succeed with probability
exactly $1/\abs{\range}$.
As we discuss in detail once we have formalized the definitions, it is easy to see that strong list-decoding theorems are impossible without
this property of the input distribution, which also rules out its application to relations that are not functions on their domain $\inputspace$.

This has an interesting interaction with multiplicative adversary methods.
There is a parameter $\eta\ge 1/\abs{\range}$ where the method makes
no claim about the query cost of any quantum algorithm with success
probability at most $\eta$, corresponding to subspaces of the adversary
matrix with eigenvalue at most $\lambda$. 
The
same consideration discussed above rules out using the method to prove a list-decoding
theorem unless $\eta=1/\abs{\range}$.
We show that this is the only limitation by proving
that a list-decoding direct
product theorem holds for every function for which the relational multiplicative adversary method yields a 
single-copy bound with $\eta=1/\abs{\range}$.
The proof of this theorem relies on the same second-moment argument based on bounding
the 4-norms of amplitudes of events that could contribute to algorithmic
success.

However,
the reduction from the negative-weights adversary method to the multiplicative adversary method provided in \cite{DBLP:conf/coco/AmbainisMRR11} gives a value of $\eta > 1/\abs{\range}$ while the reduction in \cite{LeeRolandSQDPT} does not use the parameter $\eta$ at all.
Thus, in order to have any hope of proving strong list-decoding direct product theorems from negative-weights adversary arguments, a stronger reduction is required.

We prove precisely such a sharper reduction for every Boolean-valued function to the relational multiplicative adversary method.
This yields the following theorem that a strong quantum list-decoding direct product theorem holds for all Boolean-valued functions:

\begin{theorem}
    \label{thm:list-decoding-boolean-valued}
    Every Boolean output function satisfies a quantum strong list-decoding direct product property; namely there are constants $\gamma > 0$ and $k_0 \geq 1$ such that for all sufficiently small constant $\varepsilon > 0$, any function $f: \inputspace \to \{0,1\}$ where $\inputspace \subseteq \domain^n$, and any integers $k_0 \leq k < m$, $Q_{1-(1-\varepsilon)^{\gamma k}}(\easierlistdp{f}{k}{m})$ is $\Omega(k \cdot Q_\varepsilon(f))$.
\end{theorem}

\paragraph{Application to quantum time-space tradeoff lower bounds}

One of the original motivations of Klauck, \v{S}palek, and de Wolf in~\cite{KSdW07} for proving quantum strong direct product theorems was their application to quantum time-space tradeoff lower bounds. 
By showing how to embed independent copies of an $\OR_n$ function into the problems of sorting, Boolean matrix-vector product, and Boolean matrix multiplication, they were able to prove the first quantum time-space tradeoff lower bounds.
They extended a classical method for proving time space tradeoffs introduced by Borodin and Cook~\cite{BC82} to quantum computation.
The so-called Borodin-Cook method breaks computation into short segments
on which more than $k$ outputs (where $k$ is closely related to the space bound) could only be produced with a probability that is exponentially in $k$.
The combined large query lower bound and exponential decay of
success probability given by strong direct product theorems make them a natural fit for this method.

However, there is a weakness to the lower bounds produced
by the use of traditional strong direct product theorems in this context:  The lower bound only holds if the set of
$k$ coordinates that have their output values produced in a segment is fixed, independent of the input. 
Following~\cite{bkw:qmatrix-journal}, we call algorithms like this that have fixed time-steps for the coordinates,  
\emph{output-oblivious} algorithms.
This output-oblivious restriction is natural in circuit models where one thinks of each output wire (occurring after fixed timesteps) as being associated a specific coordinate, but not if each output
consists of the pair of both the name of the coordinate
and its value or if an algorithm is allowed to dynamically choose when to produce outputs, as would be typical in RAM algorithms.
Other quantum time-space tradeoff lower bounds that
also only apply to output-oblivious algorithms are the
alternative lower bound proof for sorting~\cite{HM21} based on recording-query arguments, the sharper lower bound for Boolean matrix multiplication~\cite{bkw:qmatrix-journal} based on a better embedding of the direct product of $\OR_n$, and
the bound for solving systems of integer linear inequalities~\cite{ASdW09}
based on other strong direct product results.

Our results making quantum strong-direct product theorems selective allow us to bound the success probability of quantum query algorithms that get to choose, in an input-dependent way, which output coordinates will have their
answers produced in any given short segment.
However, this alone is not sufficient to extend the output-oblivious time-space tradeoff lower bounds for sorting
and Boolean matrix-multiplication to general algorithms
because for those problems, set of output coordinates
in a segment also determines how the direct product problem is embedded in the original function.

Nonetheless, 
We give a proof-of-concept example showing the value of selectivity in proving general quantum time-space tradeoff lower bounds:  
The \emph{matrix mapping} problem is defined for a fixed matrix $A$ and the goal is to compute the matrix product $FA\tilde F^T$ where the input consists of a pair $F,\tilde F$ of $n\times n$ function matrices. 
Abrahamson~\cite{DBLP:journals/jcss/Abrahamson91} proved a tight classical time-space tradeoff 
of $T = \Omega(n^3/S)$ for matrix mapping for matrices $A$ with distinct entries.
A nearly tight output-oblivious quantum time-space tradeoff lower bound of $T=\Omega(n^{2.5}/S)$ follows easily from strong direct product theorem for $\search_n$, but the case of general quantum algorithms
was open and does not follow from prior techniques.  
Our strong selective direct product theorem for $\search_n$ lets us derive a fully general
quantum time-space tradeoff lower bound for matrix mapping.
Using the general transformation of~\cite{DBLP:journals/toct/BeameK25}, we also immediately obtain a cumulative quantum memory lower bound of $\Omega(n^{2.5})$ for matrix mapping.

\paragraph{Organization}
\cref{sec:prelims} gives an overview of basic notation, definitions, and relevant prior work including some detailed technical and historical background on the additive and multiplicative adversary methods in \cref{sec:adv-methods}.

In \cref{sec:our-mult-adv-and-reduction}, we give a detailed proof that the relational multiplicative adversary method gives lower bounds for functions that are within a constant factor of those produced by the negative-weights adversary method.
In \cref{sec:neg-to-mult-bool} we give a sharpened version of this connection in the case of functions with Boolean-valued outputs that yields multiplicative adversaries with nicer properties.
In \cref{sec:universal-mult} we derive the resulting universality results for the relational multiplicative adversary method (\cref{thm:mult-optimal,thm:bool-mult-bound-is-tight}) from the known universality results for the negative-weights adversary method.

We prove our results for strong selective
direct product theorems and their applications in \cref{sec:strong-select}.
In \cref{sec:specific} we present strong selective direct product theorems for $\search$, $\OR$, and $\Parity$ as a warm-up to our more general bound.
These theorems are all proven directly from first principles and follow the same high-level structure we will use later.
\cref{sec:selective-from-mult-adv} details how to achieve a strong selective direct product theorem from any multiplicative adversary lower bound (\cref{thm:selective-from-multiplicative-adversary}, the precise version of our \cref{thm:selective-relations}).
We then show how, thanks to \cref{thm:mult-optimal}, this immediately gives us strong selective direct product theorems for any function in \cref{sec:strong-selective-prod-from-query-bound}.
In particular, we prove \cref{thm:selective-allfuncs-precise}, which is a precise version of \cref{thm:selectivedp-from-additive}.
Next, in \cref{sec:threshold-selective}, we show how to leverage selective direct product theorems to obtain $(1-\delta)$-selective direct product theorems.
In \cref{sec:time-space-tradeoffs} we exhibit our example of how to use selective direct product theorems to prove time-space tradeoffs in the general quantum query model.

\Cref{sec:list-dec} contains our work on list-decoding direct product theorems.
We start \cref{sec:list-dec} by formally defining list-decoding direct products and pointing out some key limitations about generic strong list-decoding direct product theorems.
Next in \cref{sec:list-specific} we present strong list-decoding direct product theorems for $\search, \OR,$ and $\Parity$ described in \cref{thm:list-decoding-specific} as a warm-up for our multiplicative adversary derived bounds.
In \cref{sec:list-dec-mult} we prove \cref{thm:generic-list-decoding}, which shows how to convert any sufficiently strong multiplicative adversary lower bound into a strong list-decoding direct product theorem.
Finally in \cref{sec:list-decoding-for-all-boolean-out} we show how this implies a strong list-decoding direct product theorem for all Boolean-valued functions (\cref{thm:list-allbool-precise}, the precise version of \cref{thm:list-decoding-boolean-valued}).

\cref{sec:single-copy-search-or} gives simplified single copy lower bounds on the problems we consider in \cref{sec:specific}, and \cref{sec:list-dec-lemma} gives a key lemma we need for our proofs in \cref{sec:list-dec}.

\section{Preliminaries}
\label{sec:prelims}

\subsection{Binomial coefficient bounds}

We use $\binom{n}{\leq k}$ to denote the sum $\sum_{i=0}^k \binom{n}{i}$. 
We use the standard definition for binary entropy $H_2(p) := p \log(1/p) + (1-p)\log(1/(1-p))$, where $\log$ is assumed to be base 2 unless otherwise indicated.
We will need the following standard upper bounds on $\binom{n}{\leq k}$.

\begin{proposition}
    \label{prop:binom-tail}
    $\binom{n}{\leq k} \leq  2^{n H_2(k/n)}$ for all $k \leq n/2$ and $\binom{n}{\leq k} \leq \bigparen{\frac{en}{k}}^k$ for all $k \leq n$.
\end{proposition}

\blank{
\begin{proof}
    For the first part of the bound, note that for $p \leq 1/2$ we have 
    \begin{align*}
        1 &\geq \sum_{i = 0}^k p^i (1-p)^{n-i} \binom{n}{k}\geq p^k (1-p)^{n-k} \binom{n}{\leq k},
    \end{align*}
    which, plugging in $p = k/n$, tells us that 
    \begin{displaymath}
        \binom{n}{\leq k} \leq (n/k)^k (n/(n-k))^{n-k} 
        = 2^{ k \log (n/k) + (n-k) \log (n/(n - k))} \\
        = 2^{n\left (\frac{k}{n}\log(\frac{n}{k}) + \frac{n-k}{n}\log (\frac{n}{n-k})\right )} = 2^{n H_2(k/n)}. 
    \end{displaymath}
    For the second bound, we have that
    \begin{align*}
        \binom{n}{\leq k} &= \sum_{\ell = 0}^k \frac{n!}{\ell! (n-\ell)!}
        \leq \sum_{i=0}^k \frac{n^\ell}{\ell!} 
        = \sum_{\ell=0}^k \frac{n^\ell k^\ell}{ k^\ell \ell!} 
        \leq \bigparen{\frac{n}{k}}^k 
        \sum_{\ell=0}^\infty \frac{k^\ell}{\ell!} 
        = \bigparen{\frac{en}{k}}^k. \qedhere
    \end{align*}
\end{proof}
}

\subsection{Quantum computation of relations}

We review standard definitions of quantum query
algorithms with phase oracles and
how these definitions apply to computing relations -- in particular, on how these definitions apply to selective direct-product relations.
We postpone discussion of list-decoding direct-product relations until \cref{sec:list-dec}.

We consider quantum query algorithms that seek to compute $k$ correct answers given $m$ copies of some target relation $\relation \subseteq \inputspace \times \range$ where $\inputspace \subseteq \domain^n$ for some finite $\domain$.
We use $d$ to denote $|\domain|$, the size of the input space in each coordinate.
For example, one of the problems we focus on is $\relation = \search_n \subseteq \bits^n \times [n]$, though we will
actually analyze its promise version on
$\inputspace\subset\bits^n$ consisting of the set of strings of Hamming weight 1, which is a function.

In general, we will have some relation-dependent hard distribution $\calD$ supported on some subset $\inputspace$ of $\domain^n$. 
For (selective) direct product theorems, we will be given $m$ inputs, each independently drawn from $\calD$.
An input $x \in \inputspace$ drawn from $\dist$ is represented by the \emph{mixed state} $\rho_{\dist} = \sum_{x \in \inputspace} \dist(x) \op{x}{x}$ stored in an input register $\inputreg$ that is not directly accessible to the algorithm.
In practice, it is convenient to consider the input to be in a \emph{purification} of this mixed state, encoding it as $\ket{\dist} = \sum_{x \in \inputspace} \sqrt{\dist(x)} \ket{x}$.
Importantly, a quantum algorithm can never distinguish between these different input models;
when considering the algorithm's state without this input register, it will always be exactly the same whether  the input is represented as $\ket{\dist}$ or $\rho_\dist$.

While quantum algorithms do not have direct access to the input register $\inputreg$, they have complete control over a workspace register $\calW$.
We will often think of $\calW$ as consisting of an index $i$ to query,\footnote{When considering a direct product problem over $m$ copies, we will often either treat $i$ as $in + j$ or simply write it as a tuple $(i,j)$, where $i\in [m]$ now denotes which copy we are querying, and $j\in[n]$ denotes the index we query within that copy.} a phase $p$ for that query, and remaining workspace $w$, which is used to remember information between queries and produce output values.
Similar to \cite{Amb02, Zha19,HM21}, we use a general (phase) oracle operator $\oracle$ that allows the quantum algorithm to interact with a purified input register initialized to the state $\ket{\calD}_\inputreg$.

We define $\oracle$ to perform the linear mapping
\begin{equation*}\oracle \ket{i,p,w}_\calW\ket{x}_\inputreg := \omega_d^{p\, x_i} \ket{i,p,w}\ket{x},\end{equation*}
where $\omega_d$ is a primitive $d$-th root of unity
and the use of $x_i$ in its exponent is interpreted via some fixed mapping identifying $\domain$ with the elements of $\mathbb Z_d$; we keep that mapping implicit to avoid notational clutter.
We will sometimes fix the index $i$ and phase $p$ and consider $\oracle_{i,p}:= \sum_x \omega_d^{px_i} \op {x}{x} $. Note that $\oracle = \bigoplus_w \bigoplus_{i}\bigoplus_{p} \oracle_{i,p}$.

We will often drop the subscripts $\calW, \calI$ for convenience when it is clear from context.
A $T$-query quantum query algorithm $\calC$ is specified using input independent unitaries $U_0, \ldots, U_T$ acting only on $\calW$.
More precisely, the joint state of the input and quantum query algorithm at the end of the computation is given by
\begin{math}\ket{\psi_T} = U_T \oracle \ldots U_1\oracle U_0 (\ket{0, 0, 0}_\calW \otimes \ket{\calD}_\inputreg)\end{math}.
Due to the deferred measurement principle, we can assume that a quantum query algorithm delays all measurements until after the final query without loss of generality.\footnote{This assumption is often taken for granted; however, when analyzing space-bounded quantum computation one has to be somewhat more careful, as delaying measurements can increase qubit costs. 
As we will see later in this section, we interact with space bounded computation in a way where the deferred measurement principle still holds when proving lower bounds.}
This means that we can assume that our quantum algorithm is always in a pure state.
We use $\ket{\psi_t}$ to denote the state of our algorithm after the $t$-th query to the input, and this state can be expressed as $\ket{\psi_t} = \sum_{i,p,w,x} \alpha_{i,p,w,x} \ket{i,p,w}\ket{x}$.
Because quantum query algorithms cannot directly interact with their input registers and oracle calls only apply phases, we note that $\ket{\psi_t} = \sum_{x} \sqrt{\dist(x)}\ket{\psi_t^x} \ket{x}$, where each $\ket{\psi_t^x}= \sum_{i,p,w} \frac{\alpha_{i,p,w,x}}{\sqrt{\dist(x)}} \ket{i,p,w}$ is the state the algorithm would be in (with a global phase) on input $\ket{x}$.

When it would otherwise be ambiguous, we use the subscript $\rho^t_\calW$ to denote the partial trace of $\ket{\psi_t}$ that traces out the input registers, and we use $\rho^t_\inputreg = \sum_{x,x'} \sqrt{\dist(x)\dist(x')}\ip{\psi_t^{x'}}{\psi_t^x}\op{x}{x'} $ to denote the partial trace of $\ket{\psi_t}$ that traces out the worktape registers.

The output of the quantum query algorithm is determined by measuring the work register of $\ket{\psi_T}$ in the standard basis and applying some input-independent post-processing function $q$ to interpret the result as an output $q(w)\in\range$.
We say that an algorithm computes $\relation$ if, on input $x$, it produces some $y\in \range$ such that $(x,y) \in \relation$.
We use $Q_{\varepsilon}(\relation)$ to denote the quantum query complexity of computing $\relation$ with probability at least $1-\varepsilon$ on every input.
We use $Q_{\varepsilon}^\dist(\relation)$ to denote the quantum query complexity of computing $\relation$ on an input $x \sim \dist$ with probability at least $1-\varepsilon$.

    Given a relation $\relation$ defined over $\inputspace\times\range$, for each $y \in \range$ we define $\Pi_{\relation^{-1}(y)}$ to be the projection on $\inputspace$ given by
    $$\Pi_{\relation^{-1}(y)}=\sum_{x\in \relation^{-1}(y)}
        \ket{x}\bra{x}.$$
We then define $\pisucc^\relation$ for a single copy of relation $\relation$ as
\begin{equation}
    \pisucc^\relation = \sum_{w} \Pi_w \otimes \Pi_{\relation^{-1}(q(w))} = \sum_{\substack{i,p,w,x \\ (x,q(w))\in\relation}} \op{i,p,w,x}{i,p,w,x}. \label{eq:pisucc-R}
\end{equation}

By definition, the selective direct product relation $\easierdp{\relation}{k}{m}$ can be viewed as a set of pairs
$(x,y)$ with $x\in \inputspace^m$ and $y=(K,\tau)$ for $K\subseteq \binom{[m]}{k}$ and $\tau\in \range^K$ such that $x_i\in \relation^{-1}(\tau_i)$
for all $i\in K$.
For any pair $(K,\tau)$ with $K\subseteq [m]$ and $\tau\in \range^K$, we define projection $\Pi^{\relation,m}_{K,\tau}$ on $\inputspace^m$ by
\begin{equation}
\Pi^{\relation,m}_{K,\tau} = \identity_{\inputspace^{[m]\setminus K}}\otimes \bigotimes_{i\in K} \Pi_{\relation^{-1}(\tau_i)}.
\end{equation}
For any algorithm for this selective direct product problem, 
the input-independent post-processing function $q$ of the worktape contents $w$ yields a pair $q(w)=(K(w),\tau(w))$
with $K(w)\in \binom{[m]}{k}$ and $\tau(w)\in \range^{K(w)}$. 

In general, we can define the projector 
\begin{equation}\label{eq:output-k}
\pisucc^{\easierdp{\relation}{k}{m}} =\sum_w \Pi_w \otimes \Pi^{R,m}_{K(w),\tau(w)} =\sum_{\substack{i,p,w,x \\ \text{s.t. } x_i \in \relation^{-1}(\tau(w)_i)\text{ for all }i\in K(w)}}\ket{i,p,w,x}\bra{i,p,w,x}.
\end{equation}
When the problem of interest is clear from context, we often drop the superscript and simply write $\pisucc$.
The probability that the circuit produces a correct output is then given by $\norm{\pisucc \ket{\psi_T}}^2$.

\subsection{Space-bounded quantum computation}

A small portion of our paper involves
space-bounded quantum query algorithms.
Our definition of these algorithms is similar to
those of~\cite{HM21,DBLP:journals/toct/BeameK25,bkw:qmatrix-journal}.
Like unrestricted quantum query algorithms, space-bounded quantum query algorithms start in the all $\ket{0}$ state and cycle between applying input queries $\oracle$, and arbitrary input-independent unitaries $U_t$.
However, in contrast with unrestricted quantum query algorithms, which could postpone all measurements to the end of the computation without loss of generality, it is important that space-bounded quantum query algorithms are allowed to make intermediate measurements on their state after applying each $U_t$, as shown in \cref{fig:quantum-circuit}.
Adopting the notation of \cite{DBLP:journals/toct/BeameK25,bkw:qmatrix-journal}, we will consider the set of consecutive $\calO$, $U_t$ and measurement as layer $L_t$.

The space of layer $L_t$ is the number of qubits that are passed from layer $L_t$ to $L_{t+1}$ and is denoted $S_t$.
We define the space of the algorithm as the maximum space of any layer and the time (query complexity) as the total number of layers.
Thus the space needed to store the input and output is not included in this model.

Intermediate measurements enable quantum query algorithms to produce parts of their output early and discard unnecessary ancillary qubits.
Therefore there are two natural models we consider in this context.
\begin{enumerate}
    \item \textbf{General Quantum Query Model:} In this model, we assume that the quantum query algorithm has a dedicated register containing a Boolean flag and a potential output $(i,y_i) \in [m] \times \range$.
    After each query $\oracle$ and subsequent unitary operation $U_t$, the flag register is measured in the standard basis.
    Should the outcome $1$ be obtained, the output register (a subset of the qubits on the work tape that can be layer dependent) is measured in the standard basis, discarded, and optionally interpreted as an output pair $(i,y_i)$ which is written to a write-only tape.
    Otherwise, the algorithm produces no output during this layer.
    After each oracle query the algorithm is allowed to introduce additional ancillary qubits initialized to $\ket{0}$.

    In this model, there is no required structure on the output order, and indeed, the output order may depend on the input.
    This model has been previously considered in \cite{HM21, DBLP:journals/toct/BeameK25,bkw:qmatrix-journal}.

    \item \textbf{Output-oblivious Model:} This is a more restricted model of computation where the choice of when to produce each output value is independent of the input.
    In this \emph{output-oblivious} model, quantum circuits do not have a flag register.
    Instead, on predefined layers the quantum circuit measures an output register in the standard basis and interprets the result as an element of $\range$ corresponding to a fixed output index.
    This output-oblivious ordering restricts the set of allowed algorithms, making it easier to prove quantum time-space tradeoffs.
    This model has been most studied in the context of time-space tradeoffs for sorting and Boolean matrix multiplication problems \cite{KSdW07, bkw:qmatrix-journal}.
\end{enumerate}

Since the work of Klauck, \v{S}palek and de Wolf~\cite{KSdW07},
the following result of Aaronson~\cite{Aar05} has been a key tool for obtaining time-space tradeoff lower bounds for quantum computation.

\begin{proposition}[\cite{Aar05}]\label{prop:quant-union}
Let $\calC$ be a quantum circuit, $\rho$ be an $S$-qubit (possibly mixed) state, and $\maxmix$ be the $S$-qubit maximally mixed state. If $\calC$ starting in initial state $\rho$ produces some output $z$ with probability $p$, then $\calC$ starting in state $\maxmix$ will produce $z$ with probability at least $p/2^{S}$.
\end{proposition}

\begin{corollary}
    \label{cor:quant-union-oracle-reg}
    Let $\ket{\psi_t}=\sum_x \sqrt{\calD( x)}\ket{\psi^{ x}_t}_\calW\ket { x}_\calI$ be the state of a quantum algorithm after $t$ quantum queries under input distribution $\dist$.
    If $\calC$ starting in initial state $\ket{\psi_t}$ produces output $z$ with probability $p$, then if we replace the worktape register with $\maxmix$ and feed it into $\calC$, then $z$ is produced with probability at least $2^{-S} p$.
\end{corollary}

Importantly our method of proving time-space tradeoffs in \cref{sec:time-space-tradeoffs}, known as the Borodin Cook method \cite{BC82}, lets us prove query lower bounds on unitary algorithms (i.e. without intermediate measurements) and apply them to our space-bounded computation model.
More specifically we prove time-space tradeoffs by stitching together unitary query algorithms that are so short that they can only produce $\tildeO(S)$ correct outputs with a probability that is exponentially small in $S$.
Thus we can use \cref{cor:quant-union-oracle-reg} to argue that even with a $2^S$ factor amplification in their success probability caused by passing information between such algorithms, they still cannot produce more than $\tildeO(S)$ correct outputs each.
As $S$ grows, the number of such programs needed to produce all outputs goes down.
Hence, this proof strategy presents a tradeoff lower bound between time and space.
Since the space bound is applied in this technique only on the boundary between short programs, we can prove that such programs are bad at producing outputs in a unitary query model using the differed measurement principle and then apply these bounds to our space-bounded model containing intermediate measurements.

\begin{figure}
    \centering
    \includegraphics[width=.65 \textwidth]{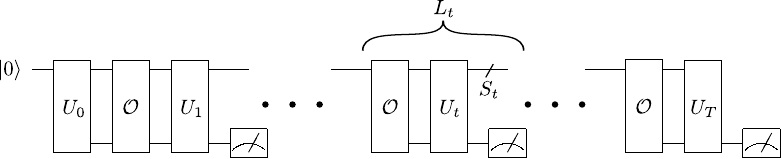}
    \caption{A general quantum query algorithm with $T$ queries.}
    \label{fig:quantum-circuit}
\end{figure}

\subsection{Quantum adversary methods}\label{sec:adv-methods}

Quantum adversary methods are a general collection of
methods that analyze quantum algorithms based on
\emph{adversary matrices} $\Gamma$ defined over the
domain $\inputspace$ that can be used to measure the
rate of progress of quantum query algorithms from
the initial state under distribution $\dist$ when the
workspace
is uncorrelated with the input to the desired states when the workspace is very highly correlated with a
desired relation or function value associated with the input.

\paragraph{Additive Adversary Method}
Beginning with the work of Ambainis~\cite{Amb02}, the \emph{additive adversary method} refers to many related quantum query lower bound proof techniques \cite{Amb02,DBLP:journals/jcss/Ambainis06,DBLP:journals/tcs/Zhang05,DBLP:conf/coco/BarnumSS03, DBLP:journals/siamcomp/LaplanteM08} based on the difficulty of distinguishing inputs that map to different outputs.
These approaches share a fundamental limitation due to a `certificate complexity barrier' preventing them from proving any lower bound stronger than $\sqrt{C_0(f)C_1(f)}$, where $C_0$ and $C_1$ are the zero and one certificate complexities of $f$ respectively \cite{DBLP:journals/tcs/Zhang05,DBLP:journals/toc/SpalekS06}.
In fact, \cite{DBLP:journals/toc/SpalekS06} showed that all of these lower bound methods are equivalent. 

The \emph{negative weights adversary method} introduced by Hoyer, Lee, and \v{S}palek~\cite{HLS07} extends these methods by using a stronger condition that measuring the final state of a quantum algorithm computing a function $f$ must actually determine the output with good probability after measurement.
They do so by extending the spectral version of the adversary method in \cite{DBLP:conf/coco/BarnumSS03} to witness matrices $\addAdvMat$ that can have negative (or even complex) entries, which turned out to be strictly more powerful for proving lower bounds.

\begin{definition}
    We use $D_i$ to denote the $\abs{\inputspace} \times \abs{\inputspace}$ matrix where $(D_{i})_{x,x'} = \indicator_{x_i \neq x'_i}$.
    We use $J$ to denote the all 1's $\abs{\inputspace} \times \abs{\inputspace}$ matrix.
\end{definition}

Here we present a re-normalized version of the negative weights adversary method in \cite{HLS07} featuring a slightly tighter dependence on $\varepsilon$ for non-Boolean outputs due to an improved analysis in subsequent work (e.g. \cite{DBLP:conf/coco/AmbainisMRR11}).

\begin{restatable}[\cite{HLS07}, normalized]{proposition}{addAdvMethod}\label{prop:add-adv-normalized}
Let $f: \inputspace \to \range$ be a function and $\dist$ be a distribution over $\inputspace \subseteq \domain^n$.
Let $\addAdvMat$ be a $\abs{\inputspace} \times \abs{\inputspace}$ Hermitian matrix such that $\addAdvMat_{x,x'} = 0$ whenever $f(x) = f(x')$, $\norm{\addAdvMat}=1$, and $\addAdvMat\ket{\dist} = \ket{\dist}$.
Let $\alg$ be a quantum algorithm that computes $f$ on distribution $\dist$ with probability at least $(1-\varepsilon)$ in $T$ queries and for $t = 0, \ldots, T$ let $\rho^t$ be the density matrix given by the state of $\alg$ after $t$ steps.
Then for $\progress_{\addAdvMat}^t = \tr[\addAdvMat\rho^t_{\inputreg}]$, where $\inputreg$ denotes the input registers of $\alg$, we have:
\begin{enumerate}
    \item\label{item:add-init} $\progress_{\addAdvMat}^0 = 1$
    \item\label{item:add-prog} $\progress_{\addAdvMat}^t - \progress_{\addAdvMat}^{t+1} \leq \addProgress(\addAdvMat) :=2 \max_i \norm{\addAdvMat \circ D_i}$
    \item\label{item:add-finish} $\progress_{\addAdvMat}^T \leq \addStop(\varepsilon)$
\end{enumerate}
Where $\addStop(\varepsilon) = 2\sqrt{\varepsilon(1-\varepsilon)} + \varepsilon$. Moreover if $\range = \{0,1\}$ then $\addStop(\varepsilon)$ can be tightened to $2\sqrt{\varepsilon(1-\varepsilon)}$.
Putting this together we can conclude that
\begin{math}\displaystyle
    T \geq \frac{1-\addStop(\varepsilon)}{ \addProgress(\addAdvMat)}.
\end{math}
\end{restatable}

\begin{definition}
    We write $\negAdvBound_{\varepsilon,\dist}(f)$ for the
    largest lower bound achievable using \cref{prop:add-adv-normalized} under distribution $\dist$ and error parameter $\varepsilon$ and we write $\negAdvBound_\varepsilon(f)=\max_{\dist} \negAdvBound_{\varepsilon,\dist}(f)$.
\end{definition}

Importantly, this extension of the adversary method to negative weights actually makes the method asymptotically tight for sufficiently small error.
In \cite{DBLP:conf/focs/Reichardt09}, Reichardt shows that optimal query lower bounds based on the negative weights adversary method can be converted into quantum query algorithms using span programs with only logarithmic loss in terms of the number of queries.
The log factor was later removed by Reichardt in \cite{DBLP:conf/soda/Reichardt11a}, giving a result that is tight for Boolean functions.
This bound still contained a log factor loss for larger input or output alphabets, which was finally removed in \cite{DBLP:conf/focs/LeeMRSS11}.

\begin{restatable}{proposition}{negWeightTight}\emph{(\cite{DBLP:conf/focs/LeeMRSS11})}
    \label{prop:reichardt-adv-upperbound}
    For any function $f:\inputspace \to \range$ where $\inputspace \subseteq \domain^n$, we have that $Q_{1/10}(f) = O(\negAdvBound_{1/10}(f))$.
\end{restatable}
\begin{remark}
    In \cite{DBLP:conf/focs/LeeMRSS11} the authors use $Q_{1/3}$ and the zero error version of $\negAdvBound$. For error $\varepsilon = 1/10$ on both quantities, the values all only change by constant factors.
    Note that we cannot pick $\varepsilon = 1/3$ since \cref{prop:add-adv-normalized} gives a negative query lower bound for non-Boolean valued functions with that choice of error.
\end{remark}

\paragraph{Multiplicative Adversary Methods}
While the negative weights adversary method is known to be tight, it only provides meaningful bounds in the low error regime; when $\varepsilon \geq 1/5$ the method already gives a meaningless lower bound for non-Boolean outputs.
\v{S}palek~\cite{Spa08} defined a new \emph{multiplicative adversary method} that can give meaningful lower bounds for significantly larger values of $\varepsilon$.
Where the negative weights adversary method tracks the \emph{additive} change in the progress measure $\progress$ per step of the algorithm, the multiplicative method instead tracks the \emph{multiplicative} change in $\progress$ after each step.
We use $\Pi_{\Gamma, < \lambda}$ to denote the projection onto the eigenspaces of $\Gamma$ with eigenvalues less than $\lambda$ and similarly define $\Pi_{\Gamma, > \lambda}, \Pi_{\Gamma, \leq \lambda}$, and $\Pi_{\Gamma, \geq \lambda}$

The following is \v{S}palek's original formulation of the method:
\begin{proposition}[Multiplicative Adversary, \cite{Spa08}]\label{prop:spa-adv-method}
    Let $f: \inputspace \to \range$ be a function and $\dist$ be a distribution over $\inputspace \subseteq \domain^{n}$.
    Let $\Gamma$ be a $\abs{\inputspace} \times \abs{\inputspace}$ positive-definite matrix with smallest eigenvalue $1$ such that $\Gamma \ket{\dist} = \ket{\dist}$.
    Let $\lambda \in (1,\norm{\Gamma}]$ and an $\eta$ satisfy $\norm{\Pi_{f^{-1}(y)} \Pi_{\Gamma, < \lambda}}^2 \leq \eta$ for all $y \in \range$.
    Let $A$ be a quantum algorithm that computes $f$ on input distribution $\dist$ with probability at least $\eta + 4 \zeta$ in $T$ queries and for $t = 0, \ldots, T$ let $\rho^t$ be the density matrix given by the state of $A$ after $t$ steps.
    Then for $W^t = \tr[\Gamma \rho_\inputreg^t]$ where $\inputreg$ denotes the input registers of $A$, we have:
    \begin{enumerate}
        \item $W^0 = 1$
        \item $W^{t+1}/W^t \leq \max_{i,p} \norm{\oracle^\dagger_{i,p} \Gamma \oracle_{i,p} \Gamma^{-1}}$
        \item $W^T \geq \zeta^2 \lambda$.
    \end{enumerate}
    Hence, for $c = \max_{i,p} \norm{\oracle^\dagger_{i,p} \Gamma \oracle_{i,p} \Gamma^{-1}}$ we have
    \begin{displaymath}
        Q_{1-\eta - 4 \zeta}(f) \geq \max_{\Gamma, \lambda} \log_c (\zeta^2 \lambda).
    \end{displaymath}
\end{proposition}

Ambainis, Magnin, Roetteler, and Roland~\cite{DBLP:conf/coco/AmbainisMRR11} simplify and improve on the multiplicative adversary method in a few key ways.\footnote{they also generalize the multiplicative adversary method to apply to quantum state generation, although this direction is not relevant to our use case.}
They assume that a quantum query is always either \emph{computing} (i.e. mapping $\ket{i}\ket{0}\ket{x} \to \ket{i}\ket{x_i}\ket{x}$) or \emph{uncomputing} (i.e. mapping $\ket{i}\ket{x_i}\ket{x} \to \ket{i}\ket{0}\ket{x}$) at a given time step.
Since any (phase or xor) quantum query can be simulated by using a computing query followed by an uncomputing query, they can assume this without loss of generality by doubling the allowed number of queries.
Under this assumption, the authors are able to give a cleaner expression for the per step multiplicative progress without needing to explicitly invoke the query oracle.
They also assume without loss of generality that their quantum algorithm ends in a state containing only the final output value in one register and all other work qubits are zeroed out, which they call \emph{coherent computation}.
This lets them define a target state 
\begin{displaymath}
    \rho^\odot =  \sum_{\substack{x,x'\\f(x) = f(x')}} \sqrt{\dist(x)\dist(x')} \op{x}{x'}
\end{displaymath}
for computing a function $f: \inputspace \to \range$ corresponding to the density matrix where $f$ has been coherently computed on all inputs (under the input distribution $\dist$).
With this notion of a target state, they are able to loosen the condition on $\lambda$ and $\eta$ to $\tr[  \Pi_{\Gamma, < \lambda}\rho^\odot] \leq \eta$.
In fact, they prove that their condition on $\lambda$ and $\eta$ is implied by the condition in \cite{Spa08} that $\norm{\Pi_{f^{-1}(y)} \Pi_{\Gamma, < \lambda}}^2 \leq \eta$.
Intuitively the \cite{Spa08} version of the condition on $\lambda$ and $\eta$ says that any input state supported on eigenspaces of $\Gamma$ with eigenvalue smaller than $\lambda$ corresponds to a distribution over inputs where the maximum probability of any fixed output value is at most $\eta$.
On the other hand, the \cite{DBLP:conf/coco/AmbainisMRR11} version of the corresponding condition says that the probability that the target state $\rho^\odot$ is observed to have energy at most $\lambda$ (with respect to $\Gamma$) is at most $\eta$.
In addition to this generalization of the condition on $\lambda$ and $\eta$, \cite{DBLP:conf/coco/AmbainisMRR11} improve on the bound in part (3) of \cref{prop:spa-adv-method} by replacing $\zeta^2 \lambda$ with $1+(\lambda-1)(\sqrt{\kappa} - \sqrt{\eta})^2$ where $\kappa$ is the target success probability.
They do this by requiring $\rho_\inputreg^T$ to be sufficiently close to the target state $\rho^\odot$.
Putting these changes together, they get the following multiplicative adversary method.

\begin{proposition}[Multiplicative Adversary, \cite{DBLP:conf/coco/AmbainisMRR11}]\label{prop:AMRR-adv-method}
    Let $f:\inputspace \to \range$ be a function, and $\dist$ be any distribution over $\inputspace \subseteq \domain^n$.
    Let $\Gamma$ be a $\abs{\inputspace} \times \abs{\inputspace}$ positive-definite matrix with smallest eigenvalue $1$ such that $\Gamma \ket{\dist} = \ket{\dist}$.
    Let $\lambda \in (1,\norm{\Gamma}]$ and an $\eta$ satisfy $\tr[\Pi_{\Gamma, < \lambda} \rho^\odot] \leq \eta$.
    Let $A$ be a quantum algorithm that coherently computes $f$ on input distribution $\dist$ with probability at least $1-\varepsilon \geq \eta$ in $T$ queries and for $t = 0, \ldots, T$ let $\rho^t$ be the density matrix given by the state of $A$ after $t$ steps.
    Then for $W^t = \tr[\Gamma \rho_\inputreg^t]$ where $\inputreg$ denotes the input registers of $A$, we have:
    \begin{enumerate}
        \item $W^0 = 1$
        \item $W^{t+1}/W^t \leq \max\{\norm{(\Gamma \circ (J-D_x))^{1/2} \Gamma^{-1/2}}^2,\norm{\Gamma^{1/2}(\Gamma \circ (J-D_x))^{-1/2}}^2 : \forall x \in [n]\}$
        \item $W^T \geq 1 + (\lambda - 1)(\sqrt{1-\varepsilon} - \sqrt{\eta})^2$.
    \end{enumerate}
    Hence, for $c = \max\{\norm{(\Gamma \circ (J-D_x))^{1/2} \Gamma^{-1/2}}^2,\norm{\Gamma^{1/2}(\Gamma \circ (J-D_x))^{-1/2}}^2 : \forall x \in [n]\}$ we have
    \begin{displaymath}
        2\cdot Q^\dist_{\varepsilon}(f) \geq \log_c(1 + (\lambda - 1)(\sqrt{1-\varepsilon} - \sqrt{\eta})^2).
    \end{displaymath}
    The factor of $2$ in this bound comes from the computing query/uncomputing query assumption.
\end{proposition}

While this is generally stronger than the original proof in \cite{Spa08} in the circumstances envisaged in~\cite{DBLP:conf/coco/AmbainisMRR11}, there are some limitations to this version of the method.
The idea of a fixed target state is fundamentally incompatible with relational computation, as such computation must admit multiple valid output states.
Since selective direct products are relations, this version of the multiplicative adversary method cannot be used to prove the bounds we want!
Moreover, while \cite{DBLP:conf/coco/AmbainisMRR11} produce a black-box reduction from the negative weights adversary method to their version of the multiplicative adversary method, the resulting converted bound is fundamentally incompatible with their strong direct product theorem for multiplicative adversaries; the conversion requires taking the limit of $\lambda \to 1$, which results in $\kappa \to 1$.
This means that \cite{DBLP:conf/coco/AmbainisMRR11} were only able to prove a direct sum theorem for quantum query complexity.

This problem was resolved by Lee and Roland~\cite{LeeRolandSQDPT}, who presented a tighter reduction from the additive adversary method to their version of the multiplicative adversary method, which is expressed as a semi-definite program.
Both their additive and multiplicative adversary methods bound the number of queries required for $\rho_\inputreg^T$ to exactly equal some state $\sigma$ and then minimize over all such bounds for states $\sigma$ where $\calF(\sigma, \rho^{\odot}) \geq \sqrt{\kappa}$ and $\calF(A,B) := \tr[\sqrt{\sqrt{A}B\sqrt{A}}]$ denotes the fidelity between $A$ and $B$.
This gives the a lower bound on the number of queries needed to compute $f$ with probability $\kappa$ under input distribution $\dist$.

\begin{proposition}[Multiplicative Adversary, \cite{LeeRolandSQDPT,DBLP:conf/stacs/MagninR13}]\label{prop:LR-adv-method}
    Let $f: \inputspace \to \range$ be a function, and $\dist$ be any distribution over $\inputspace \subseteq \domain^n$.
    Then the number of quantum queries $T$ needed for a quantum algorithm to generate state $\sigma$ is given by:
    \begin{align*}
        2 \cdot T \geq \max_{c, \Gamma \succeq 0} & \log_c(\tr[\Gamma\sigma ])\\
        \text{s.t. } &\tr[\Gamma \op{\dist}{\dist}] =1\\
         & c^{-1} \Gamma \preceq \Gamma \circ (J-D_i) \preceq c \Gamma \text{ for all } i \in [n]
    \end{align*}
    where $\Gamma$ is maximized over $\abs{\inputspace}\times \abs{\inputspace}$ positive definite matrices.
    Again the factor of $2$ loss comes from using the computing/uncomputing oracles assumption.
\end{proposition}

Lee and Roland prove a tighter reduction from the additive to multiplicative adversary method than that found in \cite{DBLP:conf/coco/AmbainisMRR11}, which they also show is compatible with their strong direct product theorem.
Unlike in \cite{DBLP:conf/coco/AmbainisMRR11}, the strong direct product theorem in \cite{LeeRolandSQDPT} is proven by showing that $T = \alpha m\cdot \madv_0(f)$ queries on input distribution $\dist$ are needed to enter a state $\rho_\inputreg^T$ such that $\calF(\rho_\inputreg^T, \rho^{\odot}) \geq \kappa^{k/2}$.
While this approach works for the function case, it does not clearly work for relational computation, as there is no single target state (or even a finite set of target states) analogous to $\rho^{\odot}$.

Jeffery and Zur~\cite{DBLP:conf/icalp/JefferyZ26} presented a version of the multiplicative adversary method that can be seen as a hybrid between those in \cite{Spa08} and \cite{DBLP:conf/coco/AmbainisMRR11}.
Their version uses the condition on $\lambda$ and $\eta$ from \cite{Spa08} with the stronger terminal value from \cite{DBLP:conf/coco/AmbainisMRR11}.
Importantly, they do this so that their multiplicative adversary method can be applied to relational computation (which they refer to as search).

\begin{proposition}[Multiplicative Adversary, \cite{DBLP:conf/icalp/JefferyZ26}]\label{prop:mult-adv}
    Let $\relation \subseteq \inputspace \times \range$ be a relation, $\dist$ be any distribution over $\inputspace \subseteq \domain^n$.
    Let $\Gamma$ be a $\abs{\inputspace} \times \abs{\inputspace}$ positive-definite matrix with smallest eigenvalue $1$ such that $\Gamma \ket{\dist} = \ket{\dist}$.
    Let $\lambda \in (1,\norm{\Gamma}]$ and an $\eta$ satisfy\footnote{In \cite{DBLP:conf/icalp/JefferyZ26} they use $\Pi_{\Gamma, <\lambda}$ instead of $\Pi_{\Gamma, \leq \lambda}$. Since $\Gamma$ has a finite number of eigenvalues, these requirements on $\lambda,\eta$ are interchangeable.} $\norm{\Pi_{\relation^{-1}(y)} \Pi_{\Gamma, \leq \lambda}}^2 \leq \eta$ for all $y \in \range$.
    Let $A$ be a quantum algorithm that computes $\relation$ on input distribution $\dist$ with probability at least $1-\varepsilon \geq \eta$ in $T$ queries and for $t = 0, \ldots, T$ let $\rho^t$ be the density matrix given by the state of $A$ after $t$ steps.
    Then for $W^t = \tr[\Gamma \rho_\inputreg^t]$ where $\inputreg$ denotes the input registers of $A$, we have:
    \begin{enumerate}
        \item $W^0 = 1$
        \item $W^{t+1}/W^t \leq \max_{i,p} \norm{\oracle^\dagger_{i,p} \Gamma^{1/2} \oracle_{i,p} \Gamma^{-1/2}}^2 = \max_{i,p,\rho} \tr[\multAdvMat \oracle_{i,p}^\dagger\, \rho\, \oracle_{i,p}]/\tr[\multAdvMat \rho]$
        \item $W^T \geq 1 + (\lambda-1)(\sqrt{1-\varepsilon} - \sqrt{\eta})^2$.
    \end{enumerate}
    Hence, for $c=\max_{i,p,\rho} \frac{\tr[\multAdvMat \oracle_{i,p}^\dagger \,\rho\, \oracle_{i,p}]}{\tr[\multAdvMat \rho]} =\max_{i,p} \norm{\oracle^\dagger_{i,p} \Gamma^{1/2} \oracle_{i,p} \Gamma^{-1/2}}^2$ we have 
    \begin{equation*}
    Q_{\varepsilon}^\dist(\relation)\ge \log_c (1+(\lambda-1) (\sqrt{1-\varepsilon} - \sqrt{\eta})^2).
    \end{equation*}
\end{proposition}

\begin{remark}\label{rem:matnorm-trace-equal-cond}
The expression $\max_{i,p,\rho} \tr[\multAdvMat \oracle_{i,p}^\dagger\, \rho\, \oracle_{i,p}]/\tr[\multAdvMat \rho]$ for the per step progress does not appear in \cite{DBLP:conf/icalp/JefferyZ26}.
The equivalence of these quantities follows from the following:
\begin{align*}
    \max_{i,p} \norm{\oracle_{i,p}^\dagger \multAdvMat^{1/2} \oracle_{i,p} \multAdvMat^{-1/2}}^2 &= \max_{i,p} \norm{\multAdvMat^{1/2} \oracle_{i,p} \multAdvMat^{-1/2}}^2 \qquad \textrm{since $\oracle^\dagger_{i,p}$ is unitary}\\
    &= \max_{i,p}\norm{(\multAdvMat^{1/2} \oracle_{i,p} \multAdvMat^{-1/2})^\dagger \multAdvMat^{1/2} \oracle_{i,p} \multAdvMat^{-1/2}} \qquad \textrm{since $\norm{M}^2 = \norm{M^\dagger M}$}\\
    &= \max_{i,p}\norm{\multAdvMat^{-1/2} \oracle_{i,p}^\dagger \multAdvMat \oracle_{i,p} \multAdvMat^{-1/2}}\\
    &= \max_{i,p,\ket{\psi}} \bra{\psi}\multAdvMat^{-1/2} \oracle_{i,p}^\dagger \multAdvMat \oracle_{i,p} \multAdvMat^{-1/2} \ket{\psi} / \ip{\psi}{\psi},\\
    \noalign{\textrm{since the largest Rayleigh quotient of Hermitian $M$ is the largest eigenvalue of $M$}}
    &= \max_{i,p,\ket{\phi}} \bra{\phi} \oracle_{i,p}^\dagger \multAdvMat \oracle_{i,p}  \ket{\phi} / \bra{\phi}\multAdvMat\ket{\phi} \qquad \textrm{setting $\ket{\phi} = \multAdvMat^{-1/2}\ket{\psi}$}\\
    &= \max_{i,p,\ket{\phi}} \tr[\multAdvMat \oracle_{i,p}  \op{\phi}{\phi} \oracle_{i,p}^\dagger] / \tr[\multAdvMat\op{\phi}{\phi}]\qquad \textrm{by the cyclic property of trace}\\
    &= \max_{i,p,\rho} \tr[\multAdvMat \oracle_{i,p}  \rho \oracle_{i,p}^\dagger] / \tr[\multAdvMat \rho] \qquad \textrm{maximized by pure $\rho$}\\
    &= \max_{i,p,\rho} \tr[\multAdvMat \oracle^\dagger_{i,p}  \rho \oracle_{i,p}] / \tr[\multAdvMat \rho] \qquad \textrm{since $\oracle^\dagger_{i,-p} = \oracle_{i,p}$}
\end{align*}
\end{remark}

\cref{prop:mult-adv} immediately yields the following multiplicative adversary lower bound for quantum query complexity:

\begin{restatable}[Multiplicative Adversary Bound, \cite{DBLP:conf/icalp/JefferyZ26}]{proposition}{multiplicative}
\label{prop:mult-adv-query-bound}
   Let $\relation \subseteq \inputspace \times \range$ be a relation and $\dist$ be a distribution over $\inputspace\subseteq \domain^n$.
    Let $\Gamma$ be a $\domainsize \times \domainsize$ positive-definite matrix with smallest eigenvalue $1$ such that $\Gamma \ket{\dist} = \ket{\dist}$.
    Let $\lambda \in (1,\norm{\Gamma}]$ and an $\eta > 0$ satisfy $\norm{\Pi_{\relation^{-1}(y)} \Pi_{\Gamma,\leq \lambda}}^2 \leq \eta$ for all $y \in \range$.
     Then, for $c=\max_{i,p,\rho} \frac{\tr[\multAdvMat \oracle_{i,p}^\dagger \,\rho\, \oracle_{i,p}]}{\tr[\multAdvMat \rho]} =\max_{i,p} \norm{\oracle^\dagger_{i,p} \Gamma^{1/2} \oracle_{i,p} \Gamma^{-1/2}}^2$ and for $\varepsilon < 1-\eta$
    we have 
    \begin{equation*}
    Q_{\varepsilon}^\dist(\relation)\ge \log_c (1+(\lambda-1) (\sqrt{1-\varepsilon} - \sqrt{\eta})^2).
    \end{equation*}
\end{restatable}

\begin{definition}
    We use $\madv^\dist_{\varepsilon,\lambda,\eta}(\relation)$ for $\eta<1-\varepsilon$ to represent the largest lower bound on $Q_{\varepsilon}^\dist(\relation)$ obtainable using \cref{prop:mult-adv-query-bound} over all valid $\Gamma$ with these parameters.
\end{definition}

Note that \cite{DBLP:conf/icalp/JefferyZ26} does not prove a reduction from the negative weights adversary method to their version of the multiplicative adversary method nor that their method satisfies a strong direct product property.
This is the version of the multiplicative adversary method that we will be using in this paper, and we show that it satisfies both of these properties.

\section{The universal application of the relational multiplicative adversary method to functions}\label{sec:our-mult-adv-and-reduction}
In this section,
we reprove and extend results in \cite{DBLP:conf/coco/AmbainisMRR11,LeeRolandSQDPT} showing that negative-weights adversary matrices can be converted into multiplicative adversary matrices that give asymptotically matching query lower bounds.
Unlike previous variants, we show that this conversion produces adversary matrices that 
we can use to yield strong direct-product lower bounds for relations.
As we show in \cref{thm:selective-from-multiplicative-adversary}, we can also use them to yield the stronger property of strong selective direct product lower bounds.

Since the negative-weights adversary method is known to be asymptotically tight for query complexity of
functions~\cite{DBLP:conf/focs/Reichardt09,DBLP:conf/soda/Reichardt11a,DBLP:conf/focs/LeeMRSS11}, this conversion will imply that quantum query complexity satisfies a strong selective direct product theorem
for all functions, as we show in~\cref{sec:strong-selective-prod-from-query-bound}.

\begin{theorem}\label{thm:add-to-mult}
    Let $f: \inputspace \to \range$ be a function and $\dist$ be a distribution over $\inputspace \subseteq \domain^n$.
    Suppose that $\addAdvMat$ is a witness to \cref{prop:add-adv-normalized} for computing $f$ on $\dist$ with error $\varepsilon\in [0,1/3]$ that gives a query lower bound of $T^\pm$.
    Let $\multAdvMat = 2 \identity - \addAdvMat$ and define progress measure $\progress^t = \tr[\multAdvMat \rho_\inputreg^t]$.
    Then $\multAdvMat$ is a witness to the multiplicative adversary method given in \cref{prop:mult-adv} with $\lambda = 3 - (\frac{4}{1-\varepsilon})^{1/3}$, $\eta = (\frac{1-\varepsilon}{4})^{1/3}$ for computing $f$ on $\dist$.
    Moreover,
    \begin{enumerate}
        \item\label{item:mult-conv-init} $\progress^0 = 2-\tr[\addAdvMat \rho_\inputreg^0] =1$,
        \item\label{item:mult-conv-prog} $\ln(\progress^{t+1}/\progress^t) \leq \max_{i,p,\rho} \ln \left(\frac{\tr[\multAdvMat \oracle_{i,p}^\dagger\,\rho\, \oracle_{i,p}]}{\tr[\multAdvMat \rho]}\right) \leq  2 \max_i \norm{\addAdvMat \circ D_i}$, and
        \item\label{item:mult-conv-finish} $\ln(1+(\lambda-1)(\sqrt{1-\varepsilon} - \sqrt{\eta})^2) = (1-2\sqrt{\varepsilon(1-\varepsilon)})/\zeta(\varepsilon)$ for
        $\zeta(\varepsilon)=\frac{1-2\sqrt{\varepsilon(1-\varepsilon)}}{\ln \left( 1- (1-(2(1-\varepsilon))^{1/3})^3\right)}<58.$
    \end{enumerate}
    Thus, $\multAdvMat$ gives a query lower bound for
    computing $f$ on distribution $\dist$ with error
    at most $\varepsilon$ of
    \begin{displaymath}
        T \geq \frac{\ln(1 + (\lambda-1)(\sqrt{1-\varepsilon} - \sqrt{\eta})^2)}{\ln (\progress^{t+1}/\progress^t)} \geq \frac{1-2\sqrt{\varepsilon(1-\varepsilon)}}{2\zeta(\varepsilon) \max_{i}\norm{\addAdvMat \circ D_i}} \geq T^\pm/\zeta(\varepsilon)\ge T^\pm / 58
    \end{displaymath}
    via that multiplicative adversary method.
\end{theorem}

Before diving into the proof of this lemma, we provide a road-map to the argument.
Similar to the approach in \cite{LeeRolandSQDPT}, this transformation involves replacing the negative-weights adversary matrix $\addAdvMat$ with the multiplicative adversary matrix $\multAdvMat = 2 \identity - \addAdvMat$.
The key technical challenges in doing so are showing that
\begin{description}
    \item
[(i)] $\multAdvMat$ satisfies the $(\eta, \lambda)$ condition of the relational  multiplicative adversary method and \item[(ii)] When $\multAdvMat$ is used as a witness in that method, the query lower bound that it derives is at most a constant factor worse than the query lower bound given by \cref{prop:add-adv-normalized} with witness $\addAdvMat$.
\end{description}

Part (i) has to address the issue that, unlike the negative-weights adversary method, the multiplicative adversary method requires specifying a pair $(\lambda, \eta)$ such that for all $y \in \range$, $\norm{\Pi_{f^{-1}(y)} \Pi_{\multAdvMat, \leq \lambda}}^2 \leq \eta$.
The version of the multiplicative adversary method used in \cite{LeeRolandSQDPT} does not feature such a pair and the version in \cite{DBLP:conf/coco/AmbainisMRR11} uses the weaker condition of $\tr[\Pi_{\multAdvMat, <\lambda} \rho^\odot] \leq \eta$ for target state $\rho^\odot$ on the pair than \cref{prop:mult-adv} requires.
Thus we must generalize the method for extracting $\lambda$ and $\eta$ from a negative-weights adversary matrix found in \cite{DBLP:conf/coco/AmbainisMRR11} to the stronger condition on these parameters used in \cref{prop:mult-adv}.
\cref{prop:lambda-thresh-trace} generalizes\footnote{This is a generalization since $\rho^\odot$ can be expressed as a convex combination of these $\sigma_y$.} the trace condition in \cite{DBLP:conf/coco/AmbainisMRR11} that applies to all input states $\sigma_y$ with amplitude only on inputs $x \in f^{-1}(y)$.
More specifically, $\tr[\Pi_{\addAdvMat, \geq \lambda^\pm} \sigma_y] = \tr[\Pi_{\multAdvMat, \leq \lambda} \sigma_y] \leq \frac{1}{3-\lambda}$.
Next we can show that \cref{prop:lambda-thresh-trace} actually implies the definitions of $\lambda$ and $\eta$ from the relational multiplicative adversary in \cref{prop:lambda-thresh-norm}.

Once we have established a pair $(\lambda, \eta)$ for $\multAdvMat$, part (ii)  shows that $\multAdvMat$ is a witness for \cref{prop:mult-adv} that gives a query bound that is at most a factor of $58$ times worse than what would be given by \cref{prop:add-adv-normalized} with witness $\addAdvMat$.
The negative-weights adversary method starts with a progress value of 1 and decrements it by at most $2 \max_i \norm{\addAdvMat \circ D_i}$ per step to\footnote{We use the special case negative-weights lower bound that only applies to Boolean output functions here since it is stronger than the general case.} $2 \sqrt{\varepsilon(1-\varepsilon)}$ while the multiplicative adversary method starts with a progress of 1 and increases it by a multiplicative factor of at most $\frac{\tr[\multAdvMat \oracle_{i,p}^\dagger \,\rho\, \oracle_{i,p}]}{\tr[\multAdvMat \rho]}$ per step to $1+(\lambda-1)(\sqrt{1-\varepsilon} - \sqrt{\eta})^2$.
This means we want to compare the multiplicative query lower bound of
\begin{displaymath}
    \frac{\ln(1+(\lambda-1)(\sqrt{1-\varepsilon} - \sqrt{\eta})^2)}{\ln(\tr[\multAdvMat \oracle_{i,p}^\dagger\, \rho\, \oracle_{i,p}]/\tr[\multAdvMat \rho])}
\end{displaymath}
to the negative-weights query lower bound of 
\begin{displaymath}
    \frac{1-2\sqrt{\varepsilon(1-\varepsilon)}}{2 \max_i \norm{\addAdvMat \circ D_i}},
\end{displaymath}
which is the lower bound in the Boolean case and is at least the lower bound for general $\range$.
We do so by individually bounding $\ln(1+(\lambda-1)(\sqrt{1-\varepsilon} - \sqrt{\eta})) \geq (1-\sqrt{\varepsilon(1-\varepsilon)}) / 58$ and $\ln(\tr[\multAdvMat \oracle_{i,p}^\dagger\, \rho\, \oracle_{i,p}]/\tr[\multAdvMat \rho]) \leq 2 \max_i \norm{\addAdvMat \circ D_i}$.
As one might expect, the main source of the loss in this conversion comes from the introduction of $\lambda$ and $\eta$ when converting from an additive matrix to a multiplicative one.
There is a tradeoff between the allowed values for these parameters in our \cref{prop:lambda-thresh-trace} that can effect the quality of our resulting multiplicative adversary bound.
We use the same choices of $\lambda$ and $\eta$ as in \cite{DBLP:conf/coco/AmbainisMRR11}, which provide the locally optimal constant factor loss in our final bound.

\paragraph{Part (i): Extracting $\lambda$ and $\eta$ from $\addAdvMat$}

We follow some of the ideas of the multiplicative adversary bound in \cite{DBLP:conf/coco/AmbainisMRR11} though that uses a different but related definition of the pair $(\lambda, \eta)$ to the one we use in \cref{prop:mult-adv}.\footnote{For coherently computing a function $f:\inputspace \to \range$, Ambainis, Magnin, Roetteler, and Roland~\cite{DBLP:conf/coco/AmbainisMRR11} define the target state:
$\rho^{\odot}=\frac{1}{\abs{\inputspace}}\sum_{\substack{x,x'\\f(x) = f(x')}} \op{x}{x'}$
and then require that $\multAdvMat, \lambda, \eta$ in their multiplicative adversary method together satisfy $\tr[\Pi_{\multAdvMat, \geq \lambda} \rho^{\odot}] \leq \eta.$
They then prove that \eqref{eq:spalek-condition} implies this condition. 
}

In order to convert a negative-weights adversary bound into a multiplicative one, with $\lambda$ and $\eta$
which mimics those of \cite{Spa08,DBLP:conf/icalp/JefferyZ26} by requiring that for every $y \in \range$
\begin{equation}
    \norm{\Pi_{f^{-1}(y)} \Pi_{\multAdvMat, \leq \lambda}}^2 \leq \eta, \label{eq:spalek-condition}
\end{equation}
we define
$$\fixedtargets_y := \curly{ \op{\psi_y}{\psi_y}\ :\ \ \ket{\psi_y} = \sum_{x : f(x) = y} \beta_x \ket {x}}$$
and further define $\fixedtargets_f = \bigcup_{y \in \range} \fixedtargets_y$. 
To prove the property~\ref{eq:spalek-condition} for our future choice of $\Gamma = 2\identity - \addAdvMat$, we prove the following property of $\Gamma^\pm$ on $\fixedtargets_f$.
This is a strengthening of Lemma 18 in \cite{DBLP:journals/corr/abs-1012-2112} (the full arXiv version of \cite{DBLP:conf/coco/AmbainisMRR11}).



\begin{lemma}\label{prop:lambda-thresh-trace}
    Let $\addAdvMat$ be any witness to \cref{prop:add-adv-normalized} for computing $f$ under distribution $\dist$ and let $-1 < \lambda^\pm < 1$.
    Let $\sigma_y \in \fixedtargets_{f}$.
    Then $
        \tr[\Pi_{\addAdvMat, \geq \lambda^\pm} \sigma_y] \leq 1/(1+\lambda^\pm)
    $.
\end{lemma}

\begin{proof}
By definition $\norm{\Gamma^\pm}=1$.
    Let matrix $F$ be given by $F_{x,x'} = \indicator_{f(x) \neq f(x')}$;  write $\sigma_y = \op{\phi_y}{\phi_y}$ and $\pibad = \Pi_{\addAdvMat, \geq \lambda^\pm}$.
    Since $\addAdvMat_{x,x'} = 0$ on all entries where $f(x) = f(x')$ and $(\sigma_y)_{x,x'} = 0$ on all entries where $f(x) \neq f(x')$, we have that:
    \begin{displaymath}
        \tr[\addAdvMat \sigma_y] = \tr[(\addAdvMat \circ F) \sigma_y] = \tr[\addAdvMat (\sigma_y \circ F)] = 0
    \end{displaymath}

    Thus, using the fact that all eigenvalues of $\addAdvMat$ are at least $-1$,
    \begin{align*}
        0 = \tr[\addAdvMat \sigma_y] & =  \tr[\addAdvMat \op{\phi_y}{\phi_y}]\\
        &= \tr[\addAdvMat (\pibad \op{\phi_y}{\phi_y} \pibad)] + \tr[\addAdvMat ((\identity-\pibad) \op{\phi_y}{\phi_y} (\identity - \pibad))]\\
        &=\bra{\phi_y} \pibad \addAdvMat \pibad \ket{\phi_y} + \bra{\phi_y} (\identity -\pibad) \addAdvMat (\identity - \pibad) \ket{\phi_y}\\
        &\geq \lambda^\pm \bra{\phi_y} \pibad \ket{\phi_y} - \bra{\phi_y} (\identity - \pibad) \ket{\phi_y}\\
        &= (\lambda^\pm + 1)\tr[\pibad \op{\phi_y}{\phi_y}] - 1.
    \end{align*}
    Rearranging this gives
    \begin{math}
        \tr[\Pi_{\addAdvMat, \geq \lambda^\pm} \sigma_y] = \tr[\pibad \op{\phi_y}{\phi_y}] \leq 1/(1+\lambda^\pm)
    \end{math}
    as desired.
\end{proof}

We now show how \cref{prop:lambda-thresh-trace} implies the stronger  condition \eqref{eq:spalek-condition} on $(\lambda,\eta)$ required for our \cref{prop:mult-adv}.
As we will see later in this section, we are interested in choosing a pair $(\lambda,\eta)$ for the matrix $\multAdvMat = 2\identity- \addAdvMat$.

\begin{lemma}\label{prop:lambda-thresh-norm}
    Let $\addAdvMat$ be any witness to \cref{prop:add-adv-normalized} for computing $f$ under distribution $\dist$ and let $1 < \lambda < 3$.
    Let $\multAdvMat = 2\identity-\addAdvMat$.
    Then
    $
        \forall y \in \range, \norm{\Pi_{f^{-1}(y)} \Pi_{\multAdvMat, \leq \lambda}}^2 \leq 1/(3-\lambda)
    $.
\end{lemma}
\begin{proof}
    By \cref{prop:lambda-thresh-trace} we know that, for any $\lambda^\pm \in (-1,1)$ and $\sigma_y \in \fixedtargets_f$,
    \begin{equation}\label{eqn:lambda-eta-add}
        \tr[\Pi_{\addAdvMat, \geq \lambda^\pm} \sigma_y] \leq \frac{1}{1+\lambda^\pm}.
    \end{equation}
    Since any eigenvector of $\addAdvMat$ with eigenvalue $\lambda^\pm$ is an eigenvector of $\multAdvMat$ with eigenvalue $\lambda = 2-\lambda^\pm$ we can conclude that $\Pi_{\addAdvMat, \geq \lambda^\pm} = \Pi_{\multAdvMat, \leq \lambda}$.
    Thus we can rewrite \eqref{eqn:lambda-eta-add} as
    \begin{equation}\label{eqn:lambda-eta-mult}
        \tr[\Pi_{\multAdvMat, \leq \lambda} \sigma_y] \leq \frac{1}{3-\lambda}.
    \end{equation}
    Now observe that, for any $y \in \range$,
    \begin{align*}
        \norm{\Pi_{f^{-1}(y)} \Pi_{\multAdvMat, \leq \lambda}}^2&= \norm{\Pi_{\multAdvMat, \leq \lambda}\Pi_{f^{-1}(y)}}^2 \quad \textrm{since the norm is preserved under $\dagger$ and all entries are real}\\
        &=\max_{\ket{\psi}} \norm{ \Pi_{\multAdvMat, \leq \lambda} \Pi_{f^{-1}(y)} \ket{\psi}}^2\\
        &= \max_{\ket{\psi}}\tr[ \Pi_{\multAdvMat, \leq \lambda} \Pi_{f^{-1}(y)} \op{\psi}{\psi}  \Pi_{f^{-1}(y)} \Pi_{\multAdvMat, \leq \lambda} ]\\
        &=\max_{\sigma_y \in \fixedtargets_f}\tr[\Pi_{\multAdvMat, \leq \lambda} \sigma_y \Pi_{\multAdvMat, \leq \lambda}] \quad \textrm{since the maximum is attained for some $\op{\psi}{\psi} \in \fixedtargets_f$}\\
        &=\max_{\sigma_y \in \fixedtargets_f}\tr[\Pi_{\multAdvMat, \leq \lambda} \sigma_y]\\
        &\leq \frac{1}{3-\lambda} \qquad \textrm{by \eqref{eqn:lambda-eta-mult}}
    \end{align*}
    as desired.
\end{proof}

\paragraph{Part (ii): Bounding the loss from the conversion}
Now we are ready to prove that negative-weights adversary lower bounds with witness $\addAdvMat$ can be converted to asymptotically matching multiplicative adversary lower bounds with witness $\multAdvMat = 2 \identity - \addAdvMat$.
Variants of this result have been shown before in \cite{DBLP:conf/coco/AmbainisMRR11,LeeRolandSQDPT}, however we require a different version of the reduction than those used in either of these papers.

\begin{proof}[Proof of \cref{thm:add-to-mult}]
    We start by observing that since $\addAdvMat$ is a Hermitian matrix such that $\addAdvMat \ket{\dist} = \ket{\dist}$, it must be the case for $\multAdvMat = 2 \identity - \addAdvMat$ that $\multAdvMat$ is a positive-definite Hermitian matrix with smallest eigenvalue $1$ where $\multAdvMat \ket{\dist} = \ket{\dist}$.
    For part \ref{item:mult-conv-init} of \cref{thm:add-to-mult} observe that $\progress^0 = \tr[\multAdvMat \rho^0_\inputreg] = \tr[(2\identity - \addAdvMat)\rho^0_\inputreg] = 2 - \tr[\addAdvMat \rho^0_\inputreg] = 1$.
    To prove part \ref{item:mult-conv-prog} of \cref{thm:add-to-mult} observe that
    \begin{align*}
        \max_{i,p,\rho} &\ln \left(\frac{\tr[\multAdvMat \oracle_{i,p}^\dagger\, \rho\, \oracle_{i,p}]}{\tr[\multAdvMat \rho]}\right)\\ 
        &=\max_{i,p,\rho} \ln \left(\frac{\tr[(2 \identity - \addAdvMat) \oracle_{i,p}^\dagger\, \rho\, \oracle_{i,p}]}{\tr[(2 \identity - \addAdvMat) \rho]} \right)\\
        &= \max_{i,p,\rho} \ln \left(\frac{2-\tr[\addAdvMat \oracle_{i,p}^\dagger\, \rho\, \oracle_{i,p}]}{2-\tr[\addAdvMat \rho]}\right)\\
        &= \max_{i,p,\rho} \ln \left(\frac{2-\tr[\addAdvMat\rho] +\tr[\addAdvMat (\rho - \oracle_{i,p}^\dagger\, \rho\, \oracle_{i,p})]}{2-\tr[\addAdvMat \rho]}\right)\\
        &= \ln\left(1+ \max_{i,p,\rho} \frac{\tr[\addAdvMat (\rho - \oracle_{i,p}^\dagger\, \rho\, \oracle_{i,p})]}{2-\tr[\addAdvMat \rho]}\right)\\
        &\leq \max_{i,p,\rho} \frac{\tr[\addAdvMat (\rho - \oracle_{i,p}^\dagger\, \rho\, \oracle_{i,p})]}{2-\tr[\addAdvMat \rho]} \qquad \textrm{since $\ln(1+x) \leq x$ for $x > -1$}\\
        &\leq \max_{i,p,\rho} \tr[\addAdvMat (\rho - \oracle_{i,p}^\dagger \,\rho\, \oracle_{i,p})] \qquad \textrm{since $\tr[\addAdvMat \rho] \leq 1$}\\
        &= \max_{i,p,\rho} \tr[\addAdvMat ((\rho - \oracle_{i,p}^\dagger \rho \oracle_{i,p})\circ D_i)] \qquad \textrm{since $\rho_{x,x'} = (\oracle_{i,p}^\dagger\, \rho\, \oracle_{i,p})_{x,x'}$ when $x_i = x'_i$}\\
        &=\max_{i,p,\rho} \tr[(\addAdvMat \circ D_i) (\rho - \oracle_{i,p}^\dagger\, \rho\, \oracle_{i,p})] \qquad \textrm{since $\tr[A(B\circ C)] = \tr[(A\circ C) B]$}\\
        &\leq \max_{i,p,\rho} \norm{\addAdvMat \circ D_i} \norm{\rho - \oracle_{i,p}^\dagger\, \rho\, \oracle_{i,p}}_{tr} \qquad \textrm{since $\tr[AB] \leq \norm{A}\norm{B}_{tr}$}\\
        &\leq \max_{i,p,\rho} \norm{\addAdvMat \circ D_i} (\norm{\rho}_{tr} + \norm{\oracle_{i,p}^\dagger\, \rho\, \oracle_{i,p}}_{tr}) \qquad \textrm{by triangle inequality}\\
        &= \max_i 2\norm{\addAdvMat \circ D_i} \qquad \textrm{since $\rho$ and $\oracle_{i,p}^\dagger\, \rho\, \oracle_{i,p}$ are density matrices.}
    \end{align*}
    
    To prove part \ref{item:mult-conv-finish} of \cref{thm:add-to-mult}, we invoke \cref{prop:lambda-thresh-norm} on $\multAdvMat$ with $\lambda = 3-(4/(1-\varepsilon))^{1/3}$ to get $\eta = ((1-\varepsilon)/4)^{1/3}$ and use these as our $(\lambda, \eta)$ pair.\footnote{We note that this is the same choice as in \cite{DBLP:conf/coco/AmbainisMRR11} and it gives the tightest final bound.}
    Plugging in this choice of $\lambda$ and $\eta$ gives:
    \begin{align*}
        \ln(1+(\lambda-1)(\sqrt{1-\varepsilon} - \sqrt{\eta})^2) &= \ln \bigg[1+\bigg(2-\big(\frac{4}{1-\varepsilon}\big)^{1/3}\bigg)\bigg(\sqrt{1-\varepsilon} - \big({\frac{1-\varepsilon}{4}}\big)^{1/6}\bigg)^2\bigg]\\
        &=\ln \big[ 1- \big(1-(2(1-\varepsilon))^{1/3}\big)^3\big]\\
        &= \bigg(1-2\sqrt{\varepsilon(1-\varepsilon)}\bigg)/ \zeta(\varepsilon)
    \end{align*}
    by definition of $\zeta(\varepsilon)$.
    The function $\zeta$ is plotted in \cref{fig:add-to-mult-loss}.
    For $\varepsilon \in [0,1/3]$ we see that $\zeta$ is convex and that the maximum must be achieved at either $\varepsilon=0$ or $\varepsilon = 1/3$.
    Since $\zeta(1/3)< \zeta(0) < 58$, we can conclude that $\zeta(\varepsilon) < 58$ as desired.

    Since $\multAdvMat$ is a witness to \cref{prop:mult-adv} with $\lambda = 3 - (\frac{4}{1-\varepsilon})^{1/3}$, $\eta = (\frac{1-\varepsilon}{4})^{1/3}$ we can use it to get a query lower bound of:
   {\allowdisplaybreaks[0]
    \begin{align*}
        T &\geq \frac{\ln(1+(\lambda-1)(\sqrt{1-\varepsilon}-\sqrt{\eta})^2)}{\ln(\tr[\multAdvMat \oracle_{i,p}^\dagger\, \rho\, \oracle_{i,p}]/\tr[\multAdvMat \rho])}\\
        &\geq \frac{\ln(1+(\lambda-1)(\sqrt{1-\varepsilon}-\sqrt{\eta})^2)}{2\max_i  \norm{\addAdvMat \circ D_i}}\qquad\textrm{by part \ref{item:mult-conv-prog}}\\
        &=  \frac{ 1-2\sqrt{\varepsilon(1-\varepsilon)}}{2\,\zeta(\varepsilon)\max_i  \norm{\addAdvMat \circ D_i}}\qquad\textrm{by part \ref{item:mult-conv-finish}}\\
        &\geq T^\pm/\zeta(\varepsilon)\ge T^\pm/58
        \qquad\textrm{by \cref{prop:add-adv-normalized} and part \ref{item:mult-conv-finish}.}
    \end{align*}
    Thus the multiplicative loss in converting from a negative-weights to multiplicative adversary bound is at most $\zeta(\varepsilon)<58$.
    }
\end{proof}

\begin{figure}[t]
    \centering
    \begin{minipage}{0.45\textwidth}
        \centering
        \includegraphics[width=\linewidth]{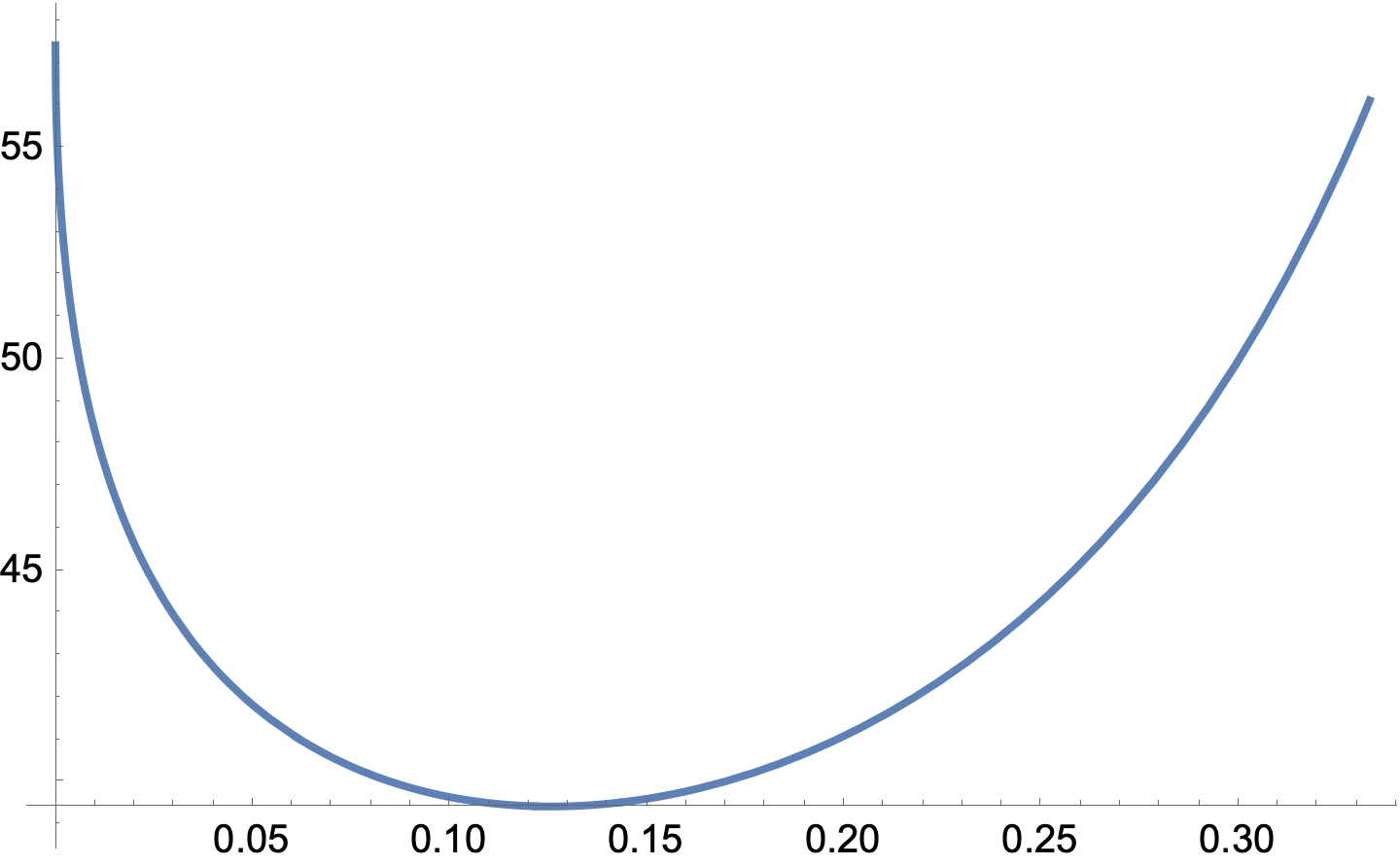}
    \end{minipage}
    \hfill
    \begin{minipage}{0.45\textwidth}
        \centering
        \includegraphics[width=\linewidth]{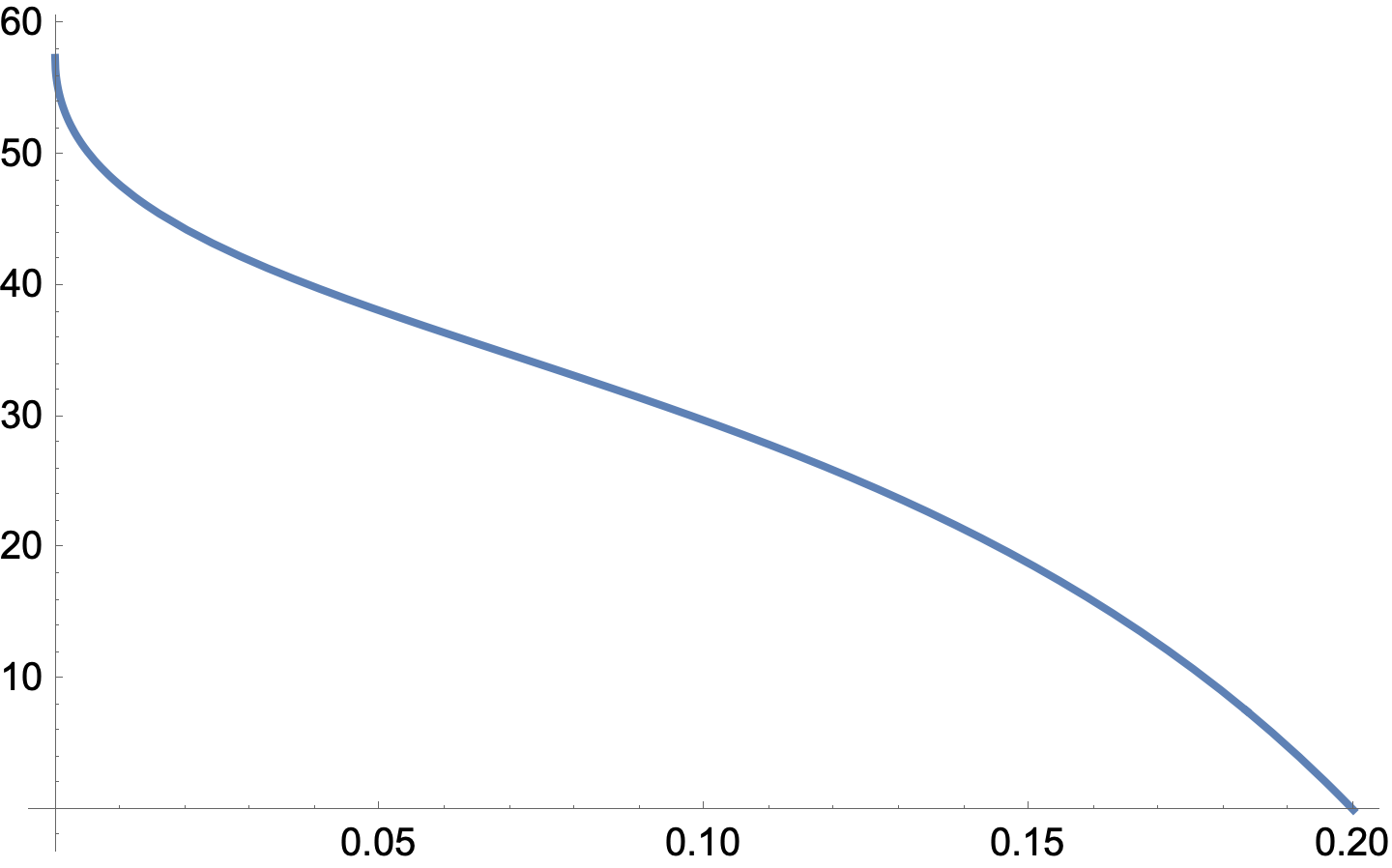}
    \end{minipage}
    \caption{The multiplicative factor loss $\zeta(\varepsilon)$ in our conversion from negative-weights to multiplicative adversary lower bounds as a function of $\varepsilon$,
    which corresponds to the harder-to-meet Boolean output target of the negative-weights adversary method, is given on the left.
    The right is what we would get if we did the same analysis with the weaker non-Boolean output stopping condition; both have the same maximum value at $\varepsilon=0$.
    }\label{fig:add-to-mult-loss}
\end{figure}

\subsection{Improved negative-weights to multiplicative conversion for Boolean-valued functions}
\label{sec:neg-to-mult-bool}

The general idea of this improved conversion is the following fact that negative-weights adversary matrices for Boolean-valued functions
have the following very specific eigenvalue property
that is not available in the general case.

\begin{lemma}\label{lem:bool-add-lambda-eta}
    Let $\addAdvMat$ be a negative-weights adversary matrix for a Boolean valued function $f:\inputspace \to \{0,1\}$ with $\inputspace \subseteq \domain^n$.
    Then for every $\lambda^\pm \in (0,\norm{\addAdvMat}]$,
        $\displaystyle
        \norm{\Pi_{f^{-1}(0)} \Pi_{\addAdvMat, \geq \lambda^\pm}}^2 = \norm{\Pi_{f^{-1}(1)} \Pi_{\addAdvMat, \geq \lambda^\pm}}^2 = 1/2$.
\end{lemma}

\begin{proof}
    Since $\addAdvMat$ is Hermitian and $\addAdvMat_{x,x'} = 0$ whenever $f(x) = f(x')$,
    by permuting the rows and the columns of $\addAdvMat$, we can express it in some basis where:
    \begin{displaymath}
        \addAdvMat = 
        \begin{bmatrix}
        0 & M\\
        M^\dagger & 0
        \end{bmatrix}
    \end{displaymath}
    and the first block of rows and columns correspond to inputs in $f^{-1}(0)$ while the second block of rows and columns correspond to inputs in $f^{-1}(1)$.
    Fix $b \in \{0,1\}$, let $v$ be any unit norm eigenvector of $\addAdvMat$ with eigenvalue $\lambdapmtemp \geq \lambda^\pm>0$.
    Let $\vsub{b}$ be shorthand for $v_{f^{-1}(b)}$, the entries of $v$ indexed by inputs in $f^{-1}(b)$.
    Thus
    \begin{displaymath}
        \addAdvMat v = \begin{bmatrix}
            0 & M \\
            M^\dagger & 0
        \end{bmatrix} \begin{bmatrix}
            \vsub{0} \\ \vsub{1}
        \end{bmatrix} = \lambdapmtemp \begin{bmatrix}
            \vsub{0} \\ \vsub{1}
        \end{bmatrix}.
    \end{displaymath}
    This implies that $M \vsub{1} = \lambdapmtemp \vsub{0}$ and $M^\dagger \vsub{0} = \lambdapmtemp \vsub{1}$.
    Then
    \begin{displaymath}
        \lambdapmtemp \norm{\vsub{0}}^2 =  \vsub{0}^\dagger (\lambdapmtemp\vsub{0}) = \vsub{0}^\dagger M \vsub{1} = (M^\dagger \vsub{0})^\dagger \vsub{1} = (\lambdapmtemp \vsub{1}^\dagger) \vsub{1}  = \lambdapmtemp \norm{\vsub{1}}^2
    \end{displaymath}
    and so $\norm{\vsub{0}}^2 = \norm{\vsub{1}}^2$
    since $\lambdapmtemp \neq 0$.
    Since $1 = \norm{v}^2 = \norm{\vsub{0}}^2 + \norm{\vsub{1}}^2$, we can conclude that $\norm{\vsub{0}}^2 = \norm{\vsub{1}}^2 = 1/2$.

    Now let $V$ be the span of all eigenspaces of $\addAdvMat$ with eigenvalues at least $\lambda^\pm > 0$.
    We will show that any orthogonal eigenbasis $v^1, \ldots, v^\ell$ of $V$ with eigenvalues $\lambda_1, \ldots, \lambda_\ell$ has the property that $\vsub{0}^i$ and $\vsub{0}^j$ are orthogonal for all $i \neq j$.
    Due to the orthogonality of $v^i$ and $v^j$, it must be the case that $(\vsub{0}^i)^\dagger \vsub{0}^j = - (\vsub{1}^i)^\dagger \vsub{1}^j$.
    Since $M^\dagger \vsub{0}^i = \lambda_i \vsub{1}^i$ and $M \vsub{1}^j = \lambda_j \vsub{0}^j$, we know that
    \begin{displaymath}
        \lambda_j (\vsub{0}^i)^\dagger \vsub{0}^j = (\vsub{0}^i)^\dagger M \vsub{1}^j = (M^\dagger \vsub{0}^i)^\dagger \vsub{1}^j = \lambda_i (\vsub{1}^i)^\dagger \vsub{1}^j =  - \lambda_i (\vsub{0}^i)^\dagger \vsub{0}^j.
    \end{displaymath}
    Therefore $(\lambda_i+\lambda_j)(\vsub{0}^i)^\dagger \vsub{0}^j=0$
    and since $\lambda_i,\lambda_j\ge\lambda^\pm > 0$, the vectors $\vsub{0}^i$ and $\vsub{0}^j$ must be orthogonal.

    We prove $\norm{\Pi_{f^{-1}(0)} \Pi_{\addAdvMat \geq \lambda^\pm}}^2 = 1/2$ and note that the other case follows by a symmetric argument.
    Let $u \in V$ be a unit norm vector that maximizes $\norm{\Pi_{f^{-1}(0)} u}$.
    We can express $u = \sum_i \alpha_i v^i$ where the $v^i$ are the orthonormal eigenbasis of $V$ and $\sum_i \abs{\alpha_i}^2 = 1$.
    Then
    \begin{align*}
        \norm{\Pi_{f^{-1}(0)} \Pi_{\addAdvMat, \geq \lambda^\pm}}^2 &= \norm{\Pi_{f^{-1}(0)} u}^2\\
        &=\norm{\Pi_{f^{-1}(0)} \sum_{i \in [\ell]} \alpha_i v^i}^2\\
        &= \norm{\sum_{i \in [\ell]} \alpha_i \vsub{0}^i}^2 \qquad \textrm{since $\Pi_{f^{-1}(0)}$ keeps exactly $\vsub{0}^i$}\\
        &=\sum_{i \in \ell} \abs{\alpha_i}^2 \norm{\vsub{0}^i}^2 \qquad \textrm{since the $\vsub{0}^i$ are mutually orthogonal}\\
        &= \sum_{i \in [\ell]} \abs{\alpha_i}^2 /2 \qquad \textrm{since each $\norm{\vsub{0}^i}^2 = 1/2$}\\
        &= 1/2.\qedhere
    \end{align*}
\end{proof}

\begin{lemma}\label{lem:boolean-mult-lambda-eta}
Let $\multAdvMat = 2 \identity - \addAdvMat$ where $\addAdvMat$ is a negative-weight adversary matrix for a Boolean valued function $f: \inputspace \to \{0,1\}$ with $\inputspace \subseteq \domain^n$.
Then for every $\lambda < 2$, $\multAdvMat$ obeys the conditions of \cref{prop:mult-adv} with $\eta = 1/2$.
\end{lemma}

\begin{proof}
    Let $\lambda^\pm = 2-\lambda$.
    Since $\Pi_{\addAdvMat, \geq \lambda^\pm} = \Pi_{\multAdvMat, \leq \lambda}$, this is an immediate consequence of \cref{lem:bool-add-lambda-eta}.
\end{proof}

We can use this very simple property to replace
the use of \cref{prop:lambda-thresh-trace,prop:lambda-thresh-norm} in the
argument that proved~\cref{thm:add-to-mult}.
In the following statement, the choice of 
$\lambda=19/10$ is somewhat arbitrary:  We need
$\lambda<2$ and the bound grows monotonically, though only marginally, as $\lambda$ approaches 2.

\begin{theorem}\label{lem:bool-add-to-mult}
    Let $f: \inputspace \to \{0,1\}$ be a function and $\dist$ be a distribution over $\inputspace \subseteq \domain^n$.
    Suppose that $\addAdvMat$ is a witness to \cref{prop:add-adv-normalized} for computing $f$ on $\dist$ with error $\varepsilon \in [0,1/2)$ that gives a query lower bound of $T^\pm$.
    Let $\multAdvMat = 2 \identity - \addAdvMat$ and define progress measures $\progress_{\addAdvMat}^t = \tr[\addAdvMat \rho_\inputreg^t]$ and $\progress^t = \tr[\multAdvMat \rho_\inputreg^t]$.
    Then $\multAdvMat$ is a witness to the multiplicative adversary method given in \cref{prop:mult-adv} with $\lambda = 19/10$, $\eta = 1/2$ for computing $f$ on $\dist$.
    Moreover,
    \begin{enumerate}
        \item\label{item:bool-mult-conv-init} $\progress^0 = \progress_{\addAdvMat}^0 =1$
        \item\label{item:bool-mult-conv-prog} $\ln(\progress^{t+1}/\progress^t) \leq \max_{i,p,\rho} \ln \left(\frac{\tr[\multAdvMat \oracle_{i,p}^\dagger\, \rho\, \oracle_{i,p}]}{\tr[\multAdvMat \rho]}\right) \leq  2 \max_i \norm{\addAdvMat \circ D_i}$
        \item\label{item:bool-mult-conv-finish} $\ln(1+(\lambda-1)(\sqrt{1-\varepsilon} - \sqrt{\eta})^2)=\ln(1+9(\sqrt{1-\varepsilon} - 1/\sqrt{2})^2/10) = (1-2\sqrt{\varepsilon(1-\varepsilon)})/\zeta_2(\varepsilon)$ where
        $\zeta_2(\varepsilon)=\frac{1-2\sqrt{\varepsilon(1-\varepsilon)}}{\ln(1+9(\sqrt{1-\varepsilon} - 1/\sqrt{2})^2/10)}<14$.
    \end{enumerate}
    Thus, $\multAdvMat$ gives a query lower bound for
    computing $f$ on distribution $\dist$ with error
    at most $\varepsilon$ of
    \begin{displaymath}
        T \geq \frac{\ln(1 + 9(\sqrt{1-\varepsilon} - 1/\sqrt{2})^2/10)}{\ln (\progress^{t+1}/\progress^t)} \geq \frac{1-2\sqrt{\varepsilon(1-\varepsilon)}}{2\,\zeta_2(\varepsilon) \max_{i}\norm{\addAdvMat \circ D_i}} \geq T^\pm/\zeta_2(\varepsilon)\ge T^\pm / 14
    \end{displaymath}
    via that multiplicative adversary method.
\end{theorem}

\begin{figure}[t]
    \centering
    \includegraphics[width=0.45\linewidth]{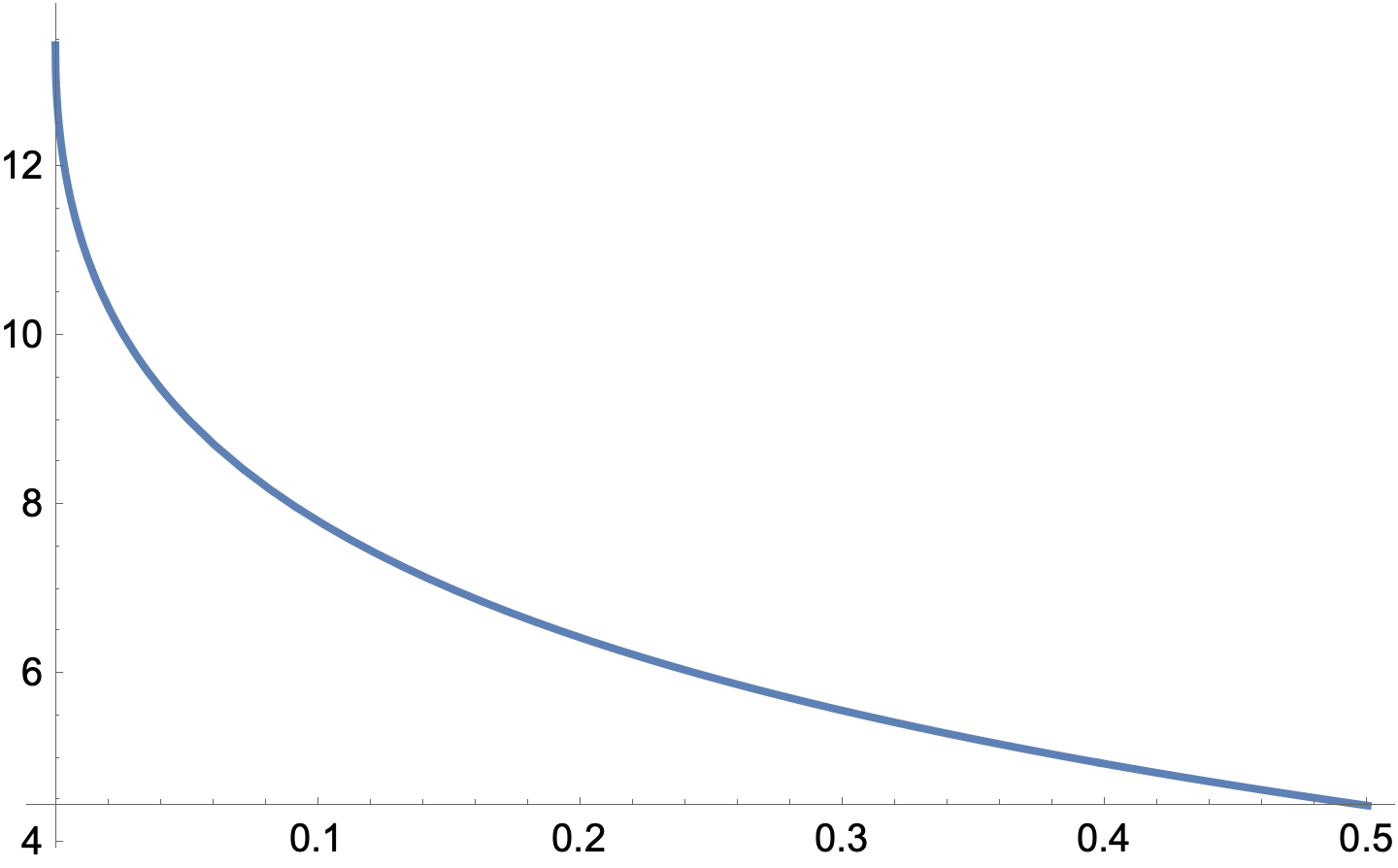}
    \caption{The multiplicative factor loss $\zeta_2(\varepsilon)$ as a function of $\varepsilon$ in our modified conversion from negative-weights to multiplicative adversary lower bounds with $\eta = 1/2$.}
    \label{fig:better-bool-loss}
\end{figure}

\begin{proof}
    By \cref{lem:boolean-mult-lambda-eta} we know that $\multAdvMat$ satisfies the conditions 
    of \cref{prop:mult-adv} for these values of $\lambda$ and $\eta$.
    Parts \ref{item:bool-mult-conv-init} and \ref{item:bool-mult-conv-prog} of this theorem follow by the exact same arguments as in the proof of \cref{thm:add-to-mult}.
    To prove part \ref{item:bool-mult-conv-finish}, we plug in our choice of $\lambda = 19/10$ and $\eta=1/2$ to get:
    \begin{align*}
        \ln(1+(\lambda-1)(\sqrt{1-\varepsilon} - \sqrt{\eta})^2) &= \ln \left(1+9(\sqrt{1-\varepsilon} -1/\sqrt{2})^2/10\right)\\
        &=(1-2\sqrt{\varepsilon(1-\varepsilon}))/\zeta_2(\varepsilon),
    \end{align*}
    by definition of $\zeta_2$.
     The function $\zeta_2(\varepsilon)$ is plotted in \cref{fig:better-bool-loss}.
    In particular, $\zeta_2(\varepsilon)$ is a decreasing function of $\varepsilon$ with maximum value $\zeta_2(0) < 14$ and hence
    part \ref{item:bool-mult-conv-finish} follows.

    Since $\multAdvMat$ is a witness to \cref{prop:mult-adv} with values $\lambda=19/10$ and $\eta=1/2$, as in the proof of \cref{thm:add-to-mult}, applying that multiplicative adversary bound we get a query lower bound of:
    \begin{align*}
        T &\geq \frac{\ln(1+(\lambda-1)(\sqrt{1-\varepsilon}-\sqrt{\eta})^2)}{\ln(\tr[\multAdvMat \oracle_{i,p}^\dagger\, \rho\, \oracle_{i,p}]/\tr[\multAdvMat \rho])}\\
        &\geq \frac{\ln(1+9(\sqrt{1-\varepsilon}-1/\sqrt{2})^2/10)}{2\max_i  \norm{\addAdvMat \circ D_i}}\qquad\textrm{by part \ref{item:bool-mult-conv-prog}}\\
        &=  \frac{ 1-2\sqrt{\varepsilon(1-\varepsilon)}}{2\,\zeta_2(\varepsilon)\max_i  \norm{\addAdvMat \circ D_i}}\qquad\textrm{by definition of $\zeta_2$}\\
        &\ge T^\pm/\zeta_2(\varepsilon)\ge T^\pm/14\qquad\textrm{by \cref{prop:add-adv-normalized} and part \ref{item:bool-mult-conv-finish}.}\qedhere
    \end{align*}
\end{proof}

\subsection{Universality of the relational multiplicative adversary for functions}\label{sec:universal-mult}

We now combine the various pieces of the argument, beginning with the optimality of the negative weight adversary method shown in \cite{DBLP:conf/focs/LeeMRSS11}.

\negWeightTight*

As an immediate consequence of \cref{thm:add-to-mult} we have the following connection between the multiplicative adversary method and quantum query complexity. 

\begin{sloppypar}
\begin{theorem}
\label{thm:mult-optimal}
There is a constant $C>0$ such that the following holds:
    For any function $f:\inputspace \rightarrow \range$ where $\inputspace \subseteq \domain^n$ 
    there is a distribution $\dist$ on $\inputspace$ such that 
     $Q_{1/10}(f)$ is at most $C\cdot\madv^\dist_{\varepsilon,\lambda,\eta}(f)$
     for $\varepsilon=1/10$, 
    $\lambda = 3 - (\frac{4}{1-\varepsilon})^{1/3}\approx 1.356$, and
    $\eta = (\frac{1-\varepsilon}{4})^{1/3}\approx 0.608$.
\end{theorem}
\end{sloppypar}

\begin{proof}
    Let $\dist$ be the distribution such that
    $\negAdvBound_{1/10}(f)=\negAdvBound_{1/10,\dist}(f)$.
    By \cref{prop:reichardt-adv-upperbound}, there is a constant $C'$ such that $$Q_{1/10}(f) \leq C' \cdot\negAdvBound_{1/10}(f) =C' \cdot\negAdvBound_{1/10,\dist}(f).$$
    By \cref{thm:add-to-mult}, with $\varepsilon=1/10$,
    $\lambda = 3 - (\frac{4}{1-\varepsilon})^{1/3}$ and $\eta = (\frac{1-\varepsilon}{4})^{1/3}$ we have
    $\madv^\dist_{\varepsilon,\lambda,\eta}(f)\ge \negAdvBound_{1/10,\dist}(f)/58$.
   Setting $C=58\,C'$ yields the claim.
\end{proof}

With our improved conversion from \cref{sec:neg-to-mult-bool} we obtain the following sharper version which has the main advantage of using $\eta=1/2$ for functions with
Boolean-valued output.

\begin{theorem}\label{thm:bool-mult-bound-is-tight}
    There is a constant $C>0$ such that the following holds:
    For any function $f:\inputspace \rightarrow \{0,1\}$ where $\inputspace \subseteq \domain^n$ 
    there is a distribution $\dist$ on $\inputspace$ such that 
     $Q_{1/10}(f)$ is at most $C\cdot\madv^\dist_{\varepsilon,\lambda,\eta}(f)$
     for $\varepsilon=1/10$, 
    $\lambda = 19/10$, and
    $\eta = 1/2$.
\end{theorem}

\begin{proof}
    Let $\dist$ be the distribution such that
    $\negAdvBound_{1/10}(f)=\negAdvBound_{1/10,\dist}(f)$.
    Again, by \cref{prop:reichardt-adv-upperbound}, there is a constant $C'$ such that $Q_{1/10}(f) \leq C' \cdot\negAdvBound_{1/10}(f)  =C' \cdot\negAdvBound_{1/10,\dist}(f)$.
    Then by \cref{lem:bool-add-to-mult} with $\varepsilon=1/10$,
    $\lambda = 19/10$, and $\eta = 1/2$, we have
    $\madv^\dist_{\varepsilon,\lambda,\eta}(f)\ge \negAdvBound_{1/10,\dist}(f)/14$.
    Setting $C=14C'$ yields the claim.
\end{proof}

\section{Strong selective direct product theorems for quantum query\\ complexity}\label{sec:strong-select}
\subsection{Strong selective direct product theorems for specific functions}
\label{sec:specific}

The starting point for proving these theorems are simple (apparently new) proofs of quantum query lower bounds for the $\search_n$, $\OR_n$, and $\Parity_n$ functions inspired by the recording query method. 
The basic idea of these single copy lower bounds is the following:
\begin{itemize}
    \item We can decompose the space of
    quantum states for the input registers
    into orthogonal components $V^0$ and
    $V^1$ such that for states in $V^0$, the 
    required output is completely
    unpredictable.
    \item  The initial state, some form of uniform superposition, is in $V^0$.
    \item We can strongly upper bound the 
    total amplitude that gets 
    moved to $V^1$ from states in $V^0$ as a function of the number
    of queries made.
\end{itemize}

Though we explicitly track the states of the input registers as recording query arguments do, we do not use the per-coordinate $\perp$ symbol to track individual un-queried coordinates as an indication of the algorithm's lack of knowledge, 
Instead, we use the size of the projection on the full subspace of states in $V^0$ 
to track the algorithm's lack of knowledge. 

We give the details of these arguments in~\cref{sec:single-copy-search-or}.

\subsubsection{Strong selective direct product theorem for \texorpdfstring{$\search$}{SEARCH}}
\label{sec:search}

Let $\search_n$ be the promise problem defined on the set of binary strings $x\in \{0,1\}^n$
of weight 1.

\selectivesearch*


Following the single copy lower bound in~\cref{sec:single-copy-search-or}, we define $\ket{\chi_0} = \frac{1}{\sqrt{n}}\sum_{j \in [n]} \ket j$ as the state corresponding to the input register, a uniform superposition over inputs of Hamming weight 1 and define the projector $\Pi_{V^0} = \identity_\workreg \otimes (\op{\chi_0}{\chi_0})_\inputreg$ onto $V^0=\spn{\ket{\chi_0}}$.
The proof for a single copy breaks down the (square root of the) success probability of the algorithm as
\begin{align*}
    \norm{\pisucc \ket{\psi_t} }&=  \norm{\pisucc (\identity-\Pi_{V^0}) \ket{\psi_t} + \pisucc  \Pi_{V^0} \ket{\psi_t} }\\
    &\leq \norm{ (\identity-\Pi_{V^0}) \ket{\psi_t}} + \norm{ \pisucc  \Pi_{V^0} \ket{\psi_t} }
\end{align*}
The second term was bounded in a straightforward manner; since the input register is forced to be uniform, the success probability is no better than a random guess. 
Bounding the first term amounts to inductively bounding how much a single quantum query can ``push'' us away from the uniform state in the input register.
We extend this methodology to multiple copies of $\search_n$.

The initial state of the quantum query algorithm for $\easierdp{\search_n}{k}{m}$ is
$\ket{\psi_0}=\ket{(0,0), 0, 0}\ket{\chi_0}^{\otimes m}$.

For any $Q\subseteq [m]$, let $\Pi_Q = \bigotimes_{i \in Q}(\identity - \Pi_{V^0})_i \bigotimes_{i \not \in Q} (\Pi_{V^0})_i$ be the projection onto $V^0=\spn\{\ket{\chi_0}\}$ for those coordinates not in $Q$, and the projection onto the orthogonal complement $V^1$ of $V^0$ otherwise, acting on the input register $\calI$.
Let $\basis_1$ be an orthonormal basis for $V^1$.

Let $\Pi_\ell := \sum_{Q \subseteq [m], |Q| = \ell} \Pi_Q$.
Intuitively, $\Pi_\ell$ is the projection onto states that have learned something about exactly $\ell$ out of the $m$ copies of our $\search_n$ relation.
Note that the $\Pi_Q$ are mutually orthogonal, so the $\Pi_\ell$ are also orthogonal and we have
 $\sum_{\ell = 0}^m \Pi_\ell = \identity$.
Define $\Pi_{\le \ell}:= \sum_{j\le \ell} \Pi_j = \sum_{Q\subset [m], |Q|\le \ell} \Pi_Q$.

 Let $\pisucc$ be as defined in \eqref{eq:output-k}. 
        Specifically, given $\ket{i,j,p,w}_\calW$, there is an input-independent answer $q(w) = (K(w),\tau(w))$ with $\tau(w)\in [n]^{K(w)}$ specifying the purported values of $\search_n(x_j)$ for each $j\in K(w)$; $\pisucc = \sum_w \Pi_w\otimes \Pi_{K(w),\tau(w)}$ projects onto those states where these answers are correct.
We will bound our success probability as
\begin{align}
    \norm{\pisucc \ket{\psi_T}}^2 
    &= \norm{\pisucc  \sum_{\ell} \Pi_\ell \ket{\psi_T}}^2
    \leq\bigg(\norm{\pisucc \Pi_{\le \floor{\beta k}} \ket {\psi_T}}+ \sum_{\ell > \beta k}\norm{ \Pi_\ell \ket {\psi_T}} \bigg)^2\label{eq:success-beta-bound}
\end{align}
for  $\beta$ satisfying the conditions of \cref{thm:search-sqdpt}.
Bounding the first term will be relatively straightforward, so we do that first.
To bound the second term, we will argue inductively that a single quantum query cannot push us too far away from the uniform input distribution on too many coordinates.

\begin{lemma}
    \label{lemma:low-succ-small-ell}
    \begin{math}
        \norm{\pisucc \Pi_{\le \ell} \ket{\psi_T}} \leq \binom{k}{\leq \ell} (1/\sqrt{n})^{k-\ell}.
    \end{math}
\end{lemma}

\begin{proof}
    We have 
    \begin{align*}
        \norm{\pisucc \Pi_{\le \ell}& \ket \psi_T}^2 \\&= \norm{\pisucc \sum_{|Q| \le \ell} \Pi_Q \ket {\psi_T}}^2 \\
        &= \bignorm{\sum_{i,j,p,w} \sum_{|Q| \le \ell}  \sum_{v\in \basis_1^Q} \alpha_{i,j,p,w,v} \pisucc  \ket{i,j,p,w}\ket{v}_Q\ket{\chi_0}_{\overline{Q}} }^2\\
        \noalign{\textrm{for some $\alpha_{i,j,p,w,v}$ with $\sum_{i,j,p,w,v}|\alpha_{i,j,p,w,v}|^2= 1$,}}
        &\leq \bignorm{\sum_{i,j,p,w} \ \sum_{\substack{A \subseteq K(w) \\ |A| \geq k-\ell}} \sum_{\substack{Q \subseteq [m] \\ K(w) \setminus Q = A}}\  \sum_{v\in \basis_1^Q} \alpha_{i,j,p,w,v} \ket{i,j,p,w} \Pi_{K(w), \tau(w)} \ket{v}_Q\ket{\chi_0}_{\overline{Q}} }^2,\\
        \noalign{\text{since $\Pi_{succ}=\sum_w \Pi_w\otimes \Pi_{K(w),\tau(w)}$ and we have only increased the number of allowed $Q$,}}
        &
        \leq \sum_{i,j,p,w}\bignorm{ \sum_{\substack{A \subseteq K(w) \\ |A| \geq k-\ell}}\,\sum_{\substack{Q\subseteq [m]\\  K(w)\setminus Q = A} } \sum_{v\in \basis_1^Q} \alpha_{i, j, p, w, v} \ket{i,j,p,w}\Pi_{K(w),\tau(w)}\ket{v}_Q\ket{\chi_0}_{\overline{Q}} }^2, \\
        \noalign{\text{since states for distinct $i,j,p,w$ are orthogonal,}}
        &\leq \sum_{i,j,p,w}\bignorm{ \sum_{\substack{A \subseteq K(w) \\ |A| \geq k-\ell}}\,\sum_{\substack{Q\subseteq [m]\\  K(w)\setminus Q = A} }  \sum_{v\in \basis_1^Q} \alpha_{i, j, p, w, x_Q} \ket{i,j,p,w}\Pi_{A,\tau(w)_{A}}\ket{v}_Q\ket{\chi_0}_{\overline{Q}} }^2, \\
        \noalign{\text{since $\Pi_{A,\tau(w)_{A}}$ requires success only on a subset of the coordinates in $K(w)$,}}
        &\leq \sum_{i,j,p,w}\binom{k}{\leq \ell}^2 \max_{\substack{A \subseteq K(w) \\ |A|\geq k-\ell}}\,\bignorm{  \sum_{\substack{Q\subseteq [m]\\  K(w)\setminus Q = A} } \sum_{v\in \basis_1^Q} \alpha_{i, j, p, w, v} \ket{i,j,p,w}\ket{v}_Q\Pi_{A, \tau(w)_{A}}\ket{\chi_0}_{\overline{Q}} }^2\\
         \noalign{\text{and, choosing $A_w^*$ to be the maximizing value of $A$, this is}}
        &\leq\binom{k}{\leq \ell}^2 \sum_{i,j,p,w} (1/n)^{k-\ell} \sum_{Q: K(w)\setminus Q = A_w^*} \sum_{v\in \basis_1^Q} |\alpha_{i,j,p,w,v}|^2\\
        &\leq \binom{k}{\leq \ell}^2 (1/n)^{k-\ell}.\qedhere
    \end{align*} 
\end{proof}
Next, we come to the core of the argument, which bounds $\Delta_{t, \ell} := \norm{\Pi_\ell \ket{\psi_t}}$.
We prove the following recursive lemma on $\Delta_{t, \ell}$.
It bounds how much progress can be made (or, counterintuitively, lost) by a single query,
generalizing the argument for the single copy case.

\begin{lemma}
    \label{lem:delta-recurrence}
    $\Delta_{t,\ell} := \norm{ \Pi_\ell \ket {\psi_t}} \leq \Delta_{t-1, \ell} + \frac{2}{\sqrt{n}} \Delta_{t-1,\ell-1} + \frac{2}{
    \sqrt n} \Delta_{t-1,\ell + 1}.$
\end{lemma}

\begin{proof}
     First observe that we have
    \begin{align*}
        \Delta_{t,\ell} &= \norm{\Pi_\ell \calO \sum_\ell \Pi_\ell \ket{\psi_{t-1}}} \\
        &\leq \norm{\Pi_\ell \calO \Pi_{\ell-1} \ket{\psi_{t-1}}} + \norm{\Pi_\ell \calO \Pi_{\ell} \ket{\psi_{t-1}}} + \norm{\Pi_\ell \calO \Pi_{\ell+1} \ket{\psi_{t-1}}},
    \end{align*}
    where we have dropped projectors outside of $\{\Pi_{\ell-1}, \Pi_{\ell}, \Pi_{\ell-1}\}$, since $\calO$ only acts on one coordinate at a time, which implies that $\Pi_\ell \calO \Pi_{\ell'} \ket{\psi} = 0$ for all other $\ell' \not\in \{\ell -1, \ell ,\ell +1\}$.
    The middle term is bounded by $\Delta_{t-1, \ell}$ since $\Pi_\ell\calO$ can only reduce the norm.

    Consider the first term. 
    For each $(i,j)\in [m]\times [n]$, let $\Pi_{i,j}$ project onto states where the query register is $(i,j)$ so $\sum_{i,j}\Pi_{i,j}=\identity$.
    Observe that for $|Q'|=\ell$ and $|Q|=\ell-1$, we have
    $\Pi_{Q'} \calO \Pi_{i,j}\Pi_{Q}=0$ unless $Q'=Q\cup\set{i}$.
    Therefore, 
    $$\Pi_\ell \calO \Pi_{\ell-1}=\sum_{|Q|=\ell-1}\  \sum_{(i,j):\ i\notin Q}\Pi_{Q\cup\, i} \calO \Pi_{i,j} \Pi_{Q}.$$
    Also observe that 
    after applying $\Pi_{Q}$, the $i$'th input register is
    spanned by $\ket{\chi_0}_i$ and that 
    $\calO\Pi_{i,j}$ maps $\ket{\chi_0}_i$ to  $-2/\sqrt{n}\ket{j}_i + \ket{\chi_0}_i$, leaving the rest
    of the state unchanged.

    Therefore, we have
    \begin{align*}
        &\norm{\Pi_\ell \calO \Pi_{\ell-1}\ket{\psi_{t-1}}}^2\\ &= \norm{\sum_{|Q|=\ell-1}\  \sum_{(i,j):\  i\notin Q}\Pi_{Q\cup\, i} \calO \Pi_{i,j} \Pi_{Q}\ket{\psi_{t-1}}}^2 \\
        &=\norm{\sum_{|Q| = \ell-1} \ \sum_{i,j,p,w:\
         i\notin Q}\Pi_{Q\cup\, i}\calO \sum_{v\in \basis_1^Q} \alpha_{i,j,p,w, v} \ket{i,j,p,w} \ket{\chi_0}_i\ket{\chi_0}_{\overline {Q\cup\, i}}\ket{v}_Q}^2\\
        &\qquad \textrm{for some amplitudes $\alpha_{i,j,p,w,v}$ with $\sum_{i,j,p,w,v}|\alpha_{i,j,p,w,v}|^2=\norm{\Pi_{Q} \ket{\psi_{t-1}}}^2 $}\\
        &= \norm{ \sum_{|Q| = \ell-1}\ \sum_{i,j,p,w:\
         i\notin Q}\  \Pi_{Q\cup\, i} \sum_{v\in \basis_1^Q} \alpha_{i,j,p,w, v}\ket{i,j,p,w}\big( \frac{-2}{\sqrt{n}} \ket{j}_i + \ket{\chi_0}_i\big)\ket{\chi_0}_{\overline {Q\cup\, i}}\ket{v}_Q}^2\\
        &= \norm{ \sum_{|Q| = \ell-1}\ \sum_{i,j,p,w:\
         i\notin Q}\ \sum_{v\in \basis_1^Q}   \frac{-2\alpha_{i,j,p,w, v}}{\sqrt{n}}\Pi_{Q\cup\, i}\ket{i,j,p,w} \ket{j}_i \ket{\chi_0}_{\overline {Q\cup\, i}}\ket{v}_Q}^2,\\
         &\qquad \textrm{since $\Pi_{Q \cup i}$ projects out $\ket{\chi_0}_i$}\\
        &\le \frac{4}{n}\sum_{|Q| = \ell-1} \norm{\Pi_{Q} \ket{\psi_{t-1}}}^2\\
        &=\frac{4}{n}\norm{\Pi_{\ell-1}\ket{\psi_{t-1}}}^2=\frac{4}{n}\Delta^2_{t-1,\ell-1}.
    \end{align*}

    We now perform a similar analysis for the last term. 
   Similar to the first term, for $|Q'|=\ell+1$ and $|Q|=\ell$, 
    we have $\Pi_Q \calO\Pi_{i,j}\Pi_{Q'}=0$ unless
    $Q=Q'\setminus \set{i}$.
    Therefore
    $$\Pi_\ell\calO\Pi_{\ell+1}=\sum_{|Q'|=\ell+1}\ \sum_{i,j:\ i\in Q'}\Pi_{Q'\setminus i}\calO \Pi_{i,j}\Pi_{Q'}.$$
    Writing any state $\ket{v}_i \in V^1$ in the standard basis yields
    $\sum_{j}\beta_j \ket{j}_i$ where  $\beta_j=\braket{v|j}$ so 
    applying query $(i,j)$ yields $\ket{v}_i-2\beta_j \ket{j}_i$ (i.e. a reflection around $\ket{j}_i$); 
    projection onto
    $V^0=\spn(\ket{\chi_0})$ in coordinate $i$ then yields $\frac{-2 \beta_j}{\sqrt{n}} \ket{\chi_0}$.
    Therefore,  we have
    \begin{align*}
        &\norm{\Pi_\ell \calO \Pi_{\ell+1}\ket{\psi_{t-1}}}^2\\ &=\norm{\sum_{|Q'|=\ell+1}\  \sum_{(i,j):\  i\in Q'}\Pi_{Q'\setminus i} \calO \Pi_{i,j} \Pi_{Q'}\ket{\psi}}^2 \\
        &=\norm{\sum_{|Q'| = \ell+1}\  \sum_{i,j,p,w:\
         i\in Q'}\Pi_{Q'\setminus i}\calO \sum_{v\in \basis_1^{Q'}} \alpha_{i,j,p,w,v} \ket{i,j,p,w} \ket{v}_{Q'}\ket{\chi_0}_{\overline {Q'}}}^2\\
        &\qquad \textrm{for some amplitudes $\alpha_{i,j,p,w,v}$ with $\sum_{i,j,p,w,v}|\alpha_{i,j,p,w,v}|^2=\norm{\Pi_{Q'} \ket{\psi_{t-1}}}^2 $,}\\
        &=\norm{\sum_{|Q'| = \ell+1}\  \sum_{i,j,p,w:\
         i\in Q'}\Pi_{Q'\setminus i}\calO \sum_{v\in \basis_1^{Q'}} \alpha_{i,j,p,w,v} \ket{i,j,p,w}\ket{v_i}_i \ket{v_{Q'\setminus i}}_{Q'\setminus i}\ket{\chi_0}_{\overline {Q'}}}^2\\
        &= \norm{ \sum_{|Q'| = \ell+1}\ \sum_{i,j,p,w:\
         i\in Q'}\ \Pi_{Q'\setminus i} \sum_{v\in \basis_1^{Q'}} \alpha_{i,j,p,w, v} \ket{i,j,p,w}\big( -2\beta^{x_i}_j \ket{j}_i + \ket{v_i}_i\big)\ket{v_{Q'\setminus i}}_{Q'\setminus i}\ket{\chi_0}_{\overline {Q'}}}^2\\
         &\qquad\textrm{for $\beta^{v_i}_j=\braket{v_i|j}$}\\
        &=\norm{ \sum_{|Q'| = \ell+1}\ \sum_{i,j,p,w:\
         i\in Q'}\ \sum_{v\in \basis_1^{Q'}}   \frac{-2\alpha_{i,j,p,w, x_{Q'}}\beta^{x_i}_j}{\sqrt{n}}\ket{i,j,p,w} \ket{\chi_0}_i \ket{v_{Q'\setminus i}}_{Q'\setminus i}\ket{\chi_0}_{\overline {Q'}}}^2,\\
         &\qquad \textrm{since $\Pi_{Q' \setminus i}$ projects out $\ket{v_i}_i \in V^1$}\\
         &=\sum_{|Q'| = \ell+1}\ \sum_{i,j,p,w:\
         i\in Q'}\ \sum_{v_{Q'\setminus i}\in \basis_1^{Q'\setminus i}} \left | \sum_{v_i\in \basis_1}\frac{2}{\sqrt n}|\alpha_{i,j,p,w,v}| \cdot\beta_j^{v_i}\right |^2\\
             &\le \frac{4}{n}\sum_{|Q'| = \ell+1}\ \sum_{i,j,p,w:\
         i\in Q'}\ \sum_{v\in \basis_1^{Q'}} |\alpha_{i,j,p,w,v}|^2 \qquad\textrm{by Cauchy-Schwarz since $\sum_{v_i\in \basis_1}|\beta^{v_i}_j|^2=\sum_{v_i\in \basis_1}\braket{v_i|j}^2\le 1$ }\\
        &\le \frac{4}{n}\sum_{|Q'| = \ell+1} \norm{\Pi_{Q'} \ket{\psi_{t-1}}}^2 \\
        &=\frac{4}{n}\norm{\Pi_{\ell+1}\ket{\psi_{t-1}}}^2=\frac{4}{n}\Delta^2_{t-1,\ell+1}.\qedhere
    \end{align*}
\end{proof}

\begin{corollary}
    \label{lem:delta-bound}
    \begin{math}\Delta_{t, \ell} := \norm{\Pi_\ell \ket{\psi_t}} \leq \binom{t }{\ell} (\frac{4}{\sqrt{n}})^\ell\end{math} 
    for $t \geq \ell$ 
    and $\Delta_{t,\ell}=0$ for $t<\ell$. 
\end{corollary}

\begin{proof}
    When $t > \ell \sqrt{n}/4$ the bound is trivial, so we can assume that $t \leq \ell \sqrt{n}/4$.
    By induction and \cref{lem:delta-recurrence}, we have that for $t\ge \ell$,
    \begin{align*}
        \Delta_{t,\ell} &\leq \binom{t-1}{\ell} (4/\sqrt{n})^{\ell} + \frac{2}{\sqrt{n}} \bigparen{\binom{t-1}{\ell-1} (4/\sqrt{n})^{\ell-1} + \binom{t-1}{\ell + 1} (4/\sqrt{n})^{\ell+1}} \\
        &= \binom{t}{\ell} \bigparen{\frac{4}{\sqrt{n}}}^\ell \bigparen{ \frac{t - \ell}{t} + \frac{\ell}{t}\cdot \frac{2}{4} + \frac{2\cdot 4 (t-\ell) (t-\ell - 1)}{n \cdot t (\ell + 1)}} \\
        &\leq  \binom{t}{\ell} \bigparen{\frac{4}{\sqrt{n}}}^\ell \bigparen{1 - \frac{1}{2} \cdot \frac{\ell }{t} + \frac{8t^2}{nt\ell}} \\
        &= \binom{t}{\ell} \bigparen{\frac{4}{\sqrt{n}}}^\ell \bigparen{1 + \frac{\ell}{t} (\frac{8t^2}{n\ell^2} - 1/2)} \\ 
        &\leq \binom{t}{\ell} \bigparen{\frac{4}{\sqrt{n}}}^\ell\qquad\textrm{since $t \leq \ell\sqrt{n} /4$.}\qedhere
    \end{align*}
\end{proof}

We now have everything we need in order to prove \cref{thm:search-sqdpt}

\begin{proof}[Proof of \cref{thm:search-sqdpt}]
    Let $\beta\in (0,1/2]$ satisfy
    $2^{2(\beta+H_2(\beta))/(1-\beta)}\le n$.
    Suppose that the quantum algorithm for $\easierdp{\search_n}{k}{m}$ uses at most $ T =(\beta/24) k \sqrt{n}$ queries.
 Then
    \begin{align*}
        \norm{\pisucc \ket{\psi_T}} &\leq  \bignorm{ \pisucc\Pi_{\le \floor{\beta k}} \ket{\psi_T}}+ \sum_{\ell > \beta k} \bignorm{\Pi_\ell \ket{\psi_T}} \\
        &\leq  \binom{k}{\leq \floor{\beta k}} \bigg(\frac1{\sqrt{n}}\bigg)^{k-\floor{\beta k}}+ \sum_{\ell > \beta k} \Delta_{T, \ell} \qquad\text{by \cref{lemma:low-succ-small-ell} and the definition of $\Delta_{T,\ell}$}\\
        &\leq   2^{H_2(\beta)k} \bigg(\frac1{\sqrt{n}}\bigg)^{k-\beta k} +\sum_{\ell > \beta k} \Delta_{T, \ell}  \qquad \text{by \cref{prop:binom-tail} since $\beta\le 1/2$}\\
        &\leq \bigg ( \frac{2^{H_2(\beta)}}{n^{(1-\beta)/2}}\bigg )^{k}+\sum_{\ell > \beta k } \binom{T}{\ell} \bigg(\frac{4}{\sqrt{n}}\bigg)^\ell     \qquad \text{by \cref{lem:delta-bound}}\\
        &\leq   \bigg ( \frac{2^{H_2(\beta)}}{n^{(1-\beta)/2}}\bigg )^{k}+2 \cdot 2^{-\beta k} \qquad \text{by our assumption on $T$} \\
        &\leq 2^{-\beta k}+2 \cdot 2^{-\beta k}  = 3 \cdot 2^{-\beta k}\qquad \text{by our assumption on $\beta$}.
    \end{align*}
    Therefore the probability of the algorithm's success $ \norm{\pisucc \ket{\psi_T}}^2\le
    9\cdot 2^{-2\beta k}$.
\end{proof}

\subsubsection{Strong selective direct product theorem for \texorpdfstring{$\OR$}{OR}}
\label{sec:or}

The hard distribution $\calD$ for $\OR_n$ takes on the all zeroes string with probability $1/2$, and each hamming weight one string with probability $1/(2n)$.
This corresponds to a mixed state.
The most natural purification of this mixed state would be the quantum pure state $\frac{1}{\sqrt{2n}}\sum_{j \in [n]} \ket j + \frac{1}{\sqrt{2}}\ket{0^n} \in \C^{n+1}$.
For ease of analysis, we instead choose the purification that initializes the input register in the superposition $\sqrt{\frac{1}{2n}}\sum_{i \in [2n]}\ket i \in \C^{2n}$, where we have identified $\ket{0^n}$ with the superposition $\frac{1}{\sqrt{n}}\sum_{j=n+1}^{2n} \ket j$.
This will allow us to directly port over much of our analysis from the proof of \cref{thm:search-sqdpt}.
Note that the algorithm will only ever query superpositions of coordinates $1,\ldots, n$ and thus cannot distinguish these choices for how to handle no instances.

\selectiveOR*

The ideas behind the proof are exactly the same as those behind \cref{thm:search-sqdpt}. 
For the remainder of this section, $\ket {\chi_0}$ (where we have dropped the $2n$ superscript for ease of notation) equals $\frac{1}{\sqrt{2n}}\sum_{j \in[2n]} \ket{j}$
We use the same case split as in \eqref{eq:success-beta-bound} with our different constraint on $\beta$.

Bounding the second term is similarly straightforward as before.
\begin{lemma}
    \label{lemma:low-succ-small-ell-or}
    \begin{math}
        \norm{\pisucc \Pi_{\le \ell}\ket{\psi_T}} \leq \binom{k}{\leq \ell} (1/\sqrt{2})^{k-\ell}.
    \end{math}
\end{lemma}

\begin{proof}[Proof Sketch]
    The proof is identical to the proof of~\cref{lemma:low-succ-small-ell}, except in the last two lines where the probability of
    a fixed answer being correct on $\ket{\chi_0}$ is $1/2$ for $\OR_n$ instead of $1/n$ for
    $\search_n$.
\end{proof}

We have already done all the work we needed to do to bound the first term.
In particular, the input distribution is exactly the same for $\OR_n$ as for $\search_{2n}$.
We therefore have the following immediate corollary.
\begin{corollary}
    \label{cor:delta-bound-OR}
    For $\ell \leq t \leq \ell \sqrt{2n}/4$, we have
    $\norm{\Pi_\ell \ket{\psi_t}}\leq \binom{t}{\ell} \big(\frac{4}{\sqrt{2n}}\big)^\ell$.
\end{corollary}

We now have what we need to prove our selective direct product theorem for $\OR_n$. 

\begin{proof}[Proof of \cref{thm:selectiveOR}]
    Suppose that $\beta>0$ satisfies $3\beta+2H_2(\beta)\le 1$ (and hence $\beta< 1/2$) and that
    the quantum algorithm for $\easierdp{\OR_n}{k}{m}$ uses at most $ T =\frac{\beta}{16} k \sqrt{n}$ queries.
  Then we have  
    \begin{align*}\norm{\pisucc \ket{\psi_T}} &\leq  \bignorm{ \pisucc\Pi_{\le \floor{\beta k}} \ket{\psi_T}}+ \sum_{\ell > \beta k} \bignorm{\Pi_\ell \ket{\psi_T}} \\
        &\leq  \binom{k}{\leq \floor{\beta k}} \bigg(\frac1{\sqrt{2}}\bigg)^{k-\floor{\beta k}}+ \sum_{\ell > \beta k} \Delta_{T, \ell} \qquad\text{by \cref{lemma:low-succ-small-ell-or} and the definition of $\Delta_{T,\ell}$}\\
        &\leq \bigg ( \frac{2^{H_2(\beta)}}{2^{(1-\beta)/2}}\bigg )^{k}+\sum_{\ell > \beta k } \binom{T}{\ell} \bigparen{\frac{4}{\sqrt{2n}}}^\ell \qquad\text{by \cref{prop:binom-tail} and \cref{cor:delta-bound-OR}}\\
        &\le 2^{-\beta k} + \sum_{\ell > \beta k} \bigparen{\frac{4eT}{\ell \sqrt{2n}}}^\ell \qquad \text{since $3\beta+2 H_2(\beta)\le 1$} \\
        &\le 2^{-\beta k} + \sum_{\ell > \beta k} 2^{-\ell}
        \qquad\text{since $T \leq \beta k \sqrt{n}/16$}\\
        &\leq 2^{-\beta k} +2 \cdot 2^{-\beta k}  \leq 3 \cdot 2^{-\beta k} 
    \end{align*}
    so the probability of the algorithm's success $\norm{\pisucc \ket{\psi_T}}^2\le 9 \cdot 2^{-2\beta k}$.
\end{proof}

\subsubsection{Strong selective direct product theorem for \texorpdfstring{$\Parity$}{OR}}
\label{sec:parity}


In order to prove \cref{thm:selectiveparity}, the following proposition, which generalizes \cref{claim:parity-key}, will be key.

\begin{lemma}
    \label{prop:superposition-parity-dp}
    Let $\ket \psi$ be a quantum state on $ k\geq \ell$ registers, each containing $n$ qubits, such that when written in the Hadamard $\{ \ket{+}, \ket -\}$ basis, $\ket{\psi}$ is a superposition over quantum states with strictly fewer than $n/2$ $\ket{-}$'s in each of the first $\ell$ registers.
    For $\tau \in \bits^\ell$, let $\Pi_\tau$ be the projection operator that projects onto standard basis states $x^{1},\ldots,x^{\ell}$ in the first $\ell$ registers such that $\Parity_n(x^{j}) = \tau_j$ for all $j \in \ell$.
    Then we have
    $\norm{\Pi_\tau \ket{\psi}}^2 = 2^{-\ell}$.
\end{lemma}

\begin{proof}
    Recall the Pauli matrix
    $
    Z = \begin{bmatrix}
        1 & 0 \\
        0 & -1
    \end{bmatrix}$.
    Observe that for a single register with $\tau_j \in \bits$, the operator $\frac{1}{2}\bigparen{ \identity + (-1)^{\tau_j} Z^{\otimes n}}$
    projects onto states $\ket{x}$ such that $\Parity_n(x) = \tau_j$.
    Thus, we can write 
    $$\Pi_\tau = 2^{-\ell} \bigotimes_{j = 1}^\ell \big(\identity + (-1)^{\tau_j}Z^{\otimes n}\big)_j \bigotimes_{j = \ell+1}^k \identity_j = 2^{-\ell}\sum_{y \in \bits^\ell} (-1)^{\tau \cdot y} \bigotimes_{j = 1}^\ell \big(Z^{\otimes n}\big)^{y_j}_j \bigotimes_{j = \ell+1}^k \identity_j.$$
     Therefore, the quantity we wish to bound is
    $$\norm{\Pi_\tau \ket{\psi}}^2 = 2^{-\ell}\sum_{y \in \bits^\ell} \bra{\psi}(-1)^{\tau \cdot y} \bigotimes_{j = 1}^\ell \big(Z^{\otimes n}\big)^{y_j}_j \bigotimes_{j = \ell+1}^k \identity_j \ket \psi.$$

    Observe that when $y = 0$, the summand is exactly equal to $2^{-\ell}$. 
    Otherwise, when $y \neq 0$, we claim that the summand is 0. 
    This is because, in the Hadamard basis, $Z$ acts as a swap operator: $Z \ket{+} = \ket -$ and $Z \ket{-} = \ket{+}$.
    Therefore, using the assumption that $\ket{\psi}$ is a superposition over states with strictly fewer than $n/2$ $\ket-$'s in the first $\ell$ registers, after applying $Z^{\otimes n}$ to at least one of those registers, we are left with a state that is only supported on basis states with at least $n/2$ $\ket{-}$'s, which is orthogonal to $\psi$ by assumption.
    Thus, the above sum is exactly $2^{-\ell}$, as desired.
\end{proof}

We are now ready to prove \cref{thm:selectiveparity}, which we restate here for convenience.

\selectiveparity*

\begin{proof}[Proof of \cref{thm:selectiveparity}]
    Our hard input distribution will be $\ket{+}^{\otimes ([m]\times [n])}$, which corresponds to a uniformly random
    choice of input vectors where coordinate $(i,j)$ corresponds to the $j$-th coordinate of the $i$-copy.
    Then, our quantum algorithm after $T$ quantum queries looks like 
    $$\ket{\psi_T} = \sum_{i,j,p,w,|S| \leq T} \alpha_{i,j,p,w,S} \ket{i,j,p,w}\ket{+}_{\overline S} \ket{-}_S,$$
    for $S \subseteq [m] \times [n]$, and we have the natural partition $S = S_1 \sqcup \cdots \sqcup S_m$, where $S_i := S \cap \{i\}\times [n]$.
    For any $S$, define $Q(S):= \{ i \in [m]: |S_i| \geq n/2\}$, and note that $|Q(S)| \leq 2T/n$ after $T$ queries.
    This implies that for all $S$ after $T$ queries, there exists at least one $A \subseteq [m]$ where $|A|\geq |K(w)| - \floor{2T/n}$ such that $ K(w)\setminus Q(S) = A$. 
    Let $\pisucc$ be as defined in \eqref{eq:output-k}. 
        Specifically, given $\ket{i,j,p,w}_\calW$ there is an input-independent answer $q(w) = (K(w),\tau(w))$ with $\tau(w)\in \bits^{K(w)}$ specifying the purported values of $\Parity_n(x_j)$ for each $j\in K(w)$, and $\pisucc = \sum_w \Pi_w\otimes \Pi_{K(w),\tau(w)}$ projects onto those states where this answer is correct.

    Then our success probability $\sigma$ can be bounded as 
    \begin{align*}
         &\norm{\pisucc \ket{\psi_T}}^2\\ &= \bignorm{ \sum_{i,j,p,w, |S| \leq T} \alpha_{i,j,p,w,S} \pisucc \ket{i,j,p,w} \ket{+}_{\overline{S'}}\ket{-}_{S'}}^2  \\
         \noalign{\textrm{for some $\alpha_{i,j,p,w,S}$ with $\sum_{i,j,p,w, |S|\leq T}|\alpha_{i,j,p,w,S}|^2=1$}}
        &= \sum_{i,j,p,w} \bignorm{\sum_{|S|\leq T}  \alpha_{i,j,p,w,S} \pisucc \ket{i,j,p,w} \ket{+}_{\overline{S}}\ket{-}_S}^2 ,\\
        \noalign{\text{since $\pisucc$ preserves orthogonality for distinct $i,j,p,w$,}}
        &\leq \sum_{i,j,p,w} \bignorm{\sum_{\substack{A \subseteq K(w)\\ |A| \geq k-\floor{2T/n}}} \sum_{S: K(w)\setminus Q(S) = A} \alpha_{i,j,p,w,S} \ket{i,j,p,w}\Pi_{K(w),\tau(w)}\ket{-}_S\ket{+}_{\overline S}}^2 \\
        &\leq \sum_{i,j,p,w} \bignorm{\sum_{\substack{A \subseteq K(w)\\ |A| \geq k-\floor{2T/n}}} \sum_{S: K(w)\setminus Q(S) = A} \alpha_{i,j,p,w,S} \ket{i,j,p,w}\Pi_{A,\tau(w)_{A}}\ket{-}_S\ket{+}_{\overline S}}^2 ,\\
        \noalign{\text{since $\Pi_{A,\tau(w)_{A}}$ requires success only on a subset of the coordinates in $K(w)$,}}
        &\leq \sum_{i,j,p,w} \binom{k}{\leq \floor{2T/n}}^2\max_{\substack{A \subseteq K(w)\\ |A| \geq k-\floor{2T/n}}}\bignorm{\sum_{S: K(w)\setminus Q(S) = A} \alpha_{i,j,p,w,S} \ket{i,j,p,w}\Pi_{A,\tau(w)_A}\ket{-}_S\ket{+}_{\overline S}}^2  \\
        \noalign{\text{and, choosing $A_w^*$ to be the maximizing value of $A$, since by definition of 
        $Q(S)$ for every }}
        \noalign{\text{$i\in A_w^*$ there are $<n/2$ elements $(i,j)\in S$, by \cref{prop:superposition-parity-dp}, this is}}
        &\leq \sum_{i,j,p,w}\binom{k}{\leq \floor{2T/n}}^2 2^{-k + \floor{2T/n}} \cdot \sum_{S : K(w)\setminus Q(S) = A_w^*} |\alpha_{i,j,p,w,S}|^2 \\
        &\leq \binom{k}{\leq \floor{2T/n}}^2 2^{-k + \floor{2T/n}}\\
        \noalign{\text{and by \cref{prop:binom-tail} this is}}
        &\leq 2^{2k H_2(2T/(kn)) - k + \floor{2T/n}} \leq 2^{[ 2H_2(2T/(kn)) - 1 +   2T/(kn)]\, k} .
    \end{align*}

    Lastly, if $T \leq kn/32$ then $H_2(2T/(kn)) \leq H_2(1/16) \leq 0.34$, so the success probability is at most $2^{k(0.68 - 1 + 1/16)} \leq 2^{-k/4}$.
\end{proof}

\subsection{Strong selective direct products from multiplicative adversaries}\label{sec:selective-from-mult-adv}

We prove that the multiplicative adversary method in \cref{prop:mult-adv} satisfies a strong selective direct product property. 
This generalizes and refines the arguments of \cite{Spa08} and \cite{DBLP:conf/coco/AmbainisMRR11} 
showing ordinary strong direct products for their versions of the multiplicative adversary method.

\begin{theorem}
\label{thm:selective-from-multiplicative-adversary}
    Let $\relation \subseteq \inputspace \times \outputspace$ for $\inputspace\subseteq \domain^n$ be a relation and $\dist$ be a distribution over $\inputspace$.
    Let $0<\delta\le 1/8$, $\lambda>1$ and $\eta > 0$ satisfy $\eta\le 2^{-H_2(\delta)/(7/8-\delta)}$ where
  $H_2(\cdot)$ is the binary entropy function.
    For $ 1/2- \eta -\sqrt{2\eta}\le \varepsilon <1-\eta$,
     there is a $k_0=k_0(\lambda,\delta)\ge 1$ such that for any integers $m\ge k\ge k_0$, we have
    \begin{equation*}
        Q_{1-\kappa^k}^{\dist^m}(\easierdp{\relation}{k}{m})\ge  \delta k\cdot \madv^\dist_{\varepsilon,\lambda,\eta}(\relation)   \end{equation*}
        for some $\kappa=\kappa(\lambda,\delta)<1$.
\end{theorem}

\begin{remark}
    The lower bound this statement requires on $\varepsilon$ in terms of $\eta$ is somewhat arbitrary, chosen with the sole goal of making the choice of $k_0$ independent of how close $1-\varepsilon-\eta$ is to 1.  The lower bound is only positive for $\eta<1.5-\sqrt{2}\approx 0.085786438
$.  
    Further,
    the values of $k_0$ and $\kappa$ only increase with $\lambda$ as a function of $1/(\lambda-1)$ so for typical examples this dependence is not particularly relevant.
\end{remark}

\begin{proof}
We first sketch the main ideas of the proof.
The multiplicative adversary method involves three
objects:  adversary matrix $\Gamma$, an eigenvalue bound $\lambda\ge 1$ and a corresponding correctness probability $\eta$, such that any input state lying in the span of eigenspaces
of the adversary matrix $\Gamma$ with eigenvalues at most $\lambda$ has squared projection at most $\eta$ on inputs with any fixed output value.

Given these parameters for $\relation$, the key to the argument will be to define corresponding parameters $\Gamma'$, $\lambda'$, and $\eta'$ useful for deriving the bound for $\easierdp{\relation}{k}{m}$ using the multiplicative
adversary method with these new parameters.

Before going ahead, we first sketch the overall ideas:
The input space is an $m$-fold tensor product of the inputs so it is natural for the adversary matrix $\Gamma'$ for $\easierdp{\relation}{k}{m}$ to be the $m$-fold tensor product of the one for $\relation$.
This is the same choice of $\Gamma'$ used in the previous strong direct product proofs of \cite{Spa08,DBLP:conf/coco/AmbainisMRR11}.
With this, in analogy with prior results for other progress bounds~\cite{Spa08,DBLP:conf/coco/AmbainisMRR11}, we will show that 
\begin{description}
    \item[(i)]
the per-query increase in the
progress measure for $\easierdp{\relation}{k}{m}$
has the same upper bound as the one for the per-query increase for the single copy $\relation$.
\end{description}

The eigenvalue upper bound $\lambda'$ we use will be of the form
$\lambda^\ell$ for $\ell$ that is some small constant fraction of
$k$.  
Because of the product structure of
$\Gamma'$, its eigenspaces are products of eigenspaces 
for individual copies of $\relation$ and hence for any output of $\easierdp{\relation}{k}{m}$, at most $\ell$ of
the $k$ claimed output values have projections on the input with squared norm larger than $\eta$, independently.   
We show that
\begin{description} 
\item[(ii)] for any output, squared norms of projections (maximum correctness probability of that output) are at most $\eta'$ 
 that is polynomial in $\eta^k$ on subspaces with eigenvalues at most $\lambda'$, even when accounting for all the possible ways that these $\ell$ coordinates intersect the $k$ coordinates output.   
\end{description} 
The final part of the argument shows how to argue
bounds based on the target
success probability.
\begin{description}
    \item[(iii)]  using the parameters $\Gamma', \lambda', \eta'$ and $\varepsilon>0$, there is some desired success probability $\kappa^k$ that is only exponentially small in $k$ such that the progress target derived from the adversary method is an 
    $\ell$-th power of the original progress target for $\relation$ and success probability $1-\varepsilon$; this yields
    a query lower bound that is $\ell=\Omega(k)$ times as large as the single copy case.
\end{description}
This last portion has some subtlety in the choice of the base of exponential decay as a function in $k$; it depends critically on how close the original 
$\lambda$ is to 1.

\medskip
We now proceed with the proof details:
We begin by instantiating the objects that yield the optimum value of $\madv^\dist_{\varepsilon,\lambda,\eta}(\relation)$.
That is, suppose that $\Gamma$ is a $\domainsize \times \domainsize$ positive-definite matrix with smallest eigenvalue $1$ such that $\Gamma \ket{\dist} = \ket{\dist}$, 
    $\lambda >1$ and $\eta > 0$ for $\eta\le \min\{1-\varepsilon,2^{-H_2(\delta)/(7/8-\delta)}\} $ satisfy $\norm{\Pi_{\relation^{-1}(y)} \Pi_{\Gamma,\leq \lambda}}^2 \leq \eta$ for all $y \in \outputspace$ and 
    \begin{equation*}
    \madv^\dist_{\varepsilon,\lambda,\eta}(\relation)=\log_c(1+(\lambda-1)(\sqrt{1-\varepsilon}-\sqrt{\eta})^2) 
    \end{equation*}
    for $c=\max_{i\in[n],p} \norm{\oracle^\dagger_{i,p} \Gamma^{1/2} \oracle_{i,p} \Gamma^{-1/2}}^2$.

    Define positive-definite matrix $\Gamma' = \Gamma^{\otimes m}$.
    We will show that we can choose $\lambda'>1$ and $\eta'>0$ and apply
    the properties of the multiplicative adversary method to yield our claimed bound.
    Observe that $\Gamma'$ has smallest eigenvalue 1 and
    \begin{equation*}
    \Gamma' \ket{\dist^m}=\Gamma^{\otimes m} \ket{\dist}^{\otimes m}=(\Gamma \ket{\dist})^{\otimes m}=\ket{\dist}^{\otimes m}=\ket{\dist^m}.
    \end{equation*}

    \paragraph{Part (i):}
    We next show that
    \begin{displaymath}
        \max_{i' \in [mn],p} \norm{\oracle^\dagger_{i',p} (\Gamma')^{1/2} \oracle_{i',p} (\Gamma')^{-1/2}}^2 = \max_{i\in[n],p} \norm{\oracle^\dagger_{i,p} \Gamma^{1/2} \oracle_{i,p} \Gamma^{-1/2}}^2
    \end{displaymath}
    as follows:
    Write $i' = (j-1)n + i$ for unique $i\in [n]$ and $j\in [m]$. 
    Since $\Gamma$ must be invertible and $\oracle_{i',p} = \identity^{\otimes {(j-1)}} \otimes \oracle_{i,p} \otimes \identity^{\otimes (m - j)}$ we have
    \begin{align*}
        \max_{i' \in [mn],p}& \norm{\oracle^\dagger_{i,p} (\Gamma')^{1/2} \oracle_{i,p} (\Gamma')^{-1/2}}^2\\
        &= \max_{i' \in [mn],p} \norm{(\identity^{\otimes {(j-1)}} \otimes \oracle^\dagger_{i,p} \otimes \identity^{\otimes (m - j)}) (\Gamma')^{1/2} (\identity^{\otimes {(j-1)}} \otimes \oracle_{i,p} \otimes \identity^{\otimes (m - j)}) (\Gamma')^{-1/2}}^2\\
        &=\max_{i' \in [mn],p} \norm{\identity^{\otimes {(j-1)}} \otimes (\oracle^\dagger_{i,p} \Gamma^{1/2} \oracle_{i,p} \Gamma^{-1/2}) \otimes \identity^{\otimes (m - j)}}^2\qquad \textrm{since $(\Gamma^{\otimes m})^{\pm 1/2} = (\Gamma^{\pm 1/2})^{\otimes m}$}\\
        &= \max_{i\in[n],p} \norm{\oracle^\dagger_{i,p} \Gamma^{1/2} \oracle_{i,p} \Gamma^{-1/2}}^2
    \end{align*}
    as desired.
  
    To prove our claimed bound, we will apply \cref{prop:mult-adv-query-bound} with some $\lambda'$ and $\eta'$ with a target $\varepsilon'=1-\kappa^k$.
    Since the base of the logarithm in the bound from  \cref{prop:mult-adv-query-bound} is still $c$, to do this it is necessary
    and sufficient to find some $\lambda'$ and $\eta'<\kappa^k$ such that for every output $y'$ of $\easierdp{\relation}{k}{m}$, $\norm{\Pi_{(\easierdp{\relation}{k}{m})^{-1}(y')} \Pi_{\Gamma',\leq \lambda'}}^2 \leq \eta'$  
    and for some fixed $\delta>0$,
    \begin{equation}
        (1+(\lambda'-1) (\sqrt{\kappa^k} - \sqrt{\eta'})^2))\ge (1+(\lambda-1) (\sqrt{1-\varepsilon} - \sqrt{\eta})^2)^{\delta k}
    \end{equation}
    for some $\kappa<1$ to be determined later.

\paragraph{Part (ii):}
    Define $\lambda'=\lambda^{\ell}$
    for some $\ell=\lceil \delta k\rceil$ for $\delta>0$.
    Let $V^0$ be the span of all eigenspaces of $\Gamma$ corresponding to
    eigenvalues $\le \lambda$ and $V^1$ be its orthogonal complement which, by definition, is the span of all
    eigenspaces of $\Gamma$ corresponding to eigenvalues $> \lambda$.
    We define the corresponding
    orthogonal projectors $\Pi_{V^0}= \Pi_{\Gamma, \leq \lambda}$
    and $\Pi_{V^1}= \Pi_{\Gamma, > \lambda}$ onto $V^0$ and $V^1$ respectively.
    By hypothesis, every $y\in \inputspace$ satisfies $\norm{\Pi_{\relation^{-1}(y)} \Pi_{\Gamma, \leq \lambda}}^2\le \eta$. Thus,
    \begin{equation}\label{eqn:selective-eta-def-single-copy}
        \norm{\Pi_{\relation^{-1}(y)} \Pi_{V^0}}^2 \leq \eta.
    \end{equation}

    By definition, every output $y'$ of $\easierdp{\relation}{k}{m}$ is
    of the form $y'=((i,y_i))_{i\in K_{y'}}$ for some $K_{y'}\subset [m]$ with $|K_{y'}|=k$. 
    This output 
    is correct for $x\in \inputspace^m$ iff $(x_i,y_i)\in \relation$ for all
    $i\in K_{y'}$.
    Therefore we can write $\Pi_{(\easierdp{\relation}{k}{m})^{-1}(y')}=\identity_{[m]\setminus K_{y'}}\otimes\bigotimes_{i\in K_{y'}} \Pi_{\relation^{-1}(y_i)}$.
    
    We will also need to use the product structure of $\Gamma'$ 
    to let us choose $\delta$ and bound $\eta'$ as a result.
    The eigenspace $V'$ of $\Gamma'$ with eigenvalue $\lambda_{V'}$ is the sum of all vector spaces $V_1\otimes \ldots \otimes V_m$ where $V_1, \ldots, V_m$ are eigenspaces of $\Gamma$ with eigenvalues $\lambda_{V_1}, \ldots, \lambda_{V_m}$ such that $\prod_i \lambda_{V_i} = \lambda_{V'}$.
    Since every $\lambda_{V_i}\ge 1$, if $\lambda_{V'}\le \lambda' = \lambda^\ell$
    then there must be fewer than $\ell$ values $i\in [m]$ where
    $\lambda_i>\lambda$ for any term in the sum.
    Therefore, in particular, for any set $K\subset [m]$ there must be strictly fewer than $\ell$ values $i\in K$ for any term in the sum such that 
    $V_i\subseteq V^1$.
    Therefore, for any set $K\subseteq [m]$,
    \begin{equation}\label{eqn:selective-bad-space-subset-of-sum-of-eigenspaces}
        V' \subseteq (\mathbb{C}^\inputspace)^{[m]\setminus K}\otimes\bigoplus_{\substack{u \in \{0,1\}^K\\ \abs{u} < \ell}} \bigotimes_{i \in K} V^{u_i}
    \end{equation}
    and hence
    \begin{equation}\label{eqn:selective-space-proj-ineq}
        \Pi_{\Gamma', \leq \lambda'} \preceq \identity_{[m]\setminus K}\otimes \sum_{\substack{u \in \{0,1\}^K \\ \abs{u} < \ell}} \bigotimes_{i \in K} \Pi_{V^{u_i}}.
    \end{equation}
    Choosing $K=K_{y'}$, we compute the value we can use to define $\eta'$:
\begin{align*}
    \norm{\Pi_{(\easierdp{\relation}{k}{m})^{-1}(y')} \Pi_{\Gamma',\le \lambda'}}^2 &= \norm{\big(\identity_{[m]\setminus K}\otimes\bigotimes_{i\in K} \Pi_{\relation^{-1}(y_i)}\big)\, \Pi_{\Gamma',\le \lambda'}}^2\\
    &\leq \norm{\big(\bigotimes_{i\in K} \Pi_{\relation^{-1}(y_i)}\big)\, \sum_{\substack{u \in \{0,1\}^K \\ \abs{u} < \ell}} \bigotimes_{i \in K} \Pi_{V^{u_i}}}^2 \qquad \textrm{by \eqref{eqn:selective-space-proj-ineq}}\\
    &=\norm{\sum_{\substack{u \in \{0,1\}^K \\ \abs{u} < \ell}} \bigotimes_{i \in K} \Pi_{\relation^{-1}(y'_i)} \Pi_{V^{u_i}}}^2\\
    &\leq \sum_{\substack{u \in \{0,1\}^K \\ \abs{u} < \ell}} \norm{\bigotimes_{i \in K} \Pi_{\relation^{-1}(y'_i)} \Pi_{V^{u_i}}}^2\qquad \textrm{by orthogonality}\\
    &= \sum_{\substack{u \in \{0,1\}^K \\ \abs{u} < \ell}} \prod_{i \in K}\norm{ \Pi_{\relation^{-1}(y'_i)} \Pi_{V^{u_i}}}^2\\
    &\leq \sum_{\substack{u \in \{0,1\}^K \\ \abs{u} < \ell}} \eta^{k-\ell + 1} \qquad \textrm{by \eqref{eqn:selective-eta-def-single-copy}}\\
    &< 2^{H_2((\ell-1)/k)k}\ \eta^{k-\ell+1}\qquad\text{by \cref{prop:binom-tail}}\\
     &\le (2^{H_2(\delta)}\eta^{1-\delta})^k\\
     &\le \eta^{k/8}\qquad\textrm{since $\eta\le 2^{-H_2(\delta)/(7/8-\delta)}$.}
\end{align*}

Defining $\eta'=\eta^{k/8}$ we have $\norm{\Pi_{(\easierdp{\relation}{k}{m})^{-1}(y')} \Pi_{\Gamma',\le \lambda'}\ket{\phi'}}^2\le \eta'$ as required.

\paragraph{Part (iii):} 

Set $\zeta=(\sqrt{1-\varepsilon} -\sqrt{\eta})^2$.
Since $\varepsilon\ge 1/2-\eta-\sqrt{2\eta}$,
$1-\varepsilon\le 1/2+\sqrt{2\eta}+\eta=(\sqrt{\eta}+1/\sqrt{2})^2$ and hence $\zeta\le 1/2$.
Now set
\begin{equation*}\alpha=\left(\frac{1+(\lambda-1)\zeta}{\lambda}\right)^\delta
=\left(1-\frac{(\lambda-1)(1-\zeta)}{\lambda}\right)^\delta\le e^{-\delta (1-1/\lambda)/2}<1
\end{equation*}
since $\zeta\le 1/2$.
We also have
$\eta\le 2^{-H_2(\delta)/(7/8-\delta)}<1$.
Define $\kappa(z)=(\alpha^{z/2}+\sqrt{\eta'})^{2/z}
=(\alpha^{z/2}+\eta^{z/16})^{2/z}$
for variable $z\ge 1$.
Observe that $\kappa(z)$ is strictly decreasing in $z$.
Further, since 
$$\alpha^{z/2}+\eta^{z/16}\le e^{-\delta (1-1/\lambda)z/4}+2^{-z H_2(\delta)/(16(7/8-\delta))}$$ approaches 0 as $z$ goes to
infinity, we can choose an integer $k_0=k_0(\lambda,\delta)$ such that $\kappa(z)<1$ for all $z\ge k_0$.
Set $\kappa=\kappa(k_0)$ which, therefore, is 
only a function of $\lambda$ and $\delta$.

Since $\kappa(z)$ is decreasing in $z$ we have
$\kappa=\kappa(k_0)\ge (\alpha^{k/2}+\eta^{k/16})^{2/k}$ for all $k\ge k_0$.
Therefore, for $k\ge k_0$ we have
\begin{align*}
1+(\lambda'-1) (\sqrt{\kappa^k} - \sqrt{\eta'})^2
&=1+(\lambda'-1) (\kappa^{k/2} - \eta^{k/16})^2\\
&\ge 1+(\lambda'-1)\alpha^{k}\\
&=1+(\lambda^{\delta k}-1) \left(\frac{1+(\lambda-1)\zeta}{\lambda}\right)^{\delta k}\\
&=1+(1-\lambda^{-\delta k})(1+(\lambda-1)\zeta)^{\delta k}\\
&\ge (1+(\lambda-1)\zeta)^{\delta k}.
\end{align*}
Therefore, by \cref{prop:mult-adv-query-bound},
\begin{align*}
Q_{1-\kappa^k}^{\dist^m}(\easierdp{\relation}{k}{m})&\ge
\log_c ((1+(\lambda'-1) (\sqrt{\kappa^k} - \sqrt{\eta'})^2)\\
&\ge \log_c (1+(\lambda-1)\zeta)^{\delta k}\\
&= \delta k\cdot \log_c (1+(\lambda-1)\zeta)\\
&= \delta k\cdot \log_c (1+(\lambda-1)(\sqrt{1-\varepsilon}-\sqrt{\eta})^2)\\
&= \delta k\cdot \madv^\dist_{\varepsilon,\lambda,\eta}(\relation)  
\end{align*}
as claimed.
\end{proof}
\subsection{A strong selective direct product theorem for quantum query complexity of all functions}\label{sec:strong-selective-prod-from-query-bound}

By applying
\cref{thm:selective-from-multiplicative-adversary}, we derive our generic quantum strong selective direct product theorem for function computation.

\begin{theorem}
    \label{thm:selective-allfuncs-precise}
Fix $\varepsilon = 1/10$.
There are constants $\alpha, \gamma>0$ and an integer $k_0\ge 1$  such that 
    for $f : \inputspace \to \range$ where $\inputspace \subseteq \domain^n$, and for all integers $k$ with $k_0\le k\le m$ we have  
      $Q_{1-(1-\varepsilon)^{\gamma k}}(\easierdp{f}{k}{m})\ge \alpha k\cdot Q_{\varepsilon}(f)$.
\end{theorem}

\begin{proof}
    Choose $\varepsilon=1/10$, $\lambda\approx 1.356$, and $\eta\approx 0.608$ as in the statement of
    \cref{thm:mult-optimal}, and let $C$ be the constant factor and $\dist$ be the distribution that results.
    We apply \cref{thm:selective-from-multiplicative-adversary}(iii) to yield an absolute constant $\kappa<1$ such that for $\delta=0.123$, $2^{-H_2(\delta)/(7/8-\delta)}\ge \eta$ and all $k\ge k_0(\delta)$ we have
    $$Q^{\dist^m}_{1-\kappa^{\delta k}}(\easierdp{f}{k}{m}) \geq \delta k \madv^\dist_{1/10, \lambda, \eta}(f) \geq \delta k Q_{1/10}(f)/C.$$
    We set $\gamma$ such that $\kappa^\delta=(1-\varepsilon)^\gamma$.
    Now $Q_\varepsilon(f)\ge Q^{\dist^*}_\varepsilon(f)$ for any distribution $\dist^*$ by
    Yao's Lemma and we can choose $\alpha=\delta/C$ to obtain the claim.
\end{proof}
\subsection{On $(1-\delta)$-selectivity of direct product problems}
\label{sec:threshold-selective}

Here we prove a relationship between the complexity of computing $(1-\delta)$-selective direct products, where up to a $\delta$ fraction of the $k$ answers
for $\easierdp{f}{k}{m}$ may be incorrect and selective direct products where
all such answers must be correct.
Essentially, this is a melding of a strong selective
direct product theorem and the idea behind threshold direct products.   

\begin{definition}
   Let $f$ be a relation and $1\le k\le m$.
    We say that an algorithm $A$ is \emph{$(1-\delta)$-correct at 
    computing $\easierdp{f}{k}{m}$ on input $x\in \domain^m$} iff  there is an 
    $\ell\le k$, such that $A$ produces a set of pairs
    $(i_1,y_1),\ldots, (i_\ell,y_\ell)$ and 
    $(x_{i_j},y_j)\in f$ for at least $(1-\delta)k$ choices of $j\in [\ell]$.
\end{definition}

\begin{theorem}
\label{thm:1-delta-selective}
Let $f$ be a relation.  Suppose that there is an algorithm that makes at most $T$ queries and for each $x\in \domain^m$ is $(1-\delta)$-correct at computing
$\easierdp{f}{k}{m}$ on $x$ with probability at least $q_k$.
Then there is an algorithm making at most $T$ queries that correctly
computes $\easierdp{f}{k'}{m}$ for $k'=\lceil (1-\delta)k\rceil$ on each input
$x\in \domain^m$ with
probability at least $q_k/\binom{k}{k'}\ge q_k\cdot 2^{-H_2(\delta)k}$ where
$H_2(\cdot)$ is the binary entropy function.
\end{theorem}

\begin{proof}
We can assume without loss of generality that $A$ produces exactly $k$ outputs
on every input by producing additional dummy answers that it has no reason to
believe are correct.   We then modify algorithm $A$ to yield an algorithm $A'$ as follows:  

On input $x\in \domain^m$, choose a subset $S\subseteq [k]$ of size $k'$ uniformly
at random.
Run algorithm $A$, keeping track of the answers
that it would have produced and denote them $(i_1,y_1),\ldots, (i_k,y_k)$ 
where $i_1<i_2<\cdots< i_k$.
Then only output the set of pairs $(i_j,y_j)$ such that $j\in S$.

Algorithm $A'$ makes exactly the same number of queries as $A$ does.  The set of pairs output by $A$ contains at least $k'$ values
$j\in [k]$ such that $f(x_{i_j})=y_j$ with probability at least $q_k$. 
Whenever this happens, at least
$1/\binom{k}{k'}$, all pairs $(i_j,y_j)$ for $j\in S$ will satisfy $(x_{i_j},y_j)\in f$.
Therefore $A'$ computes $\easierdp{f}{k'}{m}$ with probability at least
$q_k/\binom{k}{k'}$.  
Finally $\binom{k}{k'}=\binom{k}{k-k'}\le 2^{H_2((k-k')/k) k}\le 2^{H_2(\delta)k}$.
\end{proof}

This immediately yields the following corollary, which shows that any
strong selective direct product theorem implies an associated  strong $(1-\delta)$-selective direct product theorem for sufficiently small
constant $\delta>0$.

\begin{corollary}\label{cor:frac-selective-direct-products}
    Suppose that there is a constant $\gamma>0$ and an integer $T$ such that, for every integer
    $k\in [1,m]$,  any algorithm using $kT$ queries to compute 
    $\easierdp{f}{k}{m}$ is correct with probability at most $2^{-\gamma k}$ then, for
     $2\cdot H_2(\delta)/(1-2\delta)\le \gamma$,
    any algorithm  for $\easierdp{f}{k}{m}$ running
    in time at most $(1-\delta)kT$ is $(1-\delta)$-correct with 
    probability at most $2^{-\gamma k/2}$.
\end{corollary}

\begin{proof}
    Let $\delta>0$ satisfy $2\cdot H_2(\delta)/(1-2\delta)\le \gamma$.
    In particular, this implies that $\gamma (1/2-\delta)\le H_2(\delta)$ and
    hence $-\gamma/2 - H_2(\delta)\le -\gamma (1-\delta)$.
    
    Suppose that there is an
    algorithm $A$ for $\easierdp{f}{k}{m}$ using at most $T'\le (1-\delta)kT$ steps that is
    $(1-\delta)$-correct with probability  $>2^{-\gamma k/2}$.
    Applying \cref{thm:1-delta-selective} to the algorithm $A$
    yields an algorithm $A'$ for $\easierdp{f}{k'}{m}$ for
    $k'=(1-\delta)k$ that runs in time $T'$ that is correct
    with probability $> 2^{-\gamma k/2}/2^{H_2(\delta)k}$.
    By the relationship between $\delta$ and $\gamma$, this probability is $>2^{-\gamma (1-\delta)k}$.
    
    However, by the assumption, since $T'=\lfloor(1-\delta)kT\rfloor \le k'T$ algorithm $A'$ can have success probability at most $2^{-\gamma k'}\le 2^{-\gamma (1-\delta)k}$, which is a contradiction.
\end{proof}

The function $g(\delta)=2\cdot H_2(\delta)/(1-2\delta)\le 2/(1-2\delta)$ is continuous, increasing, concave for
$0<\delta<0.0730367...$ where it evaluates to $0.883343$,
and convex for $0.0730367...<\delta<1/2$.

\subsection{Quantum time-space tradeoff lower bound for matrix mapping}\label{sec:time-space-tradeoffs}

\begin{definition}
    \label{def:fn-matrix}
   We say that a matrix $F \in \bits^{n \times n}$ is a \emph{function matrix} iff every row has exactly one $1$.
   A function matrix therefore defines a function $f_F:[n]\rightarrow [n]$ given by
   $f_F(i)=\search_n(F_{i})$ where $F_{i}$ is the $i$-th row of $F$.
    Let $\calF_n\subseteq \bits^{n\times n}$ be the set of all $n\times n$ function matrices, and $F\sim \calF_n$ to denote sampling a uniformly randomly function matrix of dimension $n$. 
\end{definition}

\begin{definition}
    \label{def:F1AF2}
    For any fixed matrix $A \in \Z^{n\times n}$, we define the \emph{matrix mapping problem} $g_A: \calF_n^2 \to \Z^{n\times n}$ as $g_A(F,\tilde F):= F A \tilde F^T$ for two function matrices $F,\, \tilde F$.
\end{definition}

We observe that this definition implies that 
\begin{equation*}g_A(F,\tilde F)_{i,j}=(FA\tilde F^T)_{i,j}=A_{f_{F}(i),f_{\tilde F}(j)}=A_{\search_n(F_i),\search_n(\tilde F_j)}.
\end{equation*}

This matrix mapping problem is a generalization of a problem considered by Abrahamson~\cite{DBLP:journals/jcss/Abrahamson91} who showed that space-bounded classical algorithms require time that is $\Omega(n^3/S)$.

\begin{proposition}
For any $n \times n$ matrix $A$ and $n\times n$ function matrices $F$, $\tilde F$, there is a quantum algorithms with space
$S$ can compute $g_A(F,\tilde F)$ using $\tilde O(n^{2.5}/S)$ queries.
\end{proposition}

\begin{proof}
Using Grover's search, with space $S = O(k \log n)$ and $O(k \sqrt{n})$ queries we can find and store the location of $2k$ 1's of our $F$ and $\tilde F$ function matrices, and thus correctly output $k^2$ entries.
Since $g_A$ needs to output $n^2$ entries total, we can repeat this $n^2/k^2$ times and so obtain the claimed bound.
\end{proof}

We prove a time-space tradeoff lower bound showing that this algorithm for matrix mapping is essentially optimal.

\begin{restatable}{theorem}{FAF}
    \label{thm:F1AF2-cor}
    There exists $C>0$ such that for any $A\in \Z^{n\times n}$ with unique entries the following holds. 
    Any quantum query algorithm computing $g_A(F,\tilde F)$ with space $S < n/16 $ and success probability at least $2^{-S}$ requires $T \geq C n^{2.5}/S$ quantum queries.
\end{restatable}

We prove this using our strong selective direct product theorem for the $\search_n$ function applied to inputs with
exactly one 1.   
The main idea of the connection of matrix mapping to $\search_n$ is the following:

\begin{proposition}
\label{prop:matrix-to-search}
    Fix matrix $A\in \Z^{n\times n}$ with distinct entries and let  
   $F$ and $\tilde F$ be function matrices.  
   The values $(F A \tilde F^T)_{i_1, j_1},\ldots ,(F A \tilde F^T)_{i_k, j_k}$ for distinct $(i_\ell,j_\ell)$ pairs determine the values 
   \begin{equation*}\{\search_n(F_{i_1}),\ldots,\search_n(F_{i_k})\}\cup\{\search_n(\tilde F_{j_1}),\ldots,\search_n(\tilde F_{j_k})\}.
   \end{equation*}
   In particular, they determine the outputs of at least $2\sqrt{k}$ independent $\search_n$ problems on inputs with exactly one 1.
\end{proposition}

\begin{proof}
    Since the matrix $A$ has distinct entries, each output value $(F A \tilde F^T)_{i,j}=A_{i',j'}$ for
    a unique pair $i'$ and $j'$.
    By our observation above, we must have $i'=\search_n(F_i)$ and $j'=\search_n(\tilde F_j)$.
    Let $s=|\{i_1,\ldots,i_k\}|$ and $t=|\{j_1,\ldots, j_k\}|$.
    Since each row of $F, \tilde F$ corresponds to an independent search problem, these are the outputs of a total of $s+t$ independent $\search_n$ problems.
    Since $st\ge k$, we must have $s+t\ge 2\sqrt{k}$ by the arithmetic-geometry mean inequality.
\end{proof}

We first indicate how the ordinary direct product theorem for $\search_n$ would yield a time space tradeoff lower bound for output-oblivious algorithms.  This is the analogue of the method of~\cite{KSdW07} using the Borodin-Cook framework discussed in \cref{sec:prelims}.   

\begin{proof}[Sketch of proof for output-oblivious special case]
There are $n^2$ outputs of $g_A$ to be produced using $T$ queries and space $S$.
We break up the $T$ queries into consecutive segments of on the order of $S\sqrt{n}$ queries so there are on the order of
$T/(S\sqrt{n})$ segments.   
Since the quantum algorithm is output-oblivious there is some consecutive segment and
fixed sequence $(i_1,j_1),\ldots,(i_k,j_k)$ of output indices of $g_A$ produced during this segment, where
$k$ is at least on the order of 
\begin{equation}
\label{eq:mat-output-oblivious}
n^2 /[T/(S\sqrt{n})]=n^{2.5} S/T. 
\end{equation}
These outputs must be correct with at least the correctness probability of the algorithm as a whole.
By \cref{prop:matrix-to-search}, these $k$ outputs correspond to the answers to $2\sqrt{k}$ independent $\search_n$ problems.
This segment begins with some state given by $S$ qubits so, by \cref{prop:quant-union}, if we replaced this state by the maximally mixed state the segment will succeed in producing the outputs of these $\ell=2\sqrt{k}$ independent $\search_n$ problems with probability at least
$2^{-S}$ times its actual success probability on the segment with the correct initial state.
By the strong direct product property for $\search_n$, which says that quantum algorithms cannot solve $\ell$ independent search problems with probability better than $c^{-\ell}$ in time at most some constant times $\ell\sqrt{n}$, unless $\ell$ is at most some constant times $S$ (and hence $k$ is at most some constant times
$S^2$) we get a contradiction.
Using
\eqref{eq:mat-output-oblivious}, we get that $S^2$ must be at least on the order of $n^{2.5} S/T$ and 
hence $T$ is $\Omega(n^{2.5}/S)$.
\end{proof}

This proof fails for general quantum algorithms because the choice of output indices produced during a
segment may depend on the value of the input and hence there is no sense in which a consecutive segment completely solves any of the search problems given by the rows of the function matrices.  Even the number of values of output coordinates
produced in a given segment could be very different for different inputs.

The following lemma encapsulates our use of~\cref{thm:search-sqdpt}, our strong selective direct product theorem for $\search_n$.

\begin{lemma}
    \label{lem:FAF-lemma}
    There exists $C' >0$ such that the following holds. Let $\alpha \leq 0.02$, $F,\tilde F \sim \calF_n$ be uniformly random function matrices, $A \in \Z^{n \times n}$ have distinct entries, and $C'\leq k\leq n^2$. 
    Then the probability that any quantum algorithm making at most $\alpha  \sqrt{kn}$ queries to $F, \tilde F$ produces $k$ correct output coordinates of $g_A(F,\tilde F)$ is at most $2^{- \sqrt{k}/4}$.
\end{lemma}

\begin{proof}
    By \cref{prop:matrix-to-search}, on any input pair $F,\tilde F$, implies that producing those $k$ output coordinates of $g_A(F,\tilde F)$ on an input, involves
    producing at least $2\sqrt{k}$ answers on that input.
    Since the pair of input matrices $F,\tilde F$ together have $2n$ independent $\search_n$ problems between them, this means solving the selective direct product problem
    $\easierdp{\search_n}{2\sqrt{k}}{2n}$ on those inputs.

    Choosing $\beta = 1/4$ satisfies the assumptions needed to apply \cref{thm:search-sqdpt}. 
    Then we have that if our algorithm makes fewer than $\frac{2\beta}{24} \sqrt{k}\sqrt{n} \leq 0.02 \sqrt{kn}$ queries, we successfully output $k$ entries of $f_A$ with probability at most $9 \cdot 2^{-\sqrt{k}}\leq 2^{-\sqrt{k}/4}$ for sufficiently large $k \geq C'$.
\end{proof}

We are now prove the full version of \cref{thm:F1AF2-cor}, keeping more detailed track of parameters.

\begin{proof}[Proof of \cref{thm:F1AF2-cor}]
    Suppose we have a quantum circuit $\calC$ using $T$ queries and space $S < n/16$ that computes $f_A$ with probability greater than $2^{-S}$.
    We know that $S \geq 2\log n$, so we may assume $T \leq C n^{2.5}/(2\log n)$, since otherwise we are done.
    In particular, note that we have $2^{2S} \geq n^4 \geq T$.

    Let $\alpha = 0.02$.
    We must have $T \geq \alpha \cdot 16 S \sqrt{n}$, since otherwise by \cref{lem:FAF-lemma} (and our assumption that $S < n/16$) we have probability at most $2^{-S}$ of producing $16^2 S^2 < n^2$ outputs.

    We split the circuit $\calC$ into segments with at most $16\alpha S\sqrt{n}$ queries, so that there are $\lceil T/(16\alpha S \sqrt n)\rceil $ segments. 
    For the purposes of analysis, we instead analyze the behavior of the segment when the work space is $\maxmix$ on $S$ qubits.
    By \cref{cor:quant-union-oracle-reg}, this increases the probability of any output in this segment by at most $2^S$, and also returns the input register to the initial uniform superposition state, allowing us to apply \cref{lem:FAF-lemma}.
    
    By \cref{lem:FAF-lemma} combined with \cref{prop:quant-union}, this subcircuit correctly computes at least $(16S)^2$ entries of $f_A(F_1,F_2)$ with probability at most $2^S \cdot 2^{-16S/4} \leq 2^{-3S}\leq 2^{-S}/T$. 

    Taking a union bound over subcircuits, this tells us that the probability of any of them producing $(16S)^2$ outputs is at most $2^{-S}$.
    Since $f_A$ has $n^2$ outputs, this tells us that
    $$\lceil T/(16\alpha S \sqrt n)\rceil (16S)^2\geq n^2. $$
    Since $T\geq \alpha 16 S\sqrt n$, we have that
    $$2\cdot 16 \cdot T  \geq \alpha n^{2.5}/S,$$
    or in other words
    $$T \geq \frac{\alpha}{32} n^{2.5}/S,$$
    as desired. 
\end{proof}

Using the general transformation of~\cite{DBLP:journals/toct/BeameK25} we obtain
the following:

\begin{corollary}
     For any fixed matrix $A \in \Z^{n\times n}$ with distinct entries, any quantum query algorithm computing the matrix mapping function $g_A$ requires cumulative memory $\Omega(n^{2.5})$.
\end{corollary}

\section{Strong list-decoding direct product theorems for quantum query\\ complexity}
\label{sec:list-dec}

\emph{List-decoding direct products} are a further generalization of selective direct products where an algorithm succeeds by outputting any list of $\abs{\range}^{m-k}$ potential outputs for all $m$ problems that contains the true output \cite{DBLP:conf/focs/BenDavidB25}.

\begin{definition}\label{def:list-decoding-dp}
    Given a relation $\relation \subseteq \inputspace \times \range$,
    we use $\easierlistdp{\relation}{k}{m}$ to denote the list decoding direct product relation on $\inputspace^m \times \binom{\range^m}{\abs{\range}^{m-k}}$ such that
    \begin{math}
        ((x_1, \ldots, x_m), \listofout) \in \easierlistdp{\relation}{k}{m}
    \end{math}
    iff there exists an $(y_1,\ldots,y_m) \in \listofout$ such that $(x_i,y_i) \in \relation$ for all $i\in [m]$.
\end{definition}

A selective direct product output that gives correct values $\tau_K$ for some set $K$ of indices of size $k$ can be seen as a special case of a list-decoding direct product where the output list $\listofout$ contains all $\abs{\range}^{m-k}$ outputs consistent with
$\tau_K$ on $K$.

Very recently Ben-David and Blais proved a classical strong list-decoding direct product theorem for classical randomized query complexity of Boolean functions:

\begin{theorem}\cite{DBLP:conf/focs/BenDavidB25}
 There is a constant $c\in (1/2,1)$ such that for any $f:\bits^n\rightarrow \bits$, any positive integers $k\le m$ and every $\gamma\in [c,1)$,
 any classical randomized query algorithm computing
 $\easierlistdp{f}{k}{m}$ with success probability
at least $\gamma^k$ requires $\Omega(k)$ times the number of queries for a classical randomized query algorithm to compute $f$ with success probability $2/3$.
\end{theorem}

We are able to prove such theorems for every Boolean output function.
We first observe that there are some fundamental limitations to obtaining fully general strong list decoding direct product theorems, namely that
there is no such theorem for total relations that are not functions, and the quantum multiplicative adversary method cannot be used to obtain the quantum analogue of such a theorem except under strict circumstances.

Suppose that there is some $y^*\in \range$ such that
$\Pr_\dist[x\in \relation^{-1}(y^*)]=\eta>1/\abs{\range}$.
Define the list $\listofout^m_{y^*}$ to be the set of
all vectors in $Y^m$ that have at least
a $\eta*=(1/\abs{Y}+\eta)/2$ coordinates and define
$\delta>0$ such that $\eta^*=(1-\delta)\eta$.
Then we also have $\eta^*>(1+\delta)/\abs{Y}$.

Observe that $|\listofout^m_{y^*}|$ is equal to $\abs{\range}^m$ times the probability that the binomial distribution
$Bin(m,\frac{1}{\abs{\range}})$ has a value at least $\eta^* m\ge (1+\delta)m/\abs{\range}$.
By a Chernoff bound this probability is at most
$e^{-2\delta^2 m}$.
For $m$ sufficiently in terms of $\abs{\range}$ and $k$, this is smaller than
$\abs{\range}^{-k}$, so $\listofout^m_{y^*}$ has size at 
most $\abs{\range}^{m-k}$ and so is an allowable
output for the list-decoding direct product problem.

On the other hand, since for each $i$ $\Pr_{\dist}[x_i\in \relation^{-1}(y^*)]\ge \eta$,
$\Pr_{\dist^m}[x\notin \listofout^m_{y^*}]$ is at most
the probability that the binomial distribution 
$Bin(m,\eta)$ has a value less than $\eta^* m\le (1-\delta)\eta m$.
By a Chernoff lower bound, this is at most 
$e^{-\delta^2 m/2}$ which is always at most a small positive constant and will actually be near 0.
Therefore the answer $\listofout^m_{y^*}$ will almost
surely contain the correct output for $\relation^{\otimes m}$ at a computational cost of 0.

Therefore, we have the following:
\begin{description}
    \item{(1)} Any strong list-decoding direct product
theorem can only apply to functions, rather than arbitrary total relations, so in the following we assume that we have a function $f$ rather than a relation $\relation$.
     \item{(2)} The multiplicative adversary method with parameter $\eta$ assigns a cost of 0 to algorithms with success probability up to
     $\eta$, since this is the assumed upper bound on the success probability for algorithm states in the space spanned by eigenspaces of the adversary matrix $\Gamma$ with eigenvalue at most $\lambda$. 
     Therefore the method cannot admit a strong
     list-decoding direct product theorem if
     $\eta>1/\abs{\range}$.
     \item{(3)} Strong list-decoding direct product theorems for any non-Boolean output function cannot follow from our argument that converts the negative-weights additive adversary
     lower bounds to multiplicative ones, since our non-Boolean output argument necessarily yields values of $\eta>1/\abs{\range}$.
     \item{(4)} Strong list-decoding direct product theorems do not hold for some functions with more than $2$ outputs.
     For example, consider a 3 output version of $f_{\Shaltiel}$ where $f_{\Shaltiel}(x_1 = 0,x_2, \ldots, x_n) = -1$, and $f_{\Shaltiel}(x_1 = 1, x_2, \ldots, x_n) = x_2 \oplus \ldots \oplus x_n$.
     Any distribution over inputs where the probability that $f_{\Shaltiel}(x) = -1$ is $1/3$, as is required for $\eta = 1/ \abs{\range}$, would result in too much progress by querying the first input.
\end{description}

Nonetheless, we are able to prove
strong list-decoding direct product theorems that generalize our strong selective direct product theorems for every function with an asymptotically optimal multiplicative adversary argument with parameter $\eta = 1/\abs{\range}$.
For Boolean-valued functions, strengthened universality results for multiplicative adversaries in
 \cref{thm:bool-mult-bound-is-tight} means that we have multiplicative adversaries with this precise property and hence
 full strong list-decoding theorems for all such
 functions.

\begin{sloppypar}
In all of our strong selective direct product theorem
arguments, there is a subspace $V^0$ of $\mathbb{C}^\inputspace$ and its orthogonal complement
$V^1$ in $\mathbb{C}^\inputspace$ such that for any
$\phi\in V^0$,
$\norm{\Pi_{f^{-1}(y)}\phi}^2=\norm{\phi}^2/\abs{\range}$.
The notation for these spaces is explicit in the proof of \cref{thm:selective-from-multiplicative-adversary} with $\eta=1/\abs{\range}$,
but we also have them for the concrete problems of $\search_n$ and $\OR_n$, where
$V^0=\spn(\ket{\chi_0})$, and
$\Parity_n$ where $V^0=\spn(\
\set{\ket{-}_S\ket{+}_{\overline S}\ :\ \abs{S}<n/2})$.
\end{sloppypar}

The lower bounds then follow by analyzing a subspace $V'$
in $(\mathbb{C}^\inputspace)^m$
that is contained in the span of direct products of $V^0$ and
$V^1$ that contain at most $\ell$ copies of $V^1$ in some subset
$K$ of $[m]$ of size $k$ that corresponds to the set of output values
that must be correct in order for the selective direct product output to be correct.
(This is explicitly described in \eqref{eqn:selective-bad-space-subset-of-sum-of-eigenspaces} and implicit in the other proofs.) 
The lower bounds follow for $\ell$ a small fraction of $k$ by taking a union bound
over all $\binom{k}{\le\ell}$ choices of where the
$\le \ell$ copies of $V^1$ can appear.

For list-decoding direct product problems, all outputs are relevant, so the subspace $V'$ to be analyzed
satisfies
    \begin{equation}
    \label{eq:V'}
        V' \subseteq \bigoplus_{\substack{u \in \{0,1\}^m\\ \abs{u} < \ell}} V^u=\bigoplus_{\substack{u \in \{0,1\}^m\\ \abs{u} < \ell}} \bigotimes_{i \in [m]} V^{u_i}
    \end{equation}
Using the same ideas would now involve a union bound
over all $\binom{m}{\le \ell}$ choices of where
the $\ell$ copies of $V^1$ could appear.
This can be done, but would yield a degradation of both
the lower bound on the number of queries and the upper bound on the
success probability of roughly $\log_{\abs{\range}} (m/k)$ and therefore would not yield a full strong list-decoding direct product theorem.

To improve this, we need to augment the first moment calculation of probabilities inherent in the union bound with a second moment 
calculation.
Since the probabilities are themselves related to squared
norms of vectors, this probability second moment
calculation necessarily focuses on the fourth powers of norms of vectors.
This second moment bound is encapsulated in the following
lemma whose proof is given in \cref{sec:list-dec-lemma}.
Its proof is closely related to the inductive proof of Bonami's 4-norm/2-norm (hypercontractive) inequality for low-degree functions~\cite{Bonami1970} given in~\cite{moo:journal,DBLP:books/daglib/0033652}.

\begin{restatable}{lemma}{listdecnolosskeylemma}\label{cor:list-dec-no-loss-key-lemma}
    For any $\phi \in \bigoplus_{u \in \{0,1\}^{m}, \abs{u}\le \ell} V^u$,
    \begin{displaymath}
        \sum_{y' \in \range^{m}} \norm{\Pi_{(f^{\otimes m})^{-1}(y')}\phi}^4 \leq \frac{\max(16,3 \abs{\range})^{\ell}}{\abs{\range}^{m}} \norm{\phi}^4.
    \end{displaymath}
\end{restatable}

We first show how these bounds follow for simple
functions (with bounds only slightly weaker than those for the corresponding strong selective direct product theorems) and then apply similar ideas to the multiplicative adversary method for functions
using parameter $\eta=1/\abs{\range}$.

\subsection{Strong list-decoding direct product theorems for specific functions}\label{sec:list-specific}

We begin with $\search_n$.   We recall the following from
\cref{sec:search}:
\begin{itemize}
   \item Oracle queries for $\search_n$ are made to individual bits of the input and $\ket{j}$ is shorthand for the input $\ket{e_j}$ which has exactly one 1 in position $j$.
   For this space of inputs $\inputspace$, $\search_n$ is a function.
    \item 
$\ket{\chi_0} = \frac{1}{\sqrt{n}}\sum_{j \in [n]} \ket j$ is a uniform superposition over the inputs to a single copy of $\search_n$.
\item $\Pi_{V^0}$ projects the input for a single copy of $\search_n$ onto $V^0=\spn\{\ket{\chi_0}\}$.   
The orthogonal complement of $V^0$ is denoted by $V^1$.
\item For any quantum state $\ket{\phi}$ for a single coordinate, $\norm{\pisucc^{\search_n}\Pi_{V^0}\ket{\phi}}^2=1/n$ where $\pisucc^{\search_n}$ is the projector onto states where $\search_n$ has been correctly computed. 
This is equivalent to the statement that
$\norm{\pisucc^{\search_n}\ket{\phi}}^2=1/n$ for $\ket{\phi}\in V^0$.
\item The state is initialized as $\ket{(0,0), 0, 0}\ket{\chi_0}^{\otimes m}$.  Note that
$\ket{\chi_0}^{\otimes m}$ is in $V^u$ for $u=(0,\ldots,0)\in \set{0,1}^m$.
\item
For any $Q\subseteq [m]$,  $\Pi_Q = \bigotimes_{i \in Q}(\identity - \Pi_{V^0})_i \bigotimes_{i \not \in Q} (\Pi_{V^0})_i$.
Observe that $\Pi_Q$ is the projection that leaves the work space alone and projects the input registers onto the space
$V^{u}$ where $u\in \set{0,1}^n$ is the characteristic
vector of the set $Q$.
\item $\Pi_\ell := \sum_{Q \subseteq [m], |Q| = \ell} \Pi_Q$ and $\Pi_{\le \ell}:= \sum_{j\le \ell} \Pi_j=\sum_{Q \subseteq [m], |Q| \le \ell} \Pi_Q$.
Observe that $\Pi_{\le \ell}$ is the projection that acts only on the input registers and projects them
onto the space $\bigoplus_{u \in \{0,1\}^{m}, \abs{u} \le  \ell} V^u$.
\item
The $\Pi_Q$ are mutually orthogonal, as are the $\Pi_\ell$ and 
 $\sum_{\ell = 0}^m \Pi_\ell = \identity$.
 \end{itemize}

 Define $\pisucc^{\search_n^{L_{k,m}}}$ to be analogue of
 $\pisucc^{\search_n}$ for the list-decoding direct product problem for
 $\search_n$.
 Let $\ket{\psi_T}$ be the state of any algorithm for this problem
 after $T$ steps.

 \begin{lemma}
 \label{lem:list-search-small}
    \begin{math}
        \norm{\pisucc^{\search_n^{L_{k,m}}} \Pi_{\le \ell} \ket{\psi_T}} \leq \big(\frac{\max(16,3n)^{\ell}}{n^k}\big)^{1/4}.
    \end{math}
\end{lemma}

\begin{proof}
    For any $w$ we write $\listofout(w)$ for the list of $n^{m-k}$ possible outputs in $[n]^m$ produced by the list-decoding algorithm
    for $\search_n$. 
    For $y\in [n]^m$, we write $\Pi_y$ as a shorthand for $\Pi_{(\search^{\otimes m}_n)^{-1}(y)}$, the projection
    onto those inputs for which $y$ is the correct output of
    $\search^{\otimes m}_n$.  (Given our compact representation of the space of inputs we consider, this, paradoxically, is just the input state $\ket{y}$.) 
    For a list $\listofout$, we write $\Pi_{\listofout}$
    for the projection onto those inputs for which some $y\in \listofout$
    is a correct output of $\search^{\otimes m}_n$.
    In particular, the $\Pi_y$ are orthogonal and $\Pi_\listofout=\sum_{y\in \listofout} \Pi_y$.
    Then $\pisucc^{\search_n^{L_{k,m}}}$ projects the states with
    working registers $w$ using $\Pi_{\listofout(w)}$ and
    we have 
    \begin{align*}
        \norm{\pisucc^{\search_n^{L_{k,m}}} \Pi_{\le \ell} \ket \psi_T}^2
        &= \bignorm{\sum_{i,j,p,w} \alpha_{i,j,p,w}  \pisucc^{\search_n^{L_{k,m}}}  \ket{i,j,p,w}\ket{\phi_{i,j,p,w}}}^2\\
        \noalign{\text{where $|\alpha_{i,j,p,w}|^2=1$ and $\phi_{i,j,p,w}\in \bigoplus_{u \in \{0,1\}^{m}, \abs{u} \le  \ell} V^u$}}
        &\le \sum_{i,j,p,w} \bignorm{\alpha_{i,j,p,w} \ket{i,j,p,w} \Pi_{\listofout(w)}  \ket{\phi_{i,j,p,w}}}^2\\
        \noalign{\textrm{by definition of 
        $\pisucc^{\search_n^{L_{k,m}}}$ and the fact that it preserves orthogonality for different $i,j,p,w$}}
        &\le \max_{i,j,p,w} \bignorm{\Pi_{\listofout(w)}  \ket{\phi_{i,j,p,w}}}^2\\
        &= \max_{i,j,p,w} \bignorm{\sum_{y'\in \listofout(w)} \Pi_{y'}  \ket{\phi_{i,j,p,w}}}^2\\
        &\le \max_{i,j,p,w}\sum_{y'\in \listofout(w)}  \norm{\Pi_{y'}  \ket{\phi_{i,j,p,w}}}^2\\
        \noalign{\textrm{because the $\Pi_{y'}$ are orthogonal since $\search_n$ is a function on $\inputspace$}}
        &\le \max_{i,j,p,w} \abs{\listofout(w)}^{1/2} \big(\sum_{y'\in \range^m}\norm{\Pi_{y'} \ket{\phi_{i,j,p,w}}}^4\big)^{1/2}\\
        \noalign{\textrm{by Cauchy-Schwarz}}
        &\le n^{(m-k)/2}\cdot \bigg(\frac{\max(16,3n)^{\ell}}{n^{m}}\bigg)^{1/2}\\
        \noalign{\textrm{by \cref{cor:list-dec-no-loss-key-lemma} and our bound on $\abs{\listofout(w)}$}}
        &= \bigg(\frac{\max(16,3n)^{\ell}}{n^k}\bigg)^{1/2}.\qedhere
    \end{align*}  
\end{proof}

This allows us to obtain the following quantum strong list-decoding direct
product theorem for $\search_n$.

\begin{theorem}
\label{thm:list-search}
    Any quantum algorithm for $\search_n^{L_{k,m}}$ with
    at most $k\sqrt{n}/192$ queries has success
    probability at most $9\cdot 2^{-k/4}$.
\end{theorem}

\begin{proof}
    Let $\ell_0= \lceil k/8\rceil$ and suppose that
    the quantum algorithm for $\search_n^{L_{k,m}}$ uses at most $T = \sqrt{n}\ell_0/24$ queries.
    Recall that $\Delta_{t, \ell} := \norm{\Pi_\ell \ket{\psi_t}}$.
 Then
    \begin{align*}
        \norm{\pisucc^{\search_n^{L_{k,m}}} \ket{\psi_T}} &\leq 
        \bignorm{\pisucc^{\search_n^{L_{k,m}}}\Pi_{\le \ell_0-1} \ket{\psi_T}} +\sum_{\ell \ge \ell_0} \bignorm{\Pi_\ell \ket{\psi_T}} \\
        &\leq \bigg(\frac{\max(16,3n)^{\ell_0-1}}{n^k}\bigg)^{1/4} +
        \sum_{\ell \ge \ell_0} \Delta_{T, \ell}  \qquad\textrm{by \cref{lem:list-search-small}}\\
        &\leq \bigg(\frac{\max(16,3n)^{\ell_0-1}}{n^k}\bigg)^{1/4} +\sum_{\ell \ge \ell_0} \binom{T}{\ell} \bigg(\frac{4}{\sqrt{n}}\bigg)^\ell  \qquad\textrm{by \cref{lem:delta-bound}}\\
        &\leq \bigg(\frac{\max(16,3n)^{\ell_0-1}}{n^k}\bigg)^{1/4}+\sum_{\ell \ge \ell_0} 
        \bigg(\frac{4e T}{\ell\sqrt{n}}\bigg)^\ell \\
        &\leq n^{\ell_0-1-k/4}+\sum_{\ell \ge \ell_0} 
        \bigg(\frac{4e T}{\ell\sqrt{n}}\bigg)^\ell 
        \qquad\textrm{since $n\ge 2$}\\
        &\leq 2^{-k/8}+ 2 \cdot 2^{-\ell_0}  \le 3 \cdot 2^{-k/8}\qquad \text{by our assumptions on $\ell_0$ and $T$}.
    \end{align*}
    Therefore the probability of the algorithm's success $ \norm{\pisucc^{\search_n^{L_{k,m}}} \ket{\psi_T}}^2\le
    9\cdot 2^{-k/4}$.
\end{proof}

When we apply the same ideas as above to $\OR_n$ we
obtain the following analogue of \cref{lem:list-search-small} which replaces $n$ by $2$:

\begin{lemma}
    \label{lemma:list-or-small}
    \begin{math}
        \norm{\pisucc^{\OR_n^{L_{k,m}}} \Pi_{\le\ell}\ket{\psi_T}} \leq 2^{\ell-k/4}.
    \end{math}
\end{lemma}

This weaker bound is essentially the only version of the bound
in \cref{lem:list-search-small} used in the proof of \cref{thm:list-search}.
The only other difference for $\OR_n$ is that in the application of  
\cref{lem:delta-bound} to $\OR_n$ there $2n$ coordinates rather than $n$ coordinates, so we can
set $T=\ell_0 \sqrt{n}/16$ and obtain that
$4eT/(\ell_0\sqrt{2n})\le 1/2$.
This yields the following strong list-decoding direct product theorem:

\begin{theorem}
   Any quantum algorithm for $\OR_n^{L_{k,m}}$ with
    at most $k\sqrt{n}/144$ queries has success
    probability at most $9\cdot 2^{-k/4}$.
\end{theorem}

By similar transformations of our arguments for the selective direct product problem, the analogous strong list-decoding direct product theorem holds for $\Parity_n$, with the only difference being the running time bound.

\begin{theorem}
   There are constants $c', c''>0$ such that any quantum algorithm for $\Parity_n^{L_{k,m}}$ with
    at most $c'kn$ queries has success
    probability at most $2^{-c'' k}$.
\end{theorem}

\subsection{List-decoding and the multiplicative adversary method}\label{sec:list-dec-mult}

We prove a strong list-decoding direct product theorem for the only cases of the multiplicative adversary lower for which one could expect to prove such a result, namely for function computation where there is an asymptotically tight multiplicative adversary lower bound with $\eta = 1/\abs{\range}$.
As we showed in \cref{sec:neg-to-mult-bool}, this includes all Boolean-valued functions.

\begin{theorem}\label{thm:generic-list-decoding}
    Let $f: \inputspace \to \range$ be a function satisfying \cref{prop:mult-adv} for distribution $\dist$ over $\inputspace \subseteq \domain^n$ with adversary matrix $\multAdvMat$, $\lambda \in (1, \norm{\multAdvMat}]$, $\eta = 1/\abs{\range}$, and with error probability
    $ 1/2- 1/\abs{\range} -\sqrt{2/\abs{\range}}\le \varepsilon<1-1/\abs{\range}$.
    There is an $k_0=k_0(\lambda) \geq 1$ and a $\kappa< 1$ such that for any integers $k \geq k_0$ and $m$:
    \begin{displaymath}
        Q_{1-\kappa^{k}}^{\dist^{\otimes m}}(\easierlistdp{f}{k}{m}) \geq (k/8) \,\madv^\dist_{\varepsilon,\lambda,1/\abs{\range}}(f) .
    \end{displaymath}
\end{theorem}

\begin{proof}[Proof of \cref{thm:generic-list-decoding}]
The basic structure of this proof follows the proof structure for \cref{thm:selective-from-multiplicative-adversary} in that we derive the lower bound based on the
adversary matrix $\Gamma$, eigenvalue bound $\lambda$ and
the probability threshold $\eta$ for the function $f$, which in this case must satisfy $\eta=1/|\range|$.

We set  $\Gamma' = \Gamma^{\otimes m}$ as in the proof
of \cref{thm:selective-from-multiplicative-adversary}.
Part (i) of that proof immediately implies that the 
per-query progress bound $c$ for $\easierlistdp{f}{k}{m}$ based on $\Gamma'$ is exactly
the same as the per-query progress bound for the multiplicative adversary bound for $f$.

Part (ii) of our bound will follow the same general argument as that in the proof of \cref{thm:selective-from-multiplicative-adversary} setting $\lambda'=\lambda^\ell$ for $\ell = \ceil{k/8}$.

Using the same definitions of $V^0$, $V^1$, $\Pi_{V^0}$, $\Pi_{V^1}$, and $V'$ as in the proof of \cref{thm:selective-from-multiplicative-adversary}, by \eqref{eqn:selective-eta-def-single-copy} we have
 \begin{equation}\label{eqn:list-eta-single-copy}
        \norm{\Pi_{f^{-1}(y)} \Pi_{V^0}}^2 \leq \eta=1/|\range|
    \end{equation}
and by \eqref{eqn:selective-bad-space-subset-of-sum-of-eigenspaces} with $K=[m]$, we have
   \begin{equation}\label{eqn:list-space-proj-ineq}
        V' \subseteq \bigoplus_{\substack{u \in \{0,1\}^m\\ \abs{u} \le \ell-1}}\, \bigotimes_{i \in [m]} V^{u_i} := \listbadspace
    \end{equation}
By definition, every output of $\easierlistdp{f}{k}{m}$ is of the form $\listofout\subset \range^m$ with
$|\listofout|\le |\range|^{m-k}$.
This output is correct for
$x\in \inputspace^m$ iff there is some $y'\in \listofout$
such that $y'\in \listofout$ such that $y'\in (f^{\otimes m})^{-1}(x)$,
which is equivalent to
$y'_i\in f^{-1}(x_i)$ for all $i\in [m]$.
Since $f$ is a function, so is $f^{\otimes m}$ and hence the
$\Pi_{(f^{\otimes m})^{-1})(y')}$ for different $y'$ are
orthogonal.
Therefore 
\begin{equation}
\label{eqn:listdecoding-decomp}
  \Pi_{(\easierlistdp{f}{k}{m})^{-1}(\listofout)}= \sum_{y'\in \listofout}\Pi_{(f^{\otimes m})^{-1}(y')}
\end{equation}
By \cref{cor:list-dec-no-loss-key-lemma}, we have that for any $\ket{\phi} \in \listbadspace$:
\begin{equation}\label{eqn:list-dec-no-loss-key-lemma-statement}
    \sum_{y' \in \range^{m}} \norm{\Pi_{(f^{\otimes m})^{-1}(y')} \ket{\phi}}^4 \leq \frac{\max(16,3 \abs{\range})^{\ell-1}}{\abs{\range}^m}.
\end{equation}
It follows that for $\listofout\subset \range^m$ with $|\listofout|\le |\range|^{m-k}$, 
\begin{align*}
    \norm{&\Pi_{(\easierlistdp{f}{k}{m})^{-1}(\listofout)} \Pi_{\Gamma',\le \lambda'}}^2\\
    &\leq \max_{\ket{\phi} \in \listbadspace} \norm{\Pi_{(\easierlistdp{f}{k}{m})^{-1}(\listofout)} \ket{\phi}}^2 \qquad \textrm{by \eqref{eqn:list-space-proj-ineq}}\\
    &\leq \max_{\ket{\phi} \in \listbadspace} \sum_{y' \in \listofout} \norm{ \Pi_{(f^{\otimes m})^{-1}(y')} \ket{\phi}}^2\qquad\textrm{by \eqref{eqn:listdecoding-decomp} and orthogonality of $\Pi_{(f^{\otimes m})^{-1}(y')}$}\\
    &\leq \max_{\ket{\phi} \in \listbadspace} \abs{\listofout}^{1/2} (\sum_{y' \in \range^m} \norm{ (\Pi_{f^{\otimes m})^{-1}(y')} \ket{\phi}}^4)^{1/2}\qquad\textrm{by Cauchy-Schwarz}\\
    &\leq \abs{\listofout}^{1/2} \left(\frac{\max(16,3 \abs{\range})^{\ell -1}}{\abs{\range}^m}\right)^{1/2} \qquad \textrm{by \eqref{eqn:list-dec-no-loss-key-lemma-statement}}\\
    &\leq \abs{\range}^{(m-k)/2} \left(\frac{\max(16,3 \abs{\range})^{\ell -1}}{\abs{\range}^m}\right)^{1/2}\\
    &= \abs{\range}^{-k/2} \max(16,3 \abs{\range})^{(\ell-1)/2}\\
    &\le \abs{\range}^{2(\ell-1)-k/2}\qquad\textrm{since $\abs{Y}\ge 2$}\\
    &\leq \abs{\range}^{-k/4} \qquad \textrm{since $\ell = \ceil{k/8}$.}
\end{align*}
We can therefore set $\eta'=\abs{\range}^{-k/4}$.
It remains to show the analogue of Part (iii) of the proof
of \cref{thm:selective-from-multiplicative-adversary} part (b).
Fix $\zeta = (\sqrt{1-\varepsilon} - \sqrt{1/\abs{\range}})^2$.
Since $1/2- 1/\abs{\range} -\sqrt{2/\abs{\range}}\le \varepsilon<1-1/\abs{\range}$, we have $\zeta \leq 1/2$.
Now define:
\begin{displaymath}
    \alpha = \left(\frac{1 + (\lambda - 1) \zeta}{\lambda}\right)^{1/8} = \left(1-\frac{(\lambda-1)(1-\zeta)}{\lambda}\right)^{1/8}\le e^{-(1-1/\lambda)/16}
\end{displaymath}
since $\zeta\le 1/2$.
Define $\kappa(z)=(\alpha^{z/2}+|\range|^{-z/8})^{2/z}$ for
variable $z\ge 1$. 
Observe that $\kappa(\cdot)$ is a strictly decreasing function of $z$.
Let $k_0\ge 1$ be the minimum value such that $\kappa(k_0)<1$ and set $\kappa=\kappa(k_0)$;
the value of $k_0$ can be bounded depending only on $\lambda$.

Then for $k \ge k_0$, we will focus on success probability $\kappa^{k}$. 
Therefore the multiplicative adversary expression for our progress target is
\begin{align*}
    1+(\lambda' - 1)(\kappa^{k/2} - \sqrt{\eta'})^2 &= 1+(\lambda' -1)(\kappa^{k/2} - 1/\abs{\range}^{k/8})^2\\
    &\geq 1+(\lambda' -1)\, \alpha^k\qquad\text{by definition of $\kappa$ since $k\ge k_0$}\\
    &\geq 1 + (\lambda^{k/8} -1)\left(\frac{1+(\lambda-1)\zeta}{\lambda}\right)^{k/8}\\
    &= 1 + (1-\lambda^{-k/8})(1+(\lambda-1) \zeta)^{k/8}\\
    & > (1+(\lambda -1)\zeta)^{k/8}\\
    &= \big(1+(\lambda-1)(\sqrt{1-\varepsilon}-\sqrt{1/\abs{\range}})^2\big)^{k/8}.
\end{align*}
Since this progress target is the $k/8$-th power of the progress target for $f$ with error $\varepsilon$, and the rate of progress $c$ per query is the 
same as the one for $f$, we obtain that the time required
to achieve success probability $\kappa^{k}$ is at
least $k/8$ times the multiplicative adversary time lower bound
for $f$ with success probability $1-\varepsilon$, as required.
\end{proof}

\subsection{Strong list-decoding direct product theorem for all Boolean-valued functions}\label{sec:list-decoding-for-all-boolean-out}

We proved in \cref{sec:neg-to-mult-bool} that
it is always possible to obtain asymptotically tight multiplicative adversary 
lower bounds 
for functions with Boolean outputs such that $\eta=1/2$.
This implies that \cref{thm:generic-list-decoding} is actually sufficient to get a strong list-decoding direct product theorem for all Boolean valued functions.

\paragraph{A strong list decoding direct product theorem for all Boolean valued functions}
The proof mirrors that for selective direct products in \cref{sec:strong-selective-prod-from-query-bound}.

\begin{theorem}
    \label{thm:list-allbool-precise}
Fix $\varepsilon = 1/10$.
There are constants $\alpha, \gamma>0$ and an integer $k_0\ge 1$  such that 
    for $f : \inputspace \to \{0,1\}$ where $\inputspace \subseteq \domain^n$, and for all integers $k$ with $k_0\le k\le m$ we have  
      $Q_{1-(1-\varepsilon)^{\gamma k}}(\easierlistdp{f}{k}{m})\ge \alpha k\cdot Q_{\varepsilon}(f)$.
\end{theorem}
\begin{proof}
    Choose $\varepsilon=1/10$, $\lambda=19/10$, and $\eta=1/2$ as in the statement of
    \cref{thm:bool-mult-bound-is-tight}, and let $C$ be the constant factor and $\dist$ be the distribution that results.
    We apply \cref{thm:generic-list-decoding} to yield absolute constants $\kappa<1, k_0$ such that for all $k\ge k_0$ we have
    $$Q^{\dist^m}_{1-\kappa^{k}}(\easierlistdp{f}{k}{m}) \geq  (k/8) \madv^\dist_{1/10, 19/10, 1/2}(f) \geq (k/8) Q_{1/10}(f)/C.$$
    We set $\gamma$ such that $\kappa=(1-\varepsilon)^\gamma$.
    Now $Q_\varepsilon(f)\ge Q^{\dist^*}_\varepsilon(f)$ for any distribution $\dist^*$ by
    Yao's Lemma and we can choose $\alpha=1/(8C)$ to obtain the claim.
\end{proof}

\section{Discussion}

We have shown that the relational multiplicative adversary method is universal for the quantum query complexity of functions up to a constant factor and given a specialized tighter proof of universality in the case of Boolean-valued functions. 

We have given three different recipes for proving strong selective direct product lower bounds for quantum query complexity:
(1) a direct recipe that works for a few specific functions,
(2) a second recipe that works for any relation for
which one can directly prove a quantum query lower bound using our formulation of the multiplicative adversary method, and 
(3)
a third recipe that follows, by a reduction to the second, for any function.

For strong list-decoding direct product theorems for quantum computation of functions, we have a similar direct recipe for specific functions as well as a
general theorem based on the multiplicative adversary method, though the latter has the condition 
that the parameter $\eta$ in that method must be
precisely the inverse of the size of the function's range.
Because we have shown that this multiplicative adversary method with $\eta=1/2$ is universal up to a constant factor for the computation of all partial functions with Boolean outputs, we obtain strong list-decoding direct product theorems for all (partial) Boolean-valued functions, providing a quantum analogue to the classical results of Ben-David and Blais~\cite{DBLP:conf/focs/BenDavidB25}.

The eigenvalue properties of negative-weight adversary matrices that we used to prove the property for $\eta$ in the conversion to the multiplicative adversary method in the Boolean case do not seem to have a nice analogue for larger range sizes.  It is therefore open whether a strong list-decoding direct product theorem holds for quantum computation of all functions with larger range.
(\cite{DBLP:conf/focs/BenDavidB25} does not consider the analogous question for classical computation.)

However, the most interesting open question concerning direct products for quantum query complexity is whether
an ordinary quantum strong direct product lower bound holds for arbitrary relations.   
Drucker~\cite{DBLP:journals/cc/Drucker12}, who proved such a theorem for classical query complexity, asked this question more than a dozen years ago, shortly after the strong direct product theorem of Lee and Roland~\cite{LeeRolandSQDPT} for functions and state generation first appeared.

We also have given an example of how strong selective direct product lower bounds can be used to yield quantum time-space tradeoffs for general quantum algorithms when ordinary strong direct product lower bounds can only yield tradeoff lower bounds for output-oblivious algorithms.

It would be good to give other examples of interesting problems where we can apply our theorems.   
One natural class of problems that we expect will be interesting are lifted problems, where one takes some
multi-output function $g:\outputspace^m\rightarrow Z^\ell$ for which a time-space tradeoff
lower bound is known and substitutes the outputs of $m$ independent copies
of the function $f:X^n\rightarrow \outputspace$, yielding the
function $g\circ f^m$.
One would hope to show that the lower bound grows beyond that of $g$ by a factor that is a function of the quantum query complexity of $f$.
Another natural problem for list decoding would be computing the Gram matrix $G^f$ for a Boolean valued function $f$ where $G^f_{i,j}(x_1, \ldots, x_m) = \indicator_{f(x_i) = f(x_j)}$.
Such a time-space tradeoff likely follows from viewing $k$ partial outputs as linear constraints on the possible assignments of the $f(x_i)$.

We do not expect, however, that any form of direct product theorems
(that use independent inputs) will let us improve the quantum lower bounds for sorting or Boolean matrix multiplication. 
Any lower bound for general quantum algorithms must
use one fixed input distribution. 
The existing lower bounds for both problems in the output-oblivious case vary the distribution based on which block of the computation is being analyzed.

Finally, we note that the matrix mapping problem $FA\tilde F^T$ we consider is a generalization of the 
$PAQ$ problem of Abrahamsson \cite{DBLP:journals/jcss/Abrahamson91} which has the extra condition that the
two input function matrices $F,\tilde F$ are actually permutation matrices $P$ and $Q^T$.
This is the only problem appearing in \cite{DBLP:journals/jcss/Abrahamson91} that does not have an analogous bound in the quantum setting in \cite{bkw:qmatrix-journal} or here. 
Problems such as this, where the input coordinates that can be queried 
are correlated with each other, are not easily dealt with
in the recording query framework.  
There is slight correlation across rows in function matrices but we could handle that correlation with a different recording query encoding; that is not enough to handle permutation matrices which are weakly correlated across both rows and columns.

Rosmanis~\cite{Rosmanis21} introduced a version of the compressed oracles framework 
that extends the recording query encoding to input permutations using the representation theory of the symmetric group in order to yield quantum query lower bounds
for the problem of permutation inversion with limited storage for advice;
these were recently extended to optimal bounds in~\cite{DBLP:journals/corr/AkshimaBCGY2510}. 
However, this method only applies to inputs where the permutation mapping can be directly accessed, not the indirect version given by permutation matrices, which have much weaker correlations than directly-accessed permutations. 

More generally, we do not want to be required to develop
a full representation theory and separate extensions of the recording query encoding for each kind of symmetry/constraint that might appear in input distributions for different problems.   
While that technically would be required to get optimal results, it would be good to develop a notion of approximation for the input distribution that 
lets us get around the issue.  
For example, the classical lower bound
of Abrahamson for the $PAQ$ problem uses an approximation
based on the fact that the matrices $P$ and $Q^T$ look very much like function matrices for much of the computation.

\section{Acknowledgements}
Niels Kornerup's contributions to this work were supported by the Laboratory Directed Research and Development program (project
240650) at Sandia National Laboratories, a multimission laboratory managed and operated by National
Technology and Engineering Solutions of Sandia LLC, a wholly owned subsidiary of Honeywell International
Inc. for the U.S. Department of Energy’s National Nuclear Security Administration under contract DE-
NA0003525.
This paper describes objective technical results and analysis. Any subjective views or opinions that might be expressed in the paper do not necessarily represent the views of the U.S. Department of Energy or the United States Government.
The publisher acknowledges that the U.S. Government retains a non-exclusive, paid-up, irrevocable, world-wide license to publish or reproduce the published form of this written work or allow others to do so, for U.S. Government purposes. The DOE will provide public access to results of federally sponsored research in accordance with the \href{https://www.energy.gov/doe-public-access-plan}{DOE Public Access Plan}.
\subsection{Discussion of LLM use}\label{sec:ai-use}
Prior to the current version of \cref{sec:list-dec}, we had proven explicit list decoding direct product theorems for $\search$, $\OR$, and $\Parity$ as well as a generic result for any function satisfying \cref{prop:mult-adv} with $\eta = 1/\abs{\range}$ that all featured an $O(\log_{\abs{\range}} (m/k))$ factor loss in both the number of queries and the final probability bound.
This is the argument that we termed a union bound or first moment method in our discussion in \cref{sec:list-dec}.


The idea of using 4-norms and the Cauchy-Schwarz inequality that is in the current argument is entirely the result of our interactions with ChatGPT 5.6.  
However, neither the LLM, nor any of the references that it cited describes any arguments involving 4-norms as a second-moment method for quantum computing, something that seems very natural in retrospect. 
The work that is the result of that interaction is encapsulated in the statement and proof of 
\cref{cor:list-dec-no-loss-key-lemma} and the use of the
Cauchy-Schwarz inequality where that lemma is applied in 
\cref{sec:list-dec}.

We had suspected that our existing union bound argument was not tight, and fed it into the LLM to see if it could find a way to tighten the argument, specifically the multiplicative adversary version.
The LLM pointed towards the existence of a ``hypercontractive small-set bound that is independent of $m$'' that
could be slotted into our existing proof as a way to remove the log factor loss.
It claimed a bound that is related to but somewhat incomparable to the bound that we eventually were able to extract after an extended series of interactions.

For reasons that will become clear, we did not attempt to verify whether its precise claim was true or that its method of proving it was sound.
The LLM's stated reasoning focused on ``Efron-Stein degree'' and claimed that one could express the target of our analysis in terms of
Hilbert-valued functions on the range $\range$ that
have low Efron-Stein degree.
Then it claimed that, via Gaussian scalarization, an extension of Bonami's inequality relating 4-norms and 2-norms for low degree functions~(\cite{Bonami1970}, cited only indirectly and not by name by the LLM, though it did query~\cite{moo:journal} and many other references)  holds for Hilbert-valued functions with low Efron-Stein degree.
Finally, the LLM claimed that some kind of application of this 4-norm bound on $\range^m$ would give its claimed bound.

The use of the term ``Efron-Stein degree'', which is a concept that we could not find at the time and did not appear in our subsequent search of any of its cited sources was quite confusing,
 and the need to represent our fairly simple questions as general Hilbert-valued functions seemed overkill, which is why we did not directly try to work through the LLM's claimed argument.

Before discussing how we proceeded from that point, we want to give appropriate credit for the essential ideas
that the LLM seems to have combined.
The LLM arrived at this sketch of a solution while querying many sources.  Subsequent to finishing our work we have spent some time trying to tease out where the essential ideas originated.   
We believe that the terminology of Efron-Stein degree 
originates in a recent paper of Ellis, Kindler, and Lifshitz~\cite{DBLP:conf/stoc/EllisKL23}, which the LLM did not query or cite. That paper, inspired by ideas in~\cite{KeevashLLM2024}, which the LLM did query, showed extensions of hypercontractive inequalities to
linear maps that we think were also the basis for the LLM's formulation in terms of Hilbert-valued functions, though that term is never used in~\cite{DBLP:conf/stoc/EllisKL23}.  We have not tried to find a source for its use of the notion of Gaussian scalarization.

The main question is 
how would an LLM connect hypercontractivity and 4-norms to quantum computing analysis in the first place? 
As noted above, we could not find any mention of considering 4-norms as second moments of probabilities in quantum computing in any of the references queried by the LLM and have not found that described elsewhere either.  We also could not find specific use of hypercontractivity with quantum query complexity in those references or in other works that the LLM did not directly query.
However, the LLM did query a decade-old survey of the applications of hypercontractive inequalities in the complexity of quantum computing~\cite{Mon12}; the closest works cited in that survey are papers on one-way quantum communication complexity and quantum codes~\cite{DBLP:journals/siamcomp/GavinskyKKRW08,DBLP:conf/focs/Ben-AroyaRW08}, neither of which were queried or cited by the LLM.  (The latter paper also proved a form of hypercontractivity for matrix-valued functions that is likely related to the later work
in~\cite{DBLP:conf/stoc/EllisKL23}, though the later work is clearly independent.)

It is relatively straightforward to show that the quantities involved in quantum query algorithms with few queries can be expressed as polynomials with low degree, so once one has a connection to hypercontractivity, and 4-norms specifically,
it becomes natural to try to convert the argument into a version to which one could apply a form of Bonami's hypercontractive inequality.

The LLM queried large numbers of other references, many of which make substantial contributions or are excellent surveys, but the ones mentioned
seem to be the most significant in retrospect.

Returning to our interaction with the LLM, we believed that it would be better to modify the proof and find a version that is more compatible with the intuition we developed working on strong selective direct-product theorems than to try to directly verify and recreate this argument.
In order to keep the argument as self-contained as possible, we prompted the LLM to simplify the proof by not explicitly using hypercontractivity and
instead to use an inductive argument for our problem from first principles based on the structure of our union bound
argument.
%
%
We then had the LLM walk through the modified proof step by step and, for each step, we queried it until we were thoroughly convinced of correctness and our ability to correctly reproduce the argument.
The structure of the resulting argument is similar to the proof of the Bonami lemma in \cite{DBLP:conf/focs/MosselOO05,moo:journal,DBLP:books/daglib/0033652}.

The result is \cref{sec:list-dec-lemma}, which recreates (hand-written and verified) versions of the LLM generated proofs with improved clarity, rigor, and notation that matches what we use in the rest of the paper.

LLMs were not substantially used in any part of our results beyond the impacts we mentioned here that affected \cref{sec:list-dec,sec:list-dec-lemma}.
However, ChatGPT 6, and Claude Fable 5.1 were used to catch issues with an earlier version of this manuscript, including pointing out that our formulation of the multiplicative adversary method using the trace progress measure was equivalent to the formulation with the matrix norm progress measure in \cite{DBLP:conf/icalp/JefferyZ26}.
In this version, these issues have all been addressed by the authors.
All text in this paper (including proofs) were written and verified by the authors, who take full responsibility for the correctness and originality of all contained content.

\bibliographystyle{alpha}
\bibliography{refs}
\newpage
\appendix
\section{Single copy lower bounds for \texorpdfstring{$\search$}{Search}, \texorpdfstring{$\OR$}{OR} and \texorpdfstring{$\Parity$}{Parity}}
\label{sec:single-copy-search-or}

\subsection{The Promise Search problem}

We construct a simple proof of the following lower bound that was previously proven in \cite{DBLP:journals/siamcomp/BennettBBV97}.
\begin{lemma}\label{thm:search}
    Let $\dist$ be the distribution that is uniform over weight one binary strings of length $n$.
    Any quantum query algorithm solving $\search_n$ on the input distribution $\dist$ probability $2/3$ requires $\Omega(\sqrt{n})$ queries.
\end{lemma}

\begin{proof}
    The argument proceeds by breaking the input space into orthogonal subspaces: a space $V^0$ containing the initial input state on which the algorithm has no advantage in computing $\search_n$ and its orthogonal complement $V^1$.   
    It then shows that each query can only shift a small amount of the amplitude from $V^0$ to $V^1$.
    
    Let $\ket{j} \in \C^n$ denote the $j$-th standard basis vector, so we can purify $\dist$ as $\ket{\chi_0^n} = \frac{1}{\sqrt{n}}\sum_{j \in [n]} \ket{j}$ (in keeping notation from \cref{sec:specific}, we drop the $n$ superscript henceforth).
    Let $\calA$ be a quantum query algorithm that solves the $\search_n$ problem on input distribution $\dist$ with probability at least $2/3$.
    The final state of $\calA$ before measurement can be written as:
    \begin{displaymath}
        \ket{\psi_T} = U_T \oracle U_{T-1} \oracle \ldots U_1 \oracle U_0 \ket{0,0,0} \ket{\chi_0}
    \end{displaymath}
    In this case, $\pisucc$ projects onto basis states $\ket{i,p,w}\ket{q(w)}$.
    Next we define projector $\Pi_{V^0} = \identity_{\workreg} \otimes (\op{\chi_0}{\chi_0})_\inputreg$ which projects the input portion of the state onto the subspace $V^0=\spn(\chi_0)$.
    The success probability of $\calA$ can then be bounded as
    \begin{align}
        \norm{\pisucc \ket{\psi_T}}^2  &\leq \bigg( \norm{\pisucc (\identity - \Pi_{V^0}) \ket{\psi_T}} + \norm{\pisucc \Pi_{V^0}  \ket{\psi_T}}\bigg)^2 \label{eqn:good-bad-split}\\
        &\leq \bigg( \norm{(\identity - \Pi_{V^0}) \ket{\psi_T}} + \norm{\pisucc \Pi_{V^0} \ket{\psi_T}}\bigg)^2 \nonumber.
    \end{align}
    In general, we know that $\Pi_{V^0} \ket{\psi_T}$ can be expressed as $\sum_{i,p,w} \alpha_{i,p,w} \ket{i,p,w} \otimes \ket{\chi_0}$ for some set of coefficients $\alpha_{i,p,w}$ where $\sum_{i,p,w} |\alpha_{i,p,w}|^2 \leq 1$.
    This implies that
   \begin{align*}
        \norm{\pisucc \Pi_{V^0} \ket{\psi_T}} &= \norm{\pisucc\sum_{i,p,w} \alpha_{i,p,w} \ket{i,p,w} \otimes \sum_{j \in [n]} \frac{1}{\sqrt{n}} \ket{j}}\\
        &= \norm{\sum_{i,p,w} \frac{\alpha_{i,p,w}}{\sqrt{n}} \ket{i,p,w} \ket{q(w)}}\\
        &= \sqrt{\sum_{i,p,w} |\alpha_{i,p,w}|^2 / n} \leq \frac{1}{\sqrt{n}}.
    \end{align*}

    All that remains is to show that $\norm{(\identity - \Pi_{V^0}) \ket{\psi_T}}$, the amplitude of the projection onto the subspace where the input portion is in the orthogonal complement $V^1$ of $V^0$, remains small after $T$ queries.
    We will do this by showing that each query can at most increase the amplitude on this subspace by a factor of $O(1/\sqrt{n})$, implying that $\Omega(\sqrt{n})$ queries are necessary for any quantum query algorithm to have a constant success probability.

    Define $\Delta_t = \norm{(\identity - \Pi_{V^0}) \ket{\psi_t}}$. We will show by induction that $\Delta_t \leq \frac{2t}{\sqrt{n}}$.
    By construction $\Delta_0 = 0$ and
    \begin{align}
        \Delta_t &= \norm{(\identity - \Pi_{V^0}) \ket{\psi_t}} \label{eqn:delta-t}\\
        &\leq \norm{(\identity - \Pi_{V^0}) U_t \oracle \Pi_{V^0}\ket{\psi_{t-1}}} + \norm{(\identity - \Pi_{V^0}) U_t \oracle (\identity - \Pi_{V^0})\ket{\psi_{t-1}}}\nonumber\\
        &\leq \norm{(\identity - \Pi_{V^0}) U_t \oracle \Pi_{V^0}\ket{\psi_{t-1}}} + \norm{(\identity - \Pi_{V^0})\ket{\psi_{t-1}}}\nonumber\\
        &= \norm{(\identity - \Pi_{V^0}) U_t \oracle \Pi_{V^0}\ket{\psi_{t-1}}} + \Delta_{t-1}\nonumber.
    \end{align}
    As before, $\Pi_{V^0} \ket{\psi_t} = \sum_{i,p,w} \alpha_{i,p,w} \ket{i,p,w}\ket{\chi_0}$ for some coefficients with $\sum_{i,p,w} |\alpha_{i,p,w}|^2 \leq 1$.
    Observe also since $U_t$ commutes with $\identity - \Pi_{V^0}$, we can safely drop $U_t$ from \eqref{eqn:delta-t} yielding
    \begin{align}
        \Delta_t &\leq \norm{(\identity - \Pi_{V^0}) \oracle  \sum_{i,p,w} \alpha_{i,p,w} \ket{i,p,w}\ket{\chi_0}} + \Delta_{t-1}\label{eqn:delta-t-2}\\
        &= \norm{(\identity - \Pi_{V^0}) \oracle \big( \sum_{i,p,w} \alpha_{i,p,w} \ket{i,p,w} \otimes \frac{1}{\sqrt{n} } \sum_{j \in [n]} \ket{j} \big)} + \Delta_{t-1}.\nonumber
    \end{align}
    Now we can apply $\oracle$ to \eqref{eqn:delta-t-2} and carefully rearrange the terms to get that the right-hand side equals
    \begin{equation}\label{eqn:delta-t-3}
        \norm{(\identity-\Pi_{V^0}) \bigg(\sum_{i,p,w} \beta_{i,p,w} \ket{i,p,w} \otimes \ket{\chi_0}-\frac{2 \alpha_{i,p,w}}{\sqrt{n}} \ket{i,p,w}\ket{i} \bigg)} + \Delta_{t-1}.
    \end{equation}
    for some complex coefficients $\beta_{i,p,w}$ (depending on the various $\alpha_{i,p,w}$).
    Importantly, $\identity - \Pi_{V^0}$ annihilates the first term inside the norm in \eqref{eqn:delta-t-3} since its input portion is in $V^0$.
    Since $\identity - \Pi_{V^0}$ is a projector and $\sum_{i,p,w} |\alpha_{i,p,w}|^2 \leq 1$, we can finally conclude that
    \begin{equation}
        \Delta_t \leq \bignorm{(\identity-\Pi_{V^0}) \bigg(-\frac{2}{\sqrt{n}} \sum_{i,p,w} \alpha_{i,p,w} \ket{i,p,w}\ket{i} \bigg)} + \Delta_{t-1} \leq \frac{2}{\sqrt{n}} + \Delta_{t-1} \leq \frac{2t}{\sqrt{n}}.
    \end{equation}
    Plugging this back into \eqref{eqn:good-bad-split} gives that the probability that $\calA$ is correct after $T$ queries is at most
    \begin{displaymath}
          \left(\frac{2T}{\sqrt{n}} + \frac{1}{\sqrt{n}}\right)^2 = \frac{(2T+1)^2}{n}.
    \end{displaymath}
    Thus, to obtain a success probability that is at least $2/3$, we must have that $\frac{(2T+1)^2}{n} \geq 2/3$, which implies that $T$ is $\Omega(\sqrt{n})$.
\end{proof}

\subsection{The OR problem}

As was the case in \cref{sec:or}, the input distribution for the search problem we used is not conducive to a lower bound for OR, as there is always a marked item in the input.
As before, we create a new hard distribution $\dist'$ which has probability $1/(2n)$ on each weight one binary string of length $n$, and weight $1/2$ on the all-zeroes string.

\begin{sloppypar}
As in \cref{sec:or}, we choose a specific purification that initializes the input register in the superposition $\sqrt{\frac{1}{2n}}\sum_{i \in [2n]}\ket i \in \C^{2n}$, where we have identified input $\ket{0^n}$ with the superposition $\frac{1}{\sqrt{n}}\sum_{j=n+1}^{2n} \ket j$.
This will allow us to directly port over much of the analysis from the proof of \cref{thm:search}.
Note that the algorithm will only ever query superpositions of coordinates $1,\ldots, n$ so a $1$ in any other index is undetectable to the algorithm.
\end{sloppypar}

\begin{corollary}
    Suppose that $\calA$ is a quantum query algorithm solving $\OR_n$ on input distribution $\dist'$ which has weight $1/(2n)$ on each $x\in \bits^n$ of Hamming weight 1, and weight $1/2$ on the all-zeroes input.
    If $\calA$ succeeds with probability at least $2/3$, then $\calA$ must make at least $\Omega(\sqrt{n})$ quantum queries.
\end{corollary}

\begin{proof}
    In this setting, the worktape of the algorithm is measured in the standard basis and interpreted as a single bit $q(w)$. 
    To be very explicit, $\pisucc$ now projects onto states of the form $\ket{i,p,w}\ket{j}$ for $j \in [n]$ when $q(w) = 1$, and otherwise projects onto states of the form $\ket{i,p,w}\ket{j}$ for $j\in\curly{n+1,\ldots, 2n}$.

    As before, we define projectors $\Pi_{V^0} = \identity_{\workreg} \otimes (\op{\chi_0^{2n}}{\chi_0^{2n}})_\inputreg$.
    The success probability of $\calA$ can then be bounded as 
    \begin{displaymath}
         \norm{\pisucc \ket{\psi_T}}^2 \leq \left( \norm{(\identity - \Pi_{V^0}) \ket{\psi_T}} +\norm{\pisucc \Pi_{V^0} \ket{\psi_T}} \right)^2 
    \end{displaymath}
    Next we observe that $\norm{\pisucc \Pi_{V^0} \ket{\psi_T}}$ is at most $1/\sqrt{2}$ by the same reasoning as the $1/\sqrt{n}$ bound in \cref{thm:search}.
    Since we define the subspaces $V^0$ and $V^1$ in the same way as in \cref{thm:search}, using $2n$ coordinates instead of $n$, and the oracle, constrained to the first $n$ bits, behaves in the same way as previously in \cref{thm:search}, we can apply the same analysis to get that $\norm{(\identity - \Pi_{V^0}) \ket{\psi_T}} \leq 2T/\sqrt{2n}$.
    Thus the overall success probability of $\calA$ is bounded by:
    \begin{displaymath}
         \left(\frac{2T}{\sqrt{2n}} + \frac{1}{\sqrt{2}}\right)^2.
    \end{displaymath}
    Which implies that $T$ must be $\Omega(\sqrt{n})$ in order to achieve a success probability of at least $2/3$.
\end{proof}

\subsection{The Parity problem}

We give an alternative simple and direct proof of the lower bound first proven using the polynomial method in~\cite{DBLP:journals/jacm/BealsBCMW01}.

\begin{lemma}
    \label{prop:single-copy-parity}
    Any quantum query algorithm computing $\Parity_n$ with at most $T < n/2$ queries has success probability at most $1/2$.
\end{lemma}

We first identify a class of input states that gives
no advantage in computing $\Parity_n$.

\begin{proposition}
    \label{claim:parity-key}
    Suppose that $\ket{\psi}$ is a state on $n$ qubits of the form
    $$\ket{\psi} = \sum_{S \subseteq [n]: |S| < n/2} \alpha_S \ket{+}_{\overline{S}}\ket{-}_S,$$
    where $\alpha_S \neq 0$ iff $|S| < n/2$.
    Then $\ket{\psi}$, when written in the standard basis $\ket{\psi} = \sum_{x \in \bits^n} \beta_x \ket{x}$ is such that
    $$\sum_{x: |x| \text{ even}} |\beta_x|^2 = \sum_{x: |x| \text{ odd}} |\beta_x|^2.$$
    Equivalently, the probability of finding an even weight hamming string when measuring $\ket{\psi}$ in the standard basis is exactly $1/2$.
\end{proposition}

\begin{proof}
    Let $p_{even} := \sum_{x: |x| \text{ even}} |\beta_x|^2$, and $p_{odd} := \sum_{x: |x| \text{ odd}} |\beta_x|^2$, and recall the Pauli $Z$ matrix:
    $$
    Z := \begin{bmatrix}
        1 & 0 \\
        0 & -1
    \end{bmatrix}.
    $$

    Observe that $Z \ket 0 = \ket{0}, Z \ket 
    1 = -\ket 1, Z \ket{+} = \ket -,$ and $Z \ket{-} = \ket{+}$.
    Then we have that
    \begin{align*}
        \bra{\psi}Z^{\otimes n} \ket \psi &= \bigg(\sum_{x \in \bits^n} \beta_x^* \bra x\bigg) Z^{\otimes n} \sum_{y \in \bits^n}\beta_y \ket y\\ 
        &= \bigg(\sum_x \beta_x^* \bra x \bigg)\sum_{y}(-1)^{\indicator\{|y| \text{ odd}\}}\beta_y\ket y \\
        &= \sum_{x: |x| \text{ even}} |\beta_x|^2 - \sum_{x: |x| \text{ odd}} |\beta_x|^2 = p_{even} - p_{odd}.
    \end{align*}
    But, since $Z$ swaps $\ket{+} $ and $\ket{-}$, we have that
     \begin{align*}
          p_{even} - p_{odd} &= \bra{\psi}Z^{\otimes n} \ket \psi \\
          &= \sum_{|S| < n/2} \alpha_S^* \bra{+}_{\overline S}\bra{-}_S \sum_{|T| < n/2} \alpha_T \ket{-}_{\overline{T}} \ket{+}_{ T} \\
          &= 0,
     \end{align*}
     since $\{\ket{+}_S\ket{-}_{\overline S}: S \subseteq [n]\}$ form an orthonormal basis.
\end{proof}

\begin{proof}[Proof of \cref{prop:single-copy-parity}]
    Our hard distribution will be the uniform distribution over $n$-bit strings.
    Thus, our oracle register is initialized to $\ket{+}^{\otimes n}$, and as usual our algorithm initializes its index, phase, and worktape register to $\ket{0, 0, 0}$. 
    The state after $t$ quantum queries looks like 
    $$\sum_{i,p,w, |S| \leq t} \alpha_{i,p,w, S}\ket{i, p, w} \ket{-}_S \ket{+}_{\overline{S}} = \sum_{i, p,w}\beta_{i,p,w} \sum_{|S| \leq t} \beta^{i, p,w}_S \ket{i, p,w} \ket{+}_{\overline{S}}\ket{-}_S,$$
    where $\sum |\alpha_{i,p,w,S}|^2 = 1$, and $\alpha_{i,p,w,S} = \beta_{i,p,w} \cdot \beta^{i,p,w}_S $.

    In particular, the algorithm will measure the work tape in the standard basis and apply some input-independent post-processing function $q$ to the result in order to decide what its answer is for the $\Parity_n$ function.
    After this measurement, the quantum state in the oracle register is of the form
    $$\ket{\psi_t}_\calO \sum_{|S| \leq t} \alpha_S\ket{+}_{\overline{S}}\ket{-}_S, $$
    where $\sum_S |\alpha_S|^2 = 1$.
    Thus, by \cref{claim:parity-key} and the assumption that $t < n/2$, no matter what post-processing the quantum algorithm does, the probability of correctly computing the parity of the input is precisely $1/2$, i.e. no better than a random guess.
\end{proof}

\begin{remark}
This proof follows a similar template to the proofs
for $\search_n$ and $\OR_n$, this time with the space
$V^0$ as the span of all states of the form
given by \cref{claim:parity-key}, namely
$V^0=\spn(\set{\ket{+}_S\ket{-}_{\overline{S}}\ :\ 0\le |S|<n/2})$.
In this case $V^0$ is not one-dimensional and instead of queries moving
amplitude from $V^0$ to $V^1$ slowly at every step, the amplitude shift can happen only once the time step $T$ passes $n/2$.
\end{remark}

\section{Proof of \cref{cor:list-dec-no-loss-key-lemma}}\label{sec:list-dec-lemma}

In this section, we will prove the following bound:

\listdecnolosskeylemma*

Before diving into any technical details, we note that the proof of \cref{cor:list-dec-no-loss-key-lemma} is very similar to that of Lemma 3.2 in \cite{DBLP:conf/focs/MosselOO05} if one equates Fourier degree at most $\ell$ with membership in $B_{j,\ell}$.
As pointed out in \cref{sec:ai-use}, a more direct path to proving \cref{cor:list-dec-no-loss-key-lemma} might involve equating vectors in $\bigoplus_{u \in \{0,1\}^{m}, \abs{u}\le \ell} V^u$ with polynomials of degree $\le \ell$ and then directly applying the Bonami Lemma.
We instead derive the desired bound from first principles since it gives a more self-contained proof.

To prove \cref{cor:list-dec-no-loss-key-lemma}, we will need the help of the following key properties:

\begin{lemma}\label{lem:bad-space-output-proj-interaction}
    If $\norm{\Pi_{f^{-1}(y)} \Pi_{V^0}}^2 \leq 1/\abs{\range}$ for every $y \in \range$, then:
    \begin{displaymath}
        \Pi_{V^0}  \Pi_{f^{-1}(y)} \Pi_{V^0} = \Pi_{V^0}/\abs{\range}.
    \end{displaymath}
\end{lemma}

\begin{proof}
    By the assumption of the lemma, for every $y\in \range$ and every vector $v\in \mathbb{C}^\inputspace$, 
    \begin{displaymath}
        v^\dagger \big(\Pi_{V^0} \Pi_{f^{-1}(y)} \Pi_{V^0}\big) v = \norm{\big(\Pi_{f^{-1}(y)} \Pi_{V^0}\big) v}^2 \leq \norm{
        \big(\Pi_{V^0}\big) v}^2 / \abs{\range} = v^\dagger \Pi_{V^0}\ v / \abs{\range},
    \end{displaymath}
    implying that $\Pi_{V^0} \Pi_{f^{-1}(y)} \Pi_{V^0} \preceq \Pi_{V^0} / \abs{\range}$.
    Since $\sum_{y \in \range} \Pi_{f^{-1}(y)} = \identity$, we have $\sum_{y \in \range} (\Pi_{V^0} \Pi_{f^{-1}(y)} \Pi_{V^0} - \Pi_{V^0}/\abs{\range})= 0$ and we obtain that every $\preceq$ must be equality.
    That is, for every $y \in \range$, $\Pi_{V^0}\Pi_{f^{-1}(y)} \Pi_{V^0} = \Pi_{V^0}/\abs{\range}$.
\end{proof}

The following technical lemma is the key idea that will let
us bound 4-norms in an asymmetric
way.
The asymmetry is important so that terms correctly cancel in our proof of \cref{lem:4-norm-single-copy-ineq} (the obvious symmetric version of this lemma would prevent us from using the fact that $\sum_{y \in \range} \phi_{0,y}^\dagger \phi_{1,y} = 0$ in that proof).

\begin{lemma}\label{lem:vector-sum-norm-fourth-power-ineq}
    Let $\phi$ and $\psi$ be two vectors. Then:
    \begin{displaymath}
        \norm{\phi + \psi}^4 \leq \norm{\phi}^4 + 8 \norm{\phi}^2 \norm{\psi}^2 + 3\norm{\psi}^4 + 4 \norm{\phi}^2 \cdot \text{Real}(\phi^\dagger \psi).
    \end{displaymath}
\end{lemma}

\begin{proof}
    Observe that:
    \begin{align*}
        \norm{\phi + \psi}^4 &= ((\phi + \psi)^\dagger (\phi+\psi))^2\\
        &=(\phi^\dagger\phi + \phi^\dagger\psi + \psi^\dagger\phi + \psi^\dagger \psi)^2\\
        &=(\phi^\dagger\phi)^2 + 2(\phi^\dagger\phi)(\phi^\dagger\psi) + 2(\phi^\dagger\phi)(\psi^\dagger\phi) + 2(\phi^\dagger\phi)(\psi^\dagger\psi) \\
        & \qquad + (\phi^\dagger\psi)^2 + 2(\phi^\dagger\psi)(\psi^\dagger \phi) + 2(\phi^\dagger\psi)(\psi^\dagger\psi) + (\psi^\dagger\phi)^2 + 2(\psi^\dagger\phi)(\psi^\dagger\psi) + (\psi^\dagger \psi)^2\\
        &= \norm{\phi}^4 + 2\norm{\phi}^2\norm{\psi}^2 + \norm{\psi}^4 + (\phi^\dagger\psi + \psi^\dagger\phi)^2 + 2 (\norm{\phi}^2+\norm{\psi}^2) ((\phi^\dagger \psi) + (\psi^\dagger\phi))\\
        &=\norm{\phi}^4 + 2\norm{\phi}^2\norm{\psi}^2 + \norm{\psi}^4 + 4(\text{Real}(\phi^\dagger \psi))^2 + 4 (\norm{\phi}^2+\norm{\psi}^2) \text{Real}(\phi^\dagger \psi)\\
        &\leq \norm{\phi}^4 + 6\norm{\phi}^2\norm{\psi}^2 + \norm{\psi}^4 +4\norm{\psi}^3\norm{\phi}+ 4 \norm{\phi}^2 \text{Real}(\phi^\dagger \psi) \\
        &\qquad\textrm{since $|\text{Real}(\phi^\dagger \psi)|\le \norm{\phi}\norm{\psi}$ by Cauchy-Schwarz}\\
        &\leq \norm{\phi}^4 + 8\norm{\phi}^2\norm{\psi}^2 + 3 \norm{\psi}^4 + 4 \norm{\phi}^2 \text{Real}(\phi^\dagger \psi) \\
        &\qquad\textrm{since $2\norm{\psi}\norm{\phi}\le \norm{\phi}^2+\norm{\psi}^2$ by the arithmetic-geometric mean inequality.}\qedhere
    \end{align*}
\end{proof}

We now see how to apply \cref{lem:vector-sum-norm-fourth-power-ineq} in the context of the
projection $\Pi_{f^{-1}(y)}$ on the $j$-th coordinate
where the norm of the projection on the subspace
$V^0$ is much smaller than that on $V^1$. 
The last term in the bound will be 0 because the projections of the subspaces
$V^0$ and $V^1$ are orthogonal to each other.

\begin{lemma}\label{lem:4-norm-single-copy-ineq}
    Let $\phi_0 \in (\mathbb{C}^\inputspace)^{[m]\setminus\{j\}} \otimes V^0$ and $\phi_1 \in (\mathbb{C}^\inputspace)^{[m]\setminus\{j\}} \otimes V^1$.
    Then:
    \begin{displaymath}
    \sum_{y \in \range} \norm{(\identity_{[j-1]} \otimes \Pi_{f^{-1}(y)} \otimes \identity_{[m]\setminus [j]}) (\phi_0 +\phi_1)}^4 \leq \norm{\phi_0}^4/ \abs{\range} + 8 \norm{\phi_0}^2\norm{\phi_1}^2 / \abs{\range} + 3 \norm{\phi_1}^4.
    \end{displaymath}
\end{lemma}

\begin{proof}
    Let $\phi_{0,y} = (\identity_{[j-1]} \otimes \Pi_{f^{-1}(y)} \otimes \identity_{[m]\setminus [j]})\ \phi_0$ and $\phi_{1,y} = (\identity_{[j-1]}  \otimes \Pi_{f^{-1}(y)} \otimes \identity_{[m]\setminus [j]})\ \phi_1$ respectively.
    Observe that
    \begin{align}
        \norm{\phi_{0,y}}^2 &= \phi_{0}^\dagger\  (\identity_{[j-1]}  \otimes \Pi_{f^{-1}(y)} \otimes \identity_{[m]\setminus [j]})\  \phi_{0}\nonumber \\
        &=\phi_{0}^\dagger\ (\identity_{[j-1]}  \otimes \Pi_{V^0} \Pi_{f^{-1}(y)} \Pi_{V^0} \otimes \identity_{[m]\setminus [j]})\ \phi_{0}\nonumber \\
        &=\phi_{0}^\dagger\ (\identity_{[j-1]}  \otimes \Pi_{V^0} \otimes \identity_{[m]\setminus [j]})\  \phi_{0} / \abs{\range} \qquad\textrm{by \cref{lem:bad-space-output-proj-interaction}} \nonumber\\
        &= \norm{\phi_0}^2 / \abs{\range} \label{eqn:phizero-cut-y}\\
        \noalign{\textrm{and that}}
        \sum_{y \in \range} \norm{\phi_{1,y}}^2 &= \norm{\phi_1}^2\label{eqn:phione-cut-y-2norm}
        \end{align}
since the $\Pi_{f^{-1}(y)}$ for $y\in \range$
are orthogonal and sum to the identity.
Therefore, using the 1-norm 2-norm inequality and
\eqref{eqn:phione-cut-y-2norm}, we have
\begin{equation}
        \sum_{y \in \range} \norm{\phi_{1,y}}^4 
        \leq \bigg(\sum_{y \in \range} \norm{\phi_{1,y}}^2\bigg)^2 
        = \norm{\phi_1}^4 \label{eqn:phione-cut-y-4norm}.
    \end{equation}
Then:
\begin{align*}
        &\sum_{y \in \range} \norm{(\identity_{[j-1]}  \otimes \Pi_{f^{-1}(y)} \otimes \identity_{[m]\setminus [j]}) (\phi_0 +\phi_1)}^4\\
        &= \sum_{y \in \range} \norm{\phi_{0,y} + \phi_{1,y}}^4\\
        &\leq \sum_{y \in \range} \big(\norm{\phi_{0,y}}^4 + 8 \norm{\phi_{0,y}}^2 \norm{\phi_{1,y}}^2 + 3 \norm{\phi_{1,y}}^4 + 4\norm{\phi_{0,y}}^2 \cdot \text{Real}(\phi_{0,y}^\dagger\phi_{1,y})\big) \qquad \textrm{by \cref{lem:vector-sum-norm-fourth-power-ineq}}\\
        &= \sum_{y \in \range}\big( \norm{\phi_{0}}^4/\abs{\range}^2 + 8 \norm{\phi_{0}}^2 \norm{\phi_{1,y}}^2/\abs{\range} + 3 \norm{\phi_{1,y}}^4 + 4\norm{\phi_{0}}^2 \cdot \text{Real}(\phi_{0,y}^\dagger\phi_{1,y})/\abs{\range}\big)\qquad \textrm{by \eqref{eqn:phizero-cut-y}}\\
        &= \sum_{y \in \range} \big(\norm{\phi_{0}}^4/\abs{\range}^2 + 8 \norm{\phi_{0}}^2 \norm{\phi_{1,y}}^2/\abs{\range} + 3 \norm{\phi_{1,y}}^4 \big)\qquad \textrm{since $\sum_{y \in \range} \phi_{0,y}^\dagger\phi_{1,y} = \phi_0^\dagger\phi_1 = 0$}\\
        &\leq \norm{\phi_0}^4 / \abs{\range} + 8 \norm{\phi_0}^2\norm{\phi_1}^2 / \abs{\range} + 3 \norm{\phi_1}^4 \qquad \textrm{by \eqref{eqn:phione-cut-y-2norm} and \eqref{eqn:phione-cut-y-4norm}}
    \end{align*}
    as desired.
\end{proof}

Finally, we can prove the key inductive version that we use to show
\cref{cor:list-dec-no-loss-key-lemma}.
It upper bounds the 4-norm of the projectors that check the values of the first $j$
output coordinates when the input has at most $\ell$ 
coordinates in the space $V^1$.

\begin{lemma}\label{lem:4-norm-ineq}
For any $\phi \in B_{j,\ell} = \bigoplus_{u \in \{0,1\}^{j}, \abs{u} \leq \ell} V^u \otimes (\mathbb{C}^\inputspace)^{[m] \setminus [j]}$,
\begin{equation}
    \sum_{y' \in \range^{[j]}} \norm{\big(\Pi_{(f^{\otimes j})^{-1}(y')} \otimes \identity_{[m]\setminus [j]}\big)\ \phi}^4 \leq \frac{\max(16,3 \abs{\range})^{\ell}}{\abs{\range}^{j}} \norm{\phi}^4.\label{eq:Bjl}
\end{equation}
\end{lemma}

\begin{proof}
    We prove this by induction on $j$.
    Let $\phi \in B_{j,\ell}$.
    As the base cases, we observe that when $0\le j\le \ell$, \eqref{eq:Bjl} is vacuously true since
    the left side is at most $\norm{\phi}^4$ because the projectors are orthogonal.
    Now we assume by induction that \eqref{eq:Bjl} holds for all $\phi' \in B_{j-1,\ell}$ for all $\ell$ and define $\phi_0, \phi_1$ such that $\phi = \phi_0 + \phi_1$ where $\phi_0$ is the projection of $\phi$ onto $(\mathbb{C}^\inputspace)^{[m]\setminus\{j\}} \otimes V^0$ and $\phi_1$ is the projection of $\phi$ onto $(\mathbb{C}^\inputspace)^{[m]\setminus\{j\}} \otimes V^1$.
    Since $\phi_0$ and $\phi_1$ are orthogonal vectors,
    \begin{equation}\label{eqn:phi1-phi2-triangle}
        \norm{\phi}^2 = \norm{\phi_0}^2 + \norm{\phi_1}^2.
    \end{equation}
    By construction, we know that $\phi_0 \in B_{j-1,\ell}$ and $\phi_1 \in B_{j-1,\ell-1}$.\footnote{When $\ell = 0$ we observe that $\phi_1 = 0$ to avoid the inductive case.}
    We expand
    \begin{align}
    &\sum_{y' \in \range^{[j]}} \norm{\big(\Pi_{(f^{\otimes j})^{-1}(y')}\otimes \identity_{[m]\setminus [j]}\big)\ \phi}^4 \nonumber\\
    &=\sum_{y' \in \range^{[j-1]}}\ \sum_{y_j \in \range} \norm{(\identity_{[j-1]} \otimes \Pi_{f^{-1}(y_j)} \otimes \identity_{[m]\setminus [j]}) \,(\Pi_{(f^{\otimes(j-1)})^{-1} (y')} \otimes  \identity_{[m]\setminus [j-1]})\, (\phi_0 + \phi_1) }^4.\label{eq:expand}
        \end{align}
For each $y'\in \range^{[j-1]}$, define
$\phi^{y'}_0=(\Pi_{(f^{\otimes(j-1)})^{-1} (y')} \otimes  \identity_{[m]\setminus [j-1]})\, \phi_0$
and $\phi^{y'}_1=(\Pi_{(f^{\otimes(j-1)})^{-1} (y')} \otimes  \identity_{[m]\setminus [j-1]})\, \phi_1$.
We observe that $\phi^{y'}_0\in (\mathbb{C}^\inputspace)^{[m]\setminus [j]}\otimes V^0$ and
$\phi^{y'}_1\in (\mathbb{C}^\inputspace)^{[m]\setminus [j]}\otimes V^1$. 
We can therefore rewrite  \eqref{eq:expand} as
\begin{align*}
    &\sum_{y' \in \range^{[j-1]}}\ \sum_{y_j\in \range} \norm{(\identity_{[j-1]} \otimes \Pi_{f^{-1}(y_j)} \otimes \identity_{[m]\setminus [j]}) \,(\phi^{y'}_0+\phi^{y'}_1)}^4\\
    &\le \sum_{y' \in \range^{[j-1]}}\ \big(\norm{\phi^{y'}_0}^4/\abs{\range} + 8\norm{\phi^{y'}_0}^2 \norm{\phi^{y'}_1}^2/\abs{\range} +3\norm{\phi^{y'}_1}^4\big)\qquad\textrm{by \cref{lem:4-norm-single-copy-ineq}}\\
    &\le \big(\sum_{y' \in \range^{[j-1]}} \norm{\phi^{y'}_0}^4\big)/\abs{\range} + 8\big(\sum_{y' \in \range^{[j-1]}}\norm{\phi^{y'}_0}^4\big)^{1/2} \big(\sum_{y' \in \range^{[j-1]}}\norm{\phi^{y'}_1}^4\big)^{1/2}/\abs{\range} +3\sum_{y' \in \range^{[j-1]}}\norm{\phi^{y'}_1}^4
    \end{align*}
by applying Cauchy-Schwarz to vectors $(\norm{\phi^{y'}_0}^2)_{y'}$ and
    $(\norm{\phi^{y'}_1}^2)_{y'}$.
We now apply the inductive hypothesis to bound the
quantities
$\sum_{y' \in \range^{[j-1]}} \norm{\phi^{y'}_0}^4$ and $\sum_{y' \in \range^{[j-1]}} \norm{\phi^{y'}_1}^4$.
Observe that 
\begin{equation*}
\sum_{y' \in \range^{[j-1]}} \norm{\phi^{y'}_0}^4
=\sum_{y' \in \range^{[j-1]}}\norm{(\Pi_{(f^{\otimes(j-1)})^{-1} (y')} \otimes  \identity_{[m]\setminus [j-1]})\, \phi_0}^4
\end{equation*}
where $\phi_0\in B_{j-1,\ell}$, so by the inductive hypothesis, 
\begin{equation*}
\sum_{y' \in \range^{[j-1]}} \norm{\phi^{y'}_0}^4
\leq \frac{\max(16,3 \abs{\range})^{\ell}}{\abs{\range}^{j-1}} \norm{\phi_0}^4.
\end{equation*}
Similarly since $\phi_1\in B_{j-1,\ell-1}$, by the inductive hypothesis, we have
\begin{equation*}
\sum_{y' \in \range^{[j-1]}} \norm{\phi^{y'}_1}^4
\leq \frac{\max(16,3 \abs{\range})^{\ell-1}}{\abs{\range}^{j-1}} \norm{\phi_1}^4.
\end{equation*}
Therefore, combining these bounds on \eqref{eq:expand} we have
\begin{align*}
    &\sum_{y' \in \range^{[j]}} \norm{\big(\Pi_{(f^{\otimes j})^{-1}(y')}\otimes \identity_{[m]\setminus [j]}\big)\ \phi}^4 \nonumber\\
    &\le
    \big( \frac{\max(16,3 \abs{\range})^{\ell}}{\abs{\range}^{j-1}} \norm{\phi_0}^4\big)/\abs{\range}
    +8\,\big( \frac{\max(16,3 \abs{\range})^{\ell}}{\abs{\range}^{j-1}} \norm{\phi_0}^4\big)^{1/2}
    \big( \frac{\max(16,3 \abs{\range})^{\ell-1}}{\abs{\range}^{j-1}} \norm{\phi_1}^4\big)^{1/2}/\abs{\range}\\
    &\qquad
    +3\, \frac{\max(16,3 \abs{\range})^{\ell-1}}{\abs{\range}^{j-1}} \norm{\phi_1}^4\\
    &= \frac{\max(16,3 \abs{\range})^{\ell}}{\abs{\range}^{j}} \norm{\phi_0}^4
    +8  \frac{\max(16,3 \abs{\range})^{\ell-1/2}}{\abs{\range}^j}    \norm{\phi_0}^2 \norm{\phi_1}^2
    +3\, \frac{\max(16,3 \abs{\range})^{\ell-1}}{\abs{\range}^{j-1}} \norm{\phi_1}^4\\
    &\le \frac{\max(16,3 \abs{\range})^{\ell}}{\abs{\range}^{j}} \big(\norm{\phi_0}^4+2 \norm{\phi_0}^2 \norm{\phi_1}^2+ \norm{\phi_1}^4)\\
    &=  \frac{\max(16,3 \abs{\range})^{\ell}}{\abs{\range}^{j}} \big(\norm{\phi_0}^2+\norm{\phi_1}^2 \big)^2\\
    &= \frac{\max(16,3 \abs{\range})^{\ell}}{\abs{\range}^{j}}  \norm{\phi}^4\qquad \textrm{by \eqref{eqn:phi1-phi2-triangle}},
    \end{align*}
    which is exactly what we need to show \eqref{eq:Bjl}
    with value $j$.
    The lemma follows by induction.
\end{proof}

\cref{cor:list-dec-no-loss-key-lemma} is now
an immediate consequence:

\begin{proof}[Proof of \cref{cor:list-dec-no-loss-key-lemma}]
   Observe that in the statement of
   \cref{lem:4-norm-ineq}, we have
   $B_{m,\ell}= \bigoplus_{u \in \{0,1\}^{m}, \abs{u} \leq \ell} V^u $.
   Applying \cref{lem:4-norm-ineq} for  $\phi\in \bigoplus_{u \in \{0,1\}^{m}, \abs{u} \leq \ell} V^u $
   yields
   \begin{equation*}\sum_{y' \in \range^{[m]}} \norm{\big(\Pi_{(f^{\otimes m})^{-1}(y')} \big)\ \phi}^4 \leq \frac{\max(16,3 \abs{\range})^{\ell}}{\abs{\range}^{m}} \norm{\phi}^4
   \end{equation*}
   which is what we need to show.
\end{proof}

\end{document}